\documentclass[11pt]{article}

\usepackage{fullpage}
\usepackage[utf8]{inputenc}
\usepackage[T1]{fontenc}
\usepackage[round]{natbib}
\usepackage{amsmath}
\usepackage{amsthm}
\usepackage{bm}
\usepackage{amsfonts}
\usepackage{hyperref}
\usepackage{cleveref}
\usepackage{url}
\usepackage{graphicx}
\usepackage{xcolor} 
\usepackage{tikz}
\usepackage{enumitem}
\usepackage[ruled]{algorithm2e}
\usepackage{setspace}
\usepackage{subcaption}
\usepackage{multirow}
\usepackage{autonum}
\usepackage{cancel}
\usepackage[ruled]{algorithm2e}

\def\R{\mathbb{R}}
\def\E{\mathbb{E}}

\def\H{\mathcal{H}}
\def\S{S}
\def\U{\mathcal{U}}

\def\GP{\mathcal{GP}}

\DeclareMathOperator*{\argmax}{arg\,max}

\newtheorem{theorem}{Theorem}
\newtheorem{proposition}{Proposition}
\newtheorem{lemma}{Lemma}
\newtheorem{corollary}{Corollary}

\newtheorem{definition}{Definition}
\theoremstyle{definition}

\newtheorem{remark}{Remark}
\SetKwInOut{Parameter}{Parameter}

\Crefname{appendix}{Supplement}{Supplements}
\Crefname{subappendix}{Supplement}{Supplements}
\Crefname{subsubappendix}{Supplement}{Supplements}

\title{
	Sampling as Bandits: Evaluation-Efficient Design for Black-Box Densities
}

\author{
	Takuo Matsubara$^1$ \quad Andrew Duncan$^2$ \quad Simon Cotter$^3$ \quad Konstantinos Zygalakis$^1$
}

\date{
	\small{$^1$The University of Edinburgh \quad $^2$Imperial College London} \quad $^3$The University of Manchester
}

\begin{document}

\maketitle

\begin{abstract}
	We propose bandit importance sampling (BIS), a powerful importance sampling framework tailored for settings in which evaluating the target density is computationally expensive.
	BIS facilitates accurate sampling while minimizing the required number of target-density evaluations.
	In contrast to adaptive importance sampling, which optimizes a proposal distribution, BIS directly optimizes the set of samples through a sequential selection process driven by multi-armed bandits. 
	BIS serves as a general framework that accommodates user-defined bandit strategies. 
	Theoretically, the weak convergence of the weighted samples, and thus the consistency of the Monte Carlo estimator, is established regardless of the specific strategy employed.
	In this paper, we present a practical strategy that leverages Gaussian process surrogates to guide sample selection, adapting the principles of Bayesian optimization for sampling.
	Comprehensive numerical studies demonstrate the superior performance of BIS across multimodal, heavy-tailed distributions, and real-world Bayesian inference tasks involving Markov random fields.
\end{abstract}

\section{Introduction} \label{sec:introduction}

Statistical characterization of complex phenomena often necessitates highly structured, black-box models that are computationally expensive to evaluate \citep{Chen2010}.
While Bayesian inference provides a coherent probabilistic framework for integrating observations with prior knowledge, its application to black-box models is often computationally impractical.
Such challenges are prevalent in fields, such as computational physics, biology, and geostatistics, where likelihood evaluations involve numerical solutions to dynamical equations \citep{Beaumont2002,Girolami2008,Marin2012,Warne2020} or operations on high-dimensional matrices and massive summations \citep{Rue2005,Goodreau2009,Besag1986}.

A fundamental task in Bayesian inference is to sample from a posterior density of the model parameter, which is proportional to the likelihood and the prior density.
Conventional sampling methods, such as Markov chain Monte Carlo \citep[MCMC;][]{Brooks2011}, necessitate extensive evaluations of the target density to explore the parameter space.
The cumulative computational cost easily becomes prohibitive when the target density is expensive to evaluate.
Consequently, there is a pressing need for evaluation-efficient sampling methods tailored for black-box densities, which can operate within a limited computational budget.

Importance sampling---originally introduced as a variance reduction technique in Monte Carlo estimation \citep{Hammersley1954}---offers a compelling approach in this context, approximating the target density by a weighted average of samples drawn from an arbitrary proposal density \citep{Tokdar2010}.
Since target-density evaluations are strictly limited to the preselected samples, this approach allows for precise control over the computational budget.
However, the approximation accuracy of importance sampling is heavily contingent upon the design of the samples.
Adaptive importance sampling (AIS) is arguably the most prevalent approach to this challenge, optimizing the proposal density from which samples are drawn \citep[e.g.][]{Bugallo2017,Cotter2019,Cotter2020}.
Yet, AIS remains impractical for black-box densities, as the optimization of the proposal density reintroduces the need for intensive target-density evaluations.

In machine learning, Bayesian optimization \citep[BO;][]{Mockus1989} has emerged as a prominent framework for optimizing black-box objective functions that are expensive to evaluate \citep{Shahriari2016}.
The core principle involves approximating the underlying objective function with a computationally efficient surrogate built upon the history of observed evaluations.
Gaussian processes \citep[GPs;][]{Rasmussen2005} are a standard and versatile choice for surrogate modeling \citep{Gramacy2020}. 
GPs quantify predictive uncertainty alongside predictions, accounting for the epistemic uncertainty inherent in estimating the true function from finite data.
BO iteratively selects the query point for evaluation by maximizing an acquisition function constructed from the GP surrogate model.
This framework, sometimes referred to as GP bandits, can be cast as a bandit problem with an infinite number of arms \citep{Bull2011}.
Accordingly, this framework has been extensively leveraged for experimental design and active learning tasks \citep{Srinivas2010}.

\paragraph*{Contribution}
This work proposes bandit importance sampling (BIS), a novel importance sampling framework driven by multi-armed bandits.
BIS directly optimizes the samples, sequentially selecting them from a candidate pool via a bandit strategy.
This approach preserves the fundamental advantage of importance sampling, limiting the number of target-density evaluations strictly to the pre-specified sample size. 
BIS is a general framework compatible with arbitrary bandit strategies. 
Theoretically, we establish the weak convergence of the weighted samples to the target density independently of the employed strategy.
We present a practical implementation of the bandit strategy utilizing GP surrogate models, translating the principles of BO into sampling tasks.
We demonstrate the efficiency of BIS across diverse experiments, including sampling from the marginal posterior of a Markov random field (MRF) model applied to large-scale precipitation data in the United States.

\paragraph*{Scope}
In this work, we restrict our attention to computationally expensive densities on low-dimensional spaces.
The scalability and accuracy limitations of GPs in high-dimensional spaces remain open challenges.
Sampling from high-dimensional densities poses its own set of challenges.
The compound challenge of computational intensity and high dimensionality renders the problem intractable.
Therefore, we isolate the challenge of computational intensity.
Crucially, this setting retains substantial practical relevance.
For example, in the Lorenz weather forecast model \citep{Wilks2005}, the posterior evaluation involves 40-dimensional stochastic dynamics, while the parameter space is two-dimensional.
In hierarchical MRF models for massive spatial data \citep{Rue2005}, the high-dimensional latent field can often be marginalized out analytically, reducing the inference task to a few hyperparameters, yet the marginal posterior remains computationally expensive.

\paragraph*{Structure}
The rest of the paper is structured as follows.
\Cref{sec:background} recaps importance sampling and GP regression concisely.
\Cref{sec:methodology} introduces the general framework of BIS.
\Cref{sec:theory} establishes a theoretical analysis of the convergence rate of BIS.
\Cref{sec:BO} develops a practical implementation of the bandit strategy for BIS, which leverages GP surrogate modeling.
\Cref{sec:experiment} examines the empirical performance of BIS, with applications to Bayesian inference of computationally-expensive models.
Relevant work on GP-based experimental designs for Bayesian computation is recapped in \Cref{sec:related}.
Finally, \Cref{sec:conclusion} concludes the paper with a discussion of future directions.

\section{Background} \label{sec:background}

This section provides a concise overview of importance sampling and GP regression.
Let $\Theta \subset \R^d$ for some dimension $d$.
Any probability distribution on $\Theta$ considered in this work is identified with the density with respect to the Lebesgue measure if it exists.

\subsection{Importance Sampling}

Importance sampling is a fundamental technique, particularly effective when direct sampling from a target density $p(\theta) = q(\theta) / Z$ on $\Theta$ is impractical.
The objective is to construct a set of $M$ weighted samples, $\{ \theta_n, w_n  \}_{n=1}^{M}$, which approximates the target density $p(\theta)$.
The samples $\{ \theta_n  \}_{n=1}^{M}$ are drawn from a proposal density $u(\theta)$ instead of $p(\theta)$.
If the density $p(\theta)$ is known exactly, the importance weights $\{ w_n  \}_{n=1}^{M}$ are computed as the ratio $p(\theta_n) / u(\theta_n)$.
If the density $p(\theta)$ is known only up to the normalization constant $Z$, the \emph{self-normalized} weights $\{ w_n  \}_{n=1}^{M}$ are employed, by which the expectation of a test function $f(\theta)$ with respect to $p(\theta)$ is estimated as
\begin{align}
	\E_{\theta \sim p}[f(\theta)] \approx \sum_{n=1}^{M} w_n f(\theta_n) \quad \text{where} \quad w_n = \frac{ q(\theta_n) / u(\theta_n) }{ \sum_{n=1}^{M} q(\theta_n) / u(\theta_n) } . \label{eq:mc_estimate}
\end{align}
The weights depend solely on the unnormalized target $q(\theta)$, as the self-normalization cancels the constant $Z$.
Consistency of this Monte Carlo estimator is well established.

For importance sampling to be efficient, a substantial portion of the samples $\{ \theta_n  \}_{n=1}^{M}$ must be located within the high-probability regions of the target density $p(\theta)$.
Otherwise, the vast majority of samples will be assigned negligible weights, contributing minimally to the representation of the target density $p(\theta)$.
Panel (a) in \Cref{fig:fig_21_1} illustrates a scenario where only approximately 10\% of the samples carry significant weights, while the remaining 90\% are negligible.
AIS enhances the efficiency of importance sampling by optimizing the proposal density $u(\theta)$ to draw samples.
In AIS, a parametric family of proposal densities is typically employed, and samples are drawn from the family after or during the optimization process.
However, this approach is impractical for our setting, as the successful optimization of the proposal density $u(\theta)$ necessitates extensive evaluations of the target density $p(\theta)$.
Furthermore, since the samples are random and independent, they are prone to redundancy and often fail to be the maximally informative locations for evaluating $p(\theta)$.

Suppose that the evaluation of $p(\theta)$ is so computationally expensive that the unnormalized target $q(\theta)$ can be evaluated at most $N$ times.
A natural objective given this constraint is to identify $N$ samples that yield the most accurate approximation of the density $p(\theta)$ via importance sampling.
For instance, in Panel (a) in \Cref{fig:fig_21_1}, where 200 samples are initially drawn, selecting the $N$ samples with the largest weights yields the best accuracy given a budget of $N$ density evaluations.
To this end, we frame the problem as a sequential decision-making task to select the $N$ most informative points from a pool of $M$ well-dispersed candidates.
See Panel (b) in \Cref{fig:fig_21_1} for illustration.

\begin{figure}[t]
	\subcaptionbox{Standard Importance Sampling}{\includegraphics[width=0.475\textwidth]{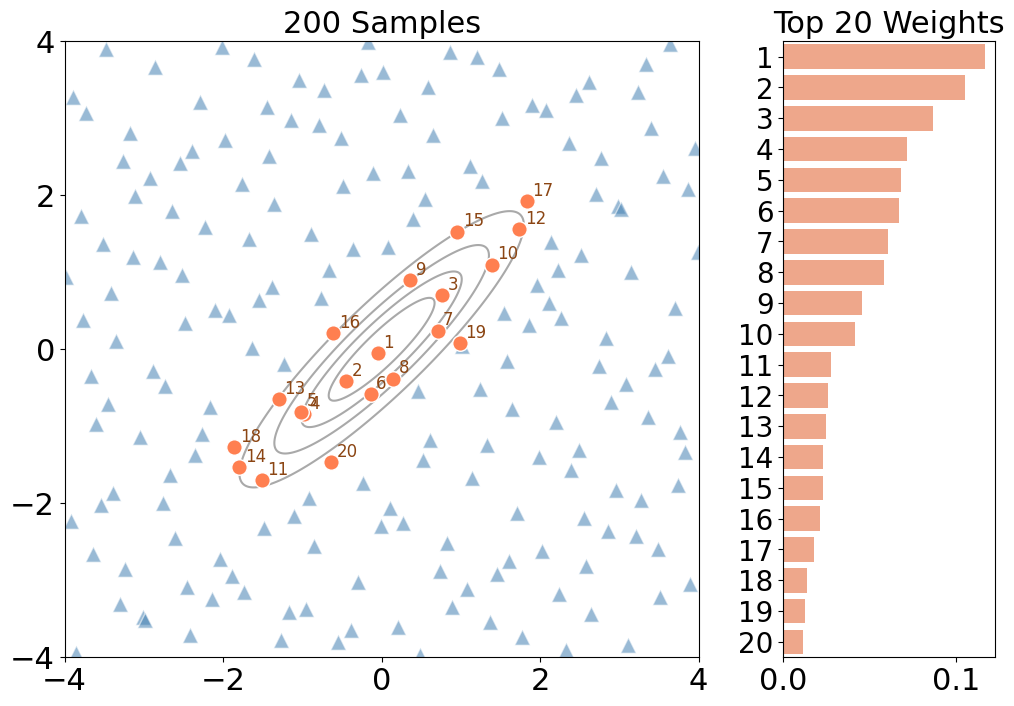}}
	\hfill
	\subcaptionbox{Bandit Importance Sampling}{\includegraphics[width=0.475\textwidth]{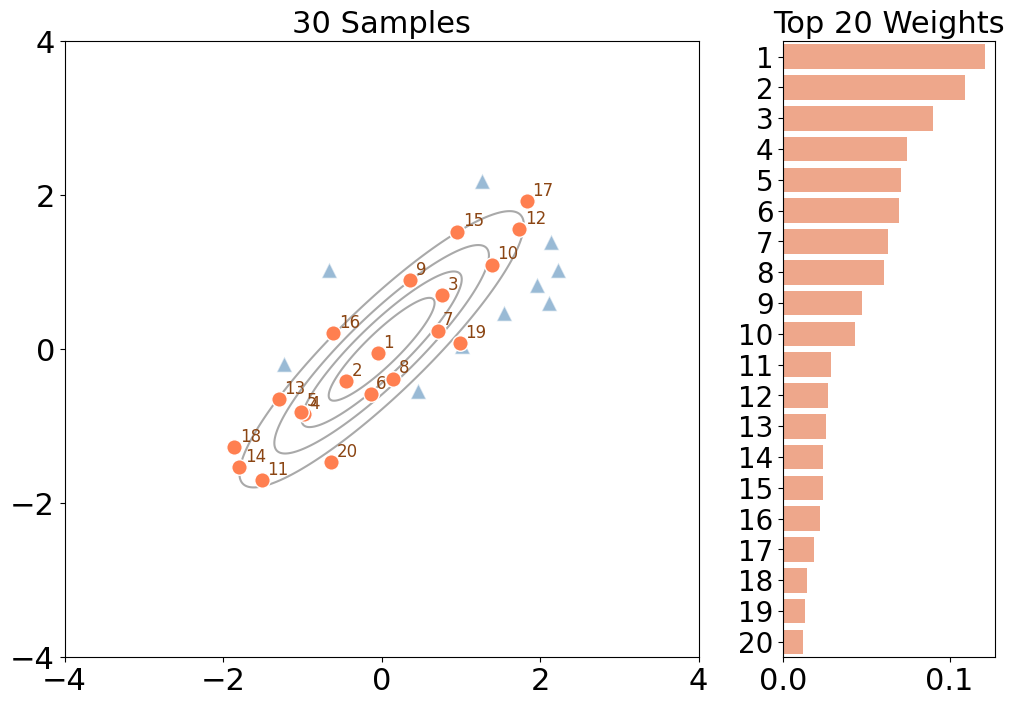}}
	\caption{Illustration of the sampling efficiency of BIS on a normal density (solid contour). Panel (a): Standard importance sampling with 200 quasi-uniform points; circles indicate the 20 largest weights (top 10\%), and triangles indicate the rest. Panel (b): Selection by BIS with a budget of $N = 30$ evaluations; the method identified all the high-weight points among the 200 candidates while evaluating the target only 30 times, achieving a 85\% reduction in computational cost.} \label{fig:fig_21_1}
\end{figure}

\subsection{GP Regression}

A GP, denoted by $\GP(m, k)$, is a stochastic process fully specified by its mean function $m: \Theta \to \R$ and covariance kernel $k: \Theta \times \Theta \to \R$.
It can be interpreted as a random variable $f \sim \GP(m, k)$ taking values in a function space on $\Theta$.
A defining characteristic of GPs is the moment condition:
\begin{align}
	m(\theta) = \E_{f \sim \GP(m, k)}[ f(\theta) ] \quad \text{and} \quad k(\theta, \theta') = \E_{f \sim \GP(m, k)}[ ( f(\theta) - m(\theta) ) ( f(\theta) - m(\theta') ) ] .
\end{align}
Another characteristic of GPs is that the joint distribution of the function values $( f(\theta_1), \dots, f(\theta_n) )$ at any finite set of inputs $( \theta_1, \dots, \theta_n )$ follows a multivariate Gaussian distribution $\mathcal{N}( \bm{m}_n, \bm{K}_n )$.
Here, the mean $\bm{m}_n \in \R^n$ denotes the vector with $i$-th entry $m(\theta_i)$ and the covariance $\bm{K}_n  \in \R^{n \times n}$ denotes the matrix with $(i,j)$-th entry $k(\theta_i, \theta_j)$.
Widely-used covariance kernels include Gaussian and Mat\'{e}rn kernels.
See, e.g.,~\cite{Rasmussen2005} for a detailed introduction.

GPs offer a flexible framework for nonparametric Bayesian regression.
Suppose we observe a set of inputs $( \theta_1, \dots, \theta_n )$ and their corresponding outputs $( y_1, \dots, y_n )$.
Conditioning the GP prior $\GP(m, k)$ on the observed data yields a posterior distribution that is also a GP, denoted by $\GP(m_n, k_n)$, with the updated mean and covariance functions given by
\begin{align}
	m_n(\theta) & = m(\theta) + \bm{k}_n(\theta)^{\mathrm{T}} \bm{K}_n^{-1} ( \bm{y}_n - \bm{m}_n ) , \\
	k_n(\theta, \theta') & = k(\theta, \theta') - \bm{k}_n(\theta)^{\mathrm{T}} \bm{K}_n^{-1} \bm{k}_n(\theta') .
\end{align}
Here, $\bm{k}_n(\theta) = ( k(\theta, \theta_1), \cdots, k(\theta, \theta_n) ) \in \R^n$ is the vector of kernel evaluations between the test point $\theta$ and the observed inputs, and $\bm{y}_n = (y_1, \dots, y_n) \in \R^n$ is the vector of the observed outputs.
In practice, it is common to assume that observations are subject to i.i.d.~Gaussian noise $\epsilon \sim \mathcal{N}(0, \sigma_*^2)$ with the scale $\sigma_*$.
In this case, $\bm{K}_n$ is replaced with $\bm{K}_n + \sigma_*^2 \bm{I}_{n \times n}$ in the aforementioned definition, where $\bm{I}_{n \times n}$ is the $n \times n$ identity matrix.

\section{General Framework} \label{sec:methodology}

We now introduce BIS, a method that frames the adaptive design of samples as a bandit problem.
This section presents a general framework of BIS, which comprises two key components: (i) a pool of candidates from which samples are selected, and (ii) a criterion for selecting a sample at each iteration.
Specifying these two components (i) and (ii) defines a concrete instance of BIS.
The construction of an effective candidate pool for black-box densities is detailed in  \Cref{sec:proposal}.
The construction of an effective point-selection criterion, which is driven by GP bandits in this work, will be elaborated in the subsequent section.
The standing setup for BIS is stated below.

\paragraph*{Setup:}
Suppose $\Theta \subset \R^d$ is a compact hyperrectangular domain with no loss of generality, since reparametrization can be performed otherwise.
Let $p(\theta) = q(\theta) / Z$ be a target density of interest on $\Theta$, where $q(\theta)$ denotes the proportional term and $Z$ denotes the normalization constant.
Assume that $p(\theta)$ is positive over $\Theta$.
A practical scenario in mind is that evaluation of the unnormalized density $q(\theta)$ is possible at each $\theta$ but computationally expensive for one's computational capacity.

\subsection{Bandit Importance Sampling} \label{sec:bis}

BIS aims to construct a set of $N$ weighted samples $\{ \theta^*_n, w^*_n \}_{n=1}^{N}$ to approximate the target density $p(\theta)$, subject to a budget of $N$ evaluations of the unnormalized target $q(\theta)$.
Instead of drawing samples randomly from a proposal density, BIS directly optimizes the samples, sequentially selecting them from a candidate pool based on a given selection criterion. 
Distinct from conventional bandits, once a candidate is selected, it is never revisited in subsequent iterations (see \Cref{rm:discrete_importance}).

We detail the explicit procedure of BIS.
Let $u(\theta)$ be a proposal density on $\Theta$. 
A sequence $\U = \{ \theta_i \}_{i=1}^{\infty}$, whose empirical distribution converges weakly to $u(\theta)$, is called the \emph{proposal sequence} in this work.
Let $\S_n$ denote the candidate pool at iteration $n$.
At the initial iteration $n = 1$, the candidate pool $\S_1$ is initialized with the first $M$ points of the proposal sequence $\U$: $\S_1 = \{ \theta_i \}_{i=1}^{M}$.
After each iteration, the candidate pool $\S_n$ is updated by replacing the selected point with the next available point from the proposal sequence $\U$.
This mechanism guarantees that no point in $\U$ is ever selected more than once.
Finally, let $U_n(\theta)$ be the selection criterion defined at iteration $n$, which governs the choice of the sample.
We repeat the following procedure from $n = 1$:
\begin{enumerate}
	\item Select a point $\theta_n^*$ from the candidate pool $\S_n$ that maximizes the criterion $U_n(\theta)$;
	\item Evaluate the unnormalized target $q(\theta_n^*)$ to compute the corresponding weight $w_n^*$;
	\item Update the candidate pool for the next iteration by $\S_{n+1} = ( \S_{n} \setminus \{ \theta_n^* \} ) \cup \{ \theta_{M+n} \}$.
\end{enumerate}
Following the final iteration $N$, we assign the self-normalized weight $w_n^*$ in \eqref{eq:mc_estimate} to each sample $\theta_n^*$.
By construction, the size of the candidate pool $| \S_n |$ remains constant at $| \S_n | = M$ throughout the entire procedure.
Algorithm \ref{alg:tmp} provides a formal summary of the BIS procedure.

\begin{algorithm}[h]
	\caption{Bandit Importance Sampling} \label{alg:tmp}
	\KwIn{sample size $N$, candidate-pool size $M$, proposal sequence $\{ \theta_i \}_{i=1}^{M+N}$ convergent to proposal density $u$, point-selection criterion $U_n$ for each iteration $n$}
	\KwOut{weighted samples $\{ \theta_n^*, w_n^* \}_{n=1}^{N}$ to approximate the target density $p$}
	$\S_1 \gets \{ \theta_i \}_{i=1}^{M}$ \hfill $\triangleright$ initialize selection pool\\
	\For{$n \gets 1, \cdots, N$}{
		$\theta_n^* \gets \argmax_{ \theta \in \S_n } U_n(\theta)$ \hfill $\triangleright$ optimize point-selection criterion\\
		$w_n^* \gets \frac{ q(\theta_n^*) }{ u(\theta_n^*) }$ \hfill $\triangleright$ evaluate density ratio $\frac{ q(\cdot) }{ u(\cdot) }$\\
		$\S_{n+1} \gets \big( \S_{n} \setminus \{ \theta_n^* \} \big) \cup \{ \theta_{M+n} \}$ \hfill $\triangleright$ update selection pool\\
	}
	$(w_1^*, \cdots, w_N^*) \gets (w_1^*, \cdots, w_N^*) / \sum_{n=1}^{N} w_n^*$ \hfill $\triangleright$ normalize weights\\
\end{algorithm}

\begin{remark}[Importance of Discrete optimization with No Revisit] \label{rm:discrete_importance}
	BIS selects points by maximizing the criterion $U_n$ over the discrete set $S_n$, subject to the constraint that previously selected points are never revisited.
	A possible alternative, analogous to BO, would be to optimize the criterion $U_n$ over the continuous set $\Theta$.
	However, as demonstrated in \Cref{apx:discrete}, such a continuous approach can lead to an overconcentration of the samples around the modes of the target density.
	This clustering behavior results in redundant evaluations and hinders the convergence of the importance-sampling approximation.
	In contrast, the use of the discrete set $S_n$ in BIS acts as a regularizer to prevent the overconcentration of the samples.
	Combined with the no-revisit mechanism, this ensures the convergence of the importance-sampling approximation, c.f., \Cref{sec:theory}.
\end{remark}

\subsection{Proposal Density and Proposal Sequence} \label{sec:proposal}

We discuss the choice of the proposal density $u$ and the corresponding proposal sequence $\U = \{ \theta_i \}_{i=1}^{\infty}$ in BIS.
In standard importance sampling, an optimal proposal density exists in terms with minimizing the variance of the Monte Carlo estimator.
However, this proposal density depends on the target density itself, which is black-box and unavailable in our setup.
Since key properties of the black-box target density---such as its modes or tail behavior---are unknown a priori, there is even little basis for constructing a tailored proposal density.
Therefore, we adopt the uniform density over the domain $\Theta$ as a natural and robust default choice for the proposal density $u$.

Crucially, the proposal sequence $\U$ need not consist of i.i.d.~random samples, since unbiasedness is no longer a primary interest for BIS.
Any dependent random or deterministic sequence can be employed, provided it converges weakly to the proposal density $u$.
Space-filling sequences from Quasi-Monte Carlo \citep[QMC;][]{Dick2010} are ideal for the proposal sequence, exhibiting rapid convergence to the uniform density $u$ \citep{Morokoff1994}.

A celebrated result in QMC theory establishes that numerical integration using $N$ points from a space-filling sequence achieves a convergence rate of $\mathcal{O}( ( \log N )^{d-1} N^{-1} )$.
This rate is asymptotically faster than the rate of standard Monte Carlo integration $\mathcal{O}( N^{-1/2} )$ for large sample sizes $N$ \citep{Morokoff1994}.
Nevertheless, importance sampling with space-filling sequences can be inefficient for small-to-medium sample sizes $N$.
This is because these sequences are designed for uniform exploration, failing to adapt to the high-probability regions of a specific target.
utilizing a space-filling sequence within the BIS framework can be viewed as a sequential permutation of the sequence, such that exploration begins in the high-probability regions of the target.
This ensures robust approximation accuracy across both small and large sample sizes $N$.

Our default recommendation of the proposal sequence is the Halton sequence \citep{Halton1964}, which is arguably the most widely adopted space-filling sequence in practice.
While originally defined for the unit cube, the Halton sequence is readily amenable to arbitrary hyperrectangles.
For clarity, we define the Halton sequence scaled to the hyperrectangular domain $\Theta = \prod_{i=1}^d [a_{(i)}, b_{(i)}]$, where $a_{(i)}$ and $b_{(i)}$ denote the lower and upper boundaries for each coordinate $i$, respectively.

\begin{definition}[Scaled Halton Sequence] \label{def:halton}
	Let $\{ \eta_i \}_{i=1}^{\infty}$ be the standard Halton sequence on the unit cube $[0, 1]^d$.
	The scaled Halton sequence $\{ \theta_i \}_{i=1}^{\infty}$ on the domain $\Theta$ is defined by $\theta_i := a + B \eta_i$, where $a := ( a_{(1)}, \cdots, a_{(d)} ) \in \R^d$ and  $B$ is a $d \times d$ diagonal matrix with $(i,i)$-th entry $b_{(i)} - a_{(i)}$.
\end{definition}

Having established the general framework of BIS, we now turn to the analysis of its theoretical guarantees.
Subsequently, we detail a practical instantiation of the point-selection criterion.

\section{Convergence Guarantee} \label{sec:theory}

We present a rigorous analysis of the convergence properties of the BIS framework.
Our aim is to highlight that BIS does not compromise the convergence rate of standard importance sampling, \emph{regardless of} the choice of the point-selection criterion. 
This resilience stems from the construction, which ensures that no candidate point is selected more than once (c.f.~\Cref{rm:discrete_importance}).
This intrinsic convergence guarantee grants users the flexibility to design bespoke point-selection criteria, safeguarding against the risk of complex designs introducing unquantifiable sampling biases.

First, \Cref{sec:convergence_metric} introduces the Maximum Mean Discrepancy \citep[MMD;][]{Smola2007,Gretton2012}, which serves as a metric to assess the approximation error between the weighted samples and the target density.
\Cref{sec:convergence_uniform} then establishes the convergence rate of Algorithm \ref{alg:tmp} under the default setting of the uniform proposal density.
Finally, \Cref{sec:convergence_general} provides the convergence rate of Algorithm \ref{alg:tmp} in a general setting where the proposal density is arbitrary.
Proofs of all the theoretical results are provided in \Cref{apx:proof}.

\subsection{Metric of Convergence} \label{sec:convergence_metric}

Consider a metric on the space of probability distributions on $\Theta$ that metrizes the weak convergence; that is, a probability distribution converges weakly to a target if and only if their distance under the metric converges to zero.
A standard strategy for establishing the convergence of the weighted samples to the target density is to demonstrate the decay of the metric between them.
In this context, the MMD is a computationally and theoretically well-suited choice of the metric.

The MMD is a metric defined via a kernel function $\kappa: \Theta \times \Theta \to \R$.
Recall that every kernel implies the existence of the existence of a uniquely associated Hilbert space of functions on $\Theta$, known as the reproducing kernel Hilbert space (RKHS).
See, e.g., \cite{Paulsen2016} for an introduction to RKHS.
Let $\mathcal{H}_\kappa$ be the RKHS associated with $\kappa$, equipped with the norm $\| \cdot \|_{\mathcal{H}_\kappa}$.
The MMD is defined as the worst-case error between two expectations with respect to $p$ and $q$:
\begin{align}
	d(p, q) := \sup_{\| f \|_{\mathcal{H}_k} \le 1} \Big| \E_{\theta \sim p}[ f(\theta) ] - \E_{\theta \sim q}[ f(\theta) ] \Big| . \label{eq:mmd}
\end{align}
The MMD satisfies the metric axioms and metrizes weak convergence of probability distributions under mild regularity conditions on the kernel \citep{Simon-Gabriel2023}.
For our theoretical analysis, the specific choice of the kernel $\kappa$ is not critical.
Provided the condition outlined below, our convergence bounds are independent of the kernel choice.
Consequently, the existence of any kernel satisfying this condition suffices to establish the weak convergence.

A kernel $\kappa$ is said to be integrally strictly positive definite if $\int_{\Theta \times \Theta} \kappa(\theta, \theta') d \mu(\theta) d \mu(\theta') > 0$ for all finite non-zero signed Borel measures $\mu$ on $\Theta$.
This is analogous to strictly positive definiteness of kernels.
For example, Gaussian kernels satisfy this property \citep{Sriperumbudur2010}.
To streamline our derivation, we impose the following condition on the kernel $\kappa$ used for MMD.
Without loss of generality, we assume that both the kernel and the derivatives are bounded by $1$.
This is justified because any bounded kernel can be rescaled by a constant factor to satisfy this condition.
The Gaussian kernels is a standard example satisfying this condition.

\vspace{5pt}
\noindent
\textbf{Kernel Condition:}
The kernel $\kappa: \Theta \times \Theta \to \R$ is integrally strictly positive definite.
The kernel $\kappa$ admits a factorization $\kappa(\theta, \theta') = \prod_{i=1}^{d} \kappa_i(\theta_{(i)}, \theta_{(i)}')$ for some univariate kernels $\kappa_1, \dots, \kappa_d$, where $\theta_{(i)}$ denotes the $i$-th coordinate of $\theta$.
For each $i = 1, \dots, d$, the kernel $\kappa_i$ and the mixed partial derivative $(\partial / \partial \theta_{(i)}) (\partial / \partial \theta'_{(i)}) \kappa_i(\theta_{(i)}, \theta'_{(i)})$, assumed to exist, are both uniformly bounded by $1$.

\subsection{Deterministic Rate in Uniform Case} \label{sec:convergence_uniform}

Our first analysis focuses on the case where the proposal density $u$ is uniform, as discussed in \Cref{sec:proposal}.
Notably, the uniformity of $u$ yields a deterministic convergence rate, irrespective of whether the proposal sequence is stochastic or deterministic.
A central quantity in this analysis is the \emph{star discrepancy} $\operatorname{D}^*$ of a set of $K$ points $\{ \theta_n \}_{n=1}^{K}$ in $\Theta$, defined as
\begin{align}
	\operatorname{D}^*\left( \{ \theta_n \}_{n=1}^{K} \right) := \sup_{ B \in \mathcal{B} } \left| \frac{\# \{ \theta_n \in B \} }{K} - \frac{\operatorname{Vol}(B)}{\operatorname{Vol}(\Theta)} \right| ,
\end{align}
where $\mathcal{B}$ denotes the set of all sub-boxes in $\Theta$ anchored at the lower corner of $\Theta$.
Here, $\# \{ \theta_n \in B \}$ denotes the number of points contained in the box $B$, and $\operatorname{Vol}(\cdot)$ denotes the volume of a given region.
Although the star discrepancy is typically defined for the unit cube, we adopt the above definition extended to the domain $\Theta$.

Let $\partial_i f(\theta)$ denote the partial derivative of a function $f(\theta)$ with respect to the i-th coordinate.
Let $\partial_{1:d} f(\theta)$ denote the first mixed partial derivative of $f(\theta)$, that is, $\partial_{1:d} f(\theta) = \partial_1 \cdots \partial_{d} f(\theta)$.
Hereafter, let $\delta_N$ denote the empirical distribution associated with the $N$ weighted samples $\{ \theta_n^*, w_n^* \}_{n=1}^{N}$ obtained from BIS.
We are now in a position to state our main result.
The first term in the bound below is derived by invoking the Koksma-Hlawka inequality, a classical result in QMC theory governing numerical integration error.
The second term arises from the point-selection procedure of BIS, which decays significantly faster than the canonical Monte Carlo integration rate $\mathcal{O}(N^{-1/2})$.

\begin{theorem} \label{thm:convergence}
	Assume that the target density $p$ admits a uniformly continuous mixed partial derivative $\partial_{1:d} p(\theta)$, and that the proposal density $u$ is uniform.
	For any proposal sequence $\{ \theta_n \}_{n=1}^{M+N}$ and selection criterion $U_n$, the weighted samples $\{ \theta_n^*, w_n^* \}_{n=1}^{N}$ of Algorithm \ref{alg:tmp} satisfy
	\begin{align}
		d(p, \delta_N) \le C_1 \operatorname{D}^*\left( \{ \theta_n \}_{n=1}^{N + M} \right) + C_2 \frac{M}{N + M} ,
	\end{align}
	where $C_1$ and $C_2$ are constants dependent sorely on $p$ and $u$.
\end{theorem}

\Cref{thm:convergence} immediately yields an explicit convergence rate for BIS.
To derive this, it suffices to employ a proposal sequence $\{ \theta_n \}_{n=1}^{M+N}$ with a known upper bound on the star discrepancy.
The Halton sequence, for example, satisfies a well-established bound \citep[see e.g.][]{Niederreiter1992}.
Consequently, we obtain the following rate under the Halton sequence in \Cref{def:halton}.

\begin{corollary} \label{col:convergence_Halton}
	Under the assumptions of \Cref{thm:convergence}, suppose that the proposal sequence $\{ \theta_n \}_{n=1}^{M+N}$ is the scaled Halton sequence.
	The weighted samples $\{ \theta_n^*, w_n^* \}_{n=1}^{N}$ satisfy
	\begin{align}
		d(p, \delta_N) \le C_1 \frac{\log ( N + M )^d}{N + M} + C_2 \frac{M}{N + M} ,
	\end{align}
	where $C_1$ and $C_2$ are constants dependent solely on $p$ and $u$.
\end{corollary}

We note that i.i.d.~random samples from the uniform density $u$ also satisfy a bound of the star discrepancy that holds almost surely \citep{Morokoff1994}.
However, we prioritize the deterministic sequence in this work due to its faster convergence rate.
Nevertheless, it is worth highlighting that i.i.d.~random samples still ensure the almost sure convergence.
The established convergence rate implies the consistency of the Monte Carlo estimator by BIS.

\begin{corollary} \label{col:weak_convergence}
	Under the assumptions in \Cref{col:convergence_Halton}, the empirical distribution $\delta_N$ of the weighted samples $\{ \theta_n^*, w_n^* \}_{n=1}^{N}$ converges weakly to the target density $p$ as $N \to \infty$.
	Consequently, the resulting Monte Carlo estimator of $\E_{\theta \sim p}[ f(\theta) ]$ is consistent for any continuous integrand $f$.
\end{corollary}

\subsection{Concentration Rate in General Case} \label{sec:convergence_general}

We extend our analysis to a general setting where the proposal density $u$ is arbitrary, and the proposal sequence $\{ \theta_i \}_{i=1}^{\infty}$ consists of i.i.d.~random samples drawn from $u$.
Given the stochastic nature of the proposal sequence $\{ \theta_i \}_{i=1}^{\infty}$, we establish a concentration bound for the MMD.

\begin{theorem} \label{thm:concentration}
	Assume that the proposal density $u$ is uniformly bounded on $\Theta$.
	Let the proposal sequence $\{ \theta_n \}_{n=1}^{M+N}$ be i.i.d.~samples from $u$. 
	For any selection criterion $U_n$, the weighted samples $\{ \theta_n^*, w_n^* \}_{n=1}^{N}$ of Algorithm \ref{alg:tmp} satisfy that, for any $0 < \epsilon < 1$,
	\begin{align}
		\mathbb{P}\left( d(p, \delta_N) \le \frac{C_1 + C_2 \sqrt{\log(2 / \epsilon)}}{\sqrt{N + M}} + C_3 \frac{M}{N + M} \right) \ge 1 - \epsilon 
	\end{align}
	where $\mathbb{P}$ denotes the probability with respect to the random samples $\theta_1, \dots, \theta_{M+N} \overset{i.i.d.}{\sim} u$, and $C_1$, $C_2$, and $C_3$ denote constants dependent sorely on $p$ and $u$.
\end{theorem}

This theorem leads to an immediate corollary regarding the almost-sure weak convergence of the weighted samples, and hence the consistency of the Monte Carlo estimator by BIS.

\begin{corollary} \label{col:weak_convergence_as}
	Under the assumptions in \Cref{thm:concentration}, the empirical distribution $\delta_N$ of the weighted samples $\{ \theta_n^*, w_n^* \}_{n=1}^{N}$ converges weakly to the target density $p$ almost surely as $N \to \infty$.
	Consequently, the resulting Monte Carlo estimator of $\E_{\theta \sim p}[ f(\theta) ]$ is consistent for any continuous integrand $f$.
\end{corollary}

Although this general setting accommodates arbitrary proposal densities, the resulting convergence rate is stochastic.
Furthermore, the deterministic rate established in \Cref{col:convergence_Halton} is notably faster in lower dimensions $d$.
This observation underscores the advantage of using the Halton sequence in \Cref{def:halton} as the proposal sequence.

\section{Adapting Bayesian optimization for Sampling} \label{sec:BO}

We now present a versatile, practical instance of the point-selection criteria, which renders BIS effective across a range of black-box densities in low dimensions.
This criterion leverages GP-surrogate modeling, analogous to BO.
In machine learning, BO has established itself as a powerful approach to the optimization of black-box functions.
BIS enables the application of this powerful bandit approach to sampling.

As mentioned in \Cref{sec:theory}, BIS allows for great flexibility in designing bespoke point-selection criteria.
The development of further variants beyond BO is reserved for future work, as the rigorous derivation of each new criterion necessitates a dedicated study, and providing an exhaustive catalog of point-selection criteria lies beyond the scope of this paper.

\subsection{Point-Selection Criterion} \label{sec:af}

The point-selection criterion $U_n$ focused in this paper leverages a GP surrogate model for the unnormalized density $q(\theta)$, the proportional term of the target density $p(\theta)$.
The unnormalized density $q(\theta)$ is strictly positive, while GP functions take values in the real line $\mathbb{R}$.
To fit a GP to the density, a link function to map the density to the real line is required.
The log-transformation, $\log q(\theta)$, is a standard approach.
Here, we generalize this by employing a convex non-negative function $\phi: \mathbb{R} \to [0, \infty)$ whose inverse $\phi^{-1}$ is well-defined on $(0, \infty)$.

\begin{figure}[t]
	\includegraphics[width=0.325\textwidth]{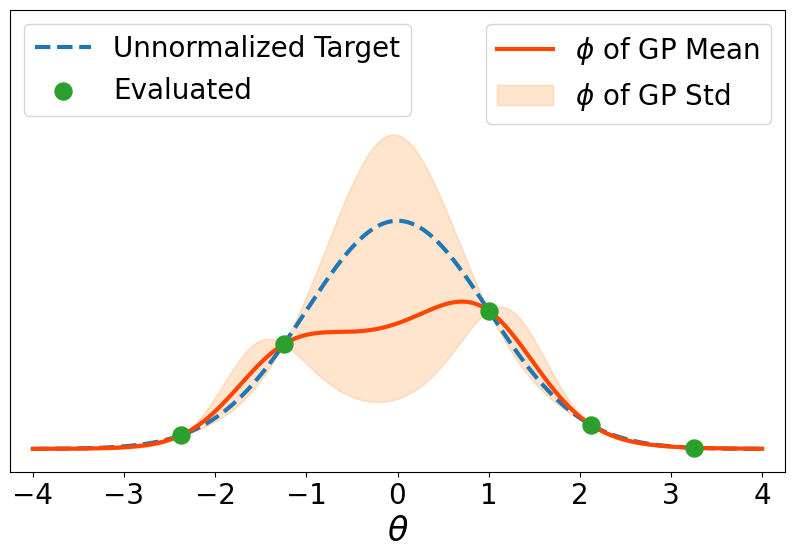}
	\hfill
	\includegraphics[width=0.325\textwidth]{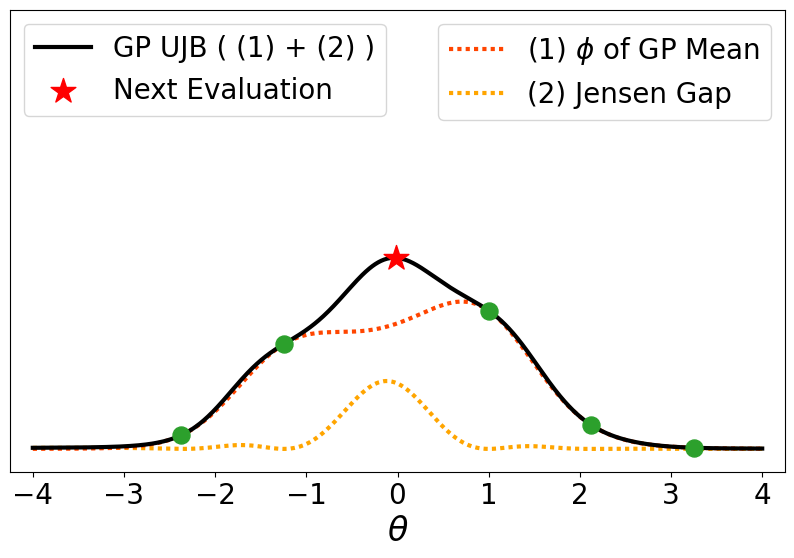}
	\hfill
	\includegraphics[width=0.325\textwidth]{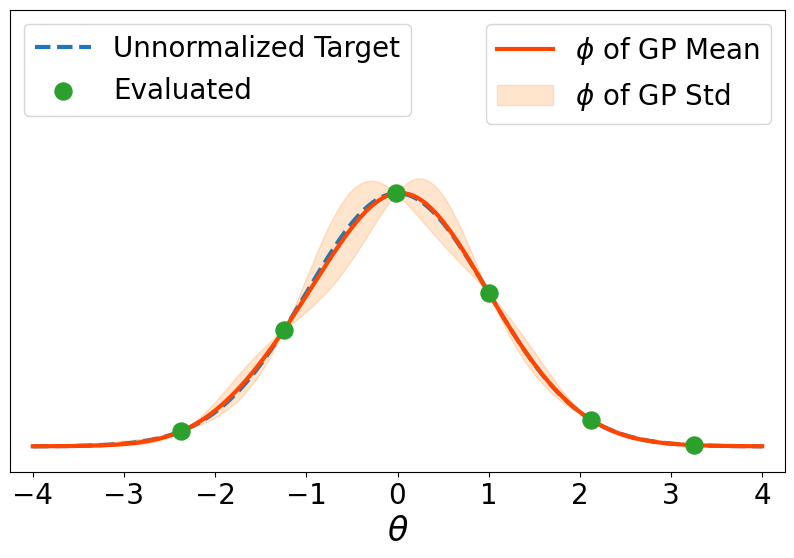}
	\hfill
	\caption{Illustration of GP-UJB using the unnormalized standard normal target $q$ and the exponential map $\phi(\cdot) = \exp(\cdot)$. Left:~The transformed GP posterior $\phi( f(\theta) )$ conditional on five evaluations of $q$. Center:~The resulting GP-UJB, the selected location for the next evaluation, and the decomposition into the exploitation and exploration terms. The exploitation term drives selection towards the mode of the current estimate, while the exploration term drives selection towards where $\phi( f(\theta) )$ is most uncertain. Right:~The transformed GP posterior after the new evaluation.} \label{fig:fig_51_1}
\end{figure}

Under this formulation, we model the transformed unnormalized density $g(\theta) = \phi^{-1}( q(\theta) )$ via a GP, where the fitted GP can be transformed back to the non-negative range by $\phi(\cdot)$.
For example, the log-transformation $\phi^{-1}(\cdot) = \log(\cdot)$ corresponds to the choice $\phi(\cdot) = \exp(\cdot)$. 
Now, we introduce our criterion, GP-UJB, equipped with appealing properties for multi-armed bandits.

\begin{definition}[GP-UJB] \label{def:gp_ujb}
Given a GP posterior $\mathcal{GP}(m_n, k_n)$ conditional on a set of $n$ evaluations $\{ \theta_i, \phi^{-1}( q(\theta_i) ) \}_{i=1}^{n}$, GP-UJB is a function $U_n: \Theta \to \R$ defined pointwise by 
\begin{align}
	U_n(\theta) := \E_{f \sim \mathcal{GP}(m_n, k_n)} \left[ \phi\big( f(\theta) \big) \right] , \label{eq:gp_ujb}
\end{align}
where $f(\theta)$ denotes the GP sample path $f \sim \mathcal{GP}(m_n, k_n)$ evaluated at $\theta$.
\end{definition}

Driving both exploitation of past information and exploration of unexplored domains is a pivotal factor in most bandit problems.
Notably, GP-UJB admits an interpretable decomposition into two terms that drive, respectively, exploitation and exploration.

\begin{remark}[Exploitation-Exploration Decomposition]
GP-UJB can be decomposed as
\begin{align}
	U_n(\theta) = \underbrace{ \phi\big( m_n(\theta) \big) }_{\text{exploitation}} + \underbrace{ \E_{f \sim \mathcal{GP}(m_n, k_n)} \left[ \phi\big( f(\theta) \big) \right] - \phi\big( m_n(\theta) \big) }_{\text{exploration}} .
\end{align}
The first term is the best predictor based on current information, i.e., the least-error plug-in approximation of $q(\theta)$ attainable by the GP.
The second term is known as the \emph{Jensen's gap} and non-negative due to Jensen's inequality.
It measures the amount of uncertainty in $\phi(f(\theta))$ at each location $\theta$.
This term is zero at the training points $\{ \theta_i \}_{i=1}^{n}$, where the GP posterior has no uncertainty. 
\Cref{fig:fig_51_1} illustrates the decomposition of GP-UJB computed for a standard normal density in $\R$.
\end{remark}

\subsection{Approximation Error of GP-UJB} \label{sec:app}

We establish that the L2 error between GP-UJB and the unnormalized target $q(\theta)$ is controlled by the average-case interpolation error of the GP regression.
In particular, if an additional regularity condition is met, the former is specifically smaller than the later.
This also supports the use of GP-UJB as it converges faster than the concentration of the GP regression itself.
In the following result, $\| \cdot \|_{L_2(\Theta)}$ denotes the L2 norm of functions on $\Theta$ under the Lebesgue measure.

\begin{proposition} \label{prop:gp_ujb_bound}
Let $g(\theta) := \phi^{-1}(q(\theta))$, that is, $q(\theta) = \phi(g(\theta))$.
Suppose that $\phi$ is differentiable in $(0, \infty)$ with the derivative $\phi'$ satisfying $| \phi'(a) | \le \exp(a)$ for all $a \in (0, \infty)$.
Let $U_n(\theta)$ be GP-UJB with the GP posterior $\mathcal{GP}(m_n, k_n)$.
It holds for some constant $C > 0$ that
\begin{align}
	\| U_n - q \|_{L_2(\Theta)}^2 \le C \E_{f \sim \mathcal{GP}(m_n, k_n)}\left[ \| f - g \|_{L^2(\Theta)}^2 \right] .
\end{align}
In particular, if $g(\theta) < -1/2$ and $m_n(\theta) + k_n(\theta, \theta) < -1/2$, we have $C < 1$.
\end{proposition}

\begin{remark} \label{re:gp_ujb_bound}
	The last conditions $g(\theta) < -1/2$ and $m_n(\theta) + k_n(\theta, \theta) < -1/2$ are mild when employing the log-transformation $\phi^{-1}(\cdot) = \log(\cdot)$.
	Crucially, the target density $p(\theta)$ is invariant to the scaling of the proportional $q(\theta)$ due to normalization.
	Consequently, $q(\theta)$ can be arbitrarily rescaled by a constant factor without altering the original density $p(\theta)$.
	Without loss of generality, we may therefore assume that $q(\theta)$ is bounded by a constant $B$, s.t., $g(\theta) = \log q(\theta) \le \log B < -1/2$.
	By similar reasoning, a sufficiently small scaling factor ensures that the GP approximation also satisfies that $m_n(\theta) + k_n(\theta, \theta) < -1/2$.
	Furthermore, this latter condition is readily verifiable numerically.
\end{remark}

The proof of \Cref{prop:gp_ujb_bound} is provided in \Cref{apx:proof}.
\Cref{prop:gp_ujb_bound} assures that, if the GP posterior converges, the GP-UJB also converges. 
Convergence analysis of the average-case GP regression error is well established; see, e.g., \citep{Kanagawa2018}.

\subsection{Hyperparameters of GP-UJB} \label{sec:setting}

Finally, we discuss the specification of the hyperparameters for GP-UJB: the convex non-negative function $\phi$ and the hyperparameters of the GP prior.

First, the choice of $\phi$ governs the trade-off between exploitation and exploration in GP-UJB. 
Computationally, it is advantageous to select $\phi$ that yields a closed-form expression of GP-UJB.
\Cref{tab:cf} summarizes examples of such choices. 
Modeling the log unnormalized density $\log q(\theta)$ via GPs is a well-established approach in the literature \citep[e.g.,][]{Osborne2012,Jarvenpaa2021}. 
Following these precedents, we adopt the log transformation $\phi^{-1}(\cdot) = \log(\cdot)$ for the inverse, meaning that $\phi(\cdot) = \exp(\cdot)$, as the default setting. 
We explore alternative choices in \Cref{apx:choice_phi}.

\begin{table}[h] 
\caption{Closed-form expressions of GP-UJB for selected functions $\phi$. Here, $\varphi$ and $\Phi$ denote the probability density function and the cumulative distribution function of the standard normal distribution, respectively, and we define $\sigma_n(\theta) := \sqrt{k_n(\theta, \theta)}$ for better presentation.} \label{tab:cf}
\centering
\begin{tabular}{ c c c }
	\hline
	Function $\phi(x)$ & Inverse $\phi^{-1}(x)$ & GP-UJB $U_n(\theta)$ \\
	\hline
	$\exp(x)$ & $\log(x)$ & $\exp\left( m_n(\theta) + 0.5 \sigma_n^2(\theta) \right)$ \\ 
	$\max(0, x)$ & $x$ & $m_n(\theta) \Phi\left( m_n(\theta) / \sigma_n(\theta) \right) + \sigma_n(\theta) \varphi\left( m_n(\theta) / \sigma_n(\theta) \right) $ \\
	$x^2$ & $\sqrt{x}$ & $m_n^2(\theta) + \sigma_n^2(\theta)$ \\
	\hline
\end{tabular}
\end{table}

Next, the mean and covariance functions of the GP prior typically depend on a set of hyperparameters $\psi$, such as the length scale and variance of a Gaussian kernel. 
Properly tuning $\psi$ is critical for the performance of GP regression. 
In the context of BO, updating $\psi$ via maximum likelihood estimation at each iteration is standard practice to maximize the performance.
Accordingly, we update $\psi$ for GP-UJB by maximizing the log-likelihood of the GP posterior at each step.

\begin{remark}[Initial Points]
	In sequential design settings, GP surrogate models may suffer from instability during the initial iterations due to overfitting to scarce data. 
	To mitigate this, it is common to select the first few points at random rather than optimization. 
	Algorithm \ref{alg:tmp} readily accommodates this strategy: for a fixed integer $N_0 > 0$, the initial $N_0$ samples $\{ \theta_n^* \}_{n=1}^{N_0}$ can be selected from the proposal sequence, such as the scaled Halton sequence, without optimization.
\end{remark}

\section{Empirical Assessment} \label{sec:experiment}

This section examines the empirical performance of BIS across four distinct experiments, ranging from numerical validation to practical application.
First, we begin with three benchmarks to validate sample quality and convergence of BIS.
Next, we consider a weather forecast model from \cite{Lorenz1995}, where Bayesian inference is performed via synthetic likelihoods based on simulations \citep{Hakkarainen2012}.
We then apply BIS to exact Bayesian inference for g-and-k models \citep{Prangle2020}, a setting where density evaluation requires numerical optimization at each data point.
Finally, we illustrate the practical utility of our approach through a real-world application, modeling precipitation data \citep{Kaufman2008} via an MRF.
Source code for the experiments is available at \url{https://github.com/takuomatsubara/Bandit-Importance-Sampling}.

\subsection{Benchmark Densities} \label{sec:benchmark}

We evaluate BIS on three bivariate benchmark densities, each designed to pose a distinct challenge.
The first is a simple unimodal Gaussian density.
The second is a bimodal density, representing a scenario with multiple modes that requires efficient domain exploration.
The third is a banana-shaped density, illustrating a case with complex tail-probability geometries.
All the three densities are defined in the following form proposed in \cite{Jarvenpaa2021}:
\begin{align}
	p(\theta) \propto \exp\bigg( - \frac{1}{2}
		\begin{bmatrix}
			T_1(\theta) \\
			T_2(\theta)
		\end{bmatrix}^\mathrm{T}
		\begin{bmatrix}
			1 & \rho \\
			\rho & 1
		\end{bmatrix}
		\begin{bmatrix}
			T_1(\theta) \\
			T_2(\theta)
		\end{bmatrix}
	\bigg) ,
\end{align}
where the definitions of $T_1$, $T_2$, and $\rho$ for each density are provided in \Cref{apx:experiment_1}. 
The domains are defined as $\Theta = [-16, 16]^2$, $\Theta = [-6, 6]^2$, and $\Theta = [-6, 6] \times [-20, 2]$, respectively.
The aim of this experiment is to demonstrate the performance of BIS under each benchmark scenario.

\begin{figure}[t]
	\includegraphics[width=0.31\textwidth]{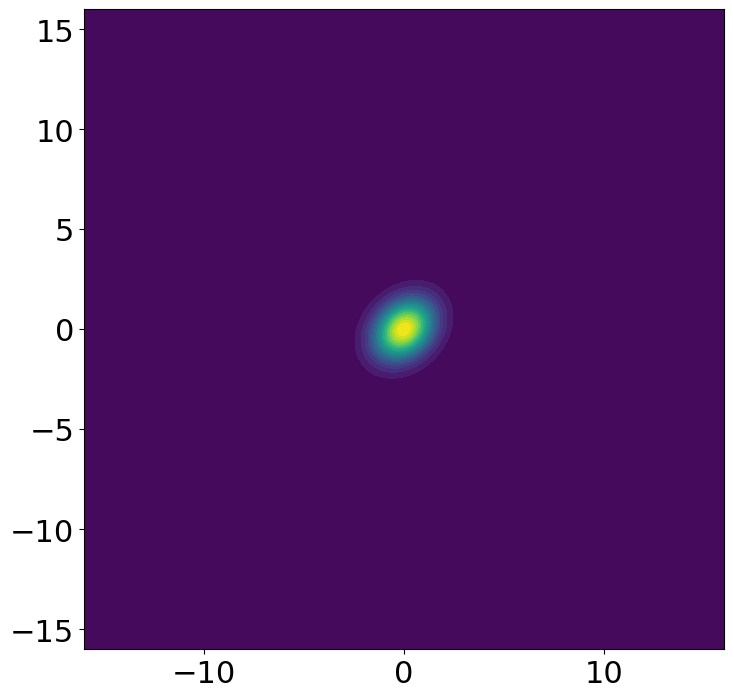}
	\hfill
	\includegraphics[width=0.3025\textwidth]{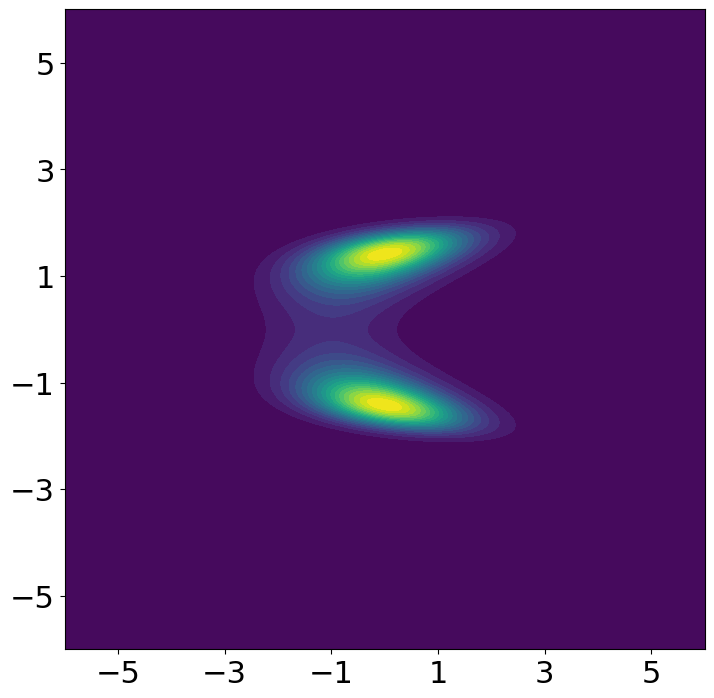}
	\hfill
	\includegraphics[width=0.31\textwidth]{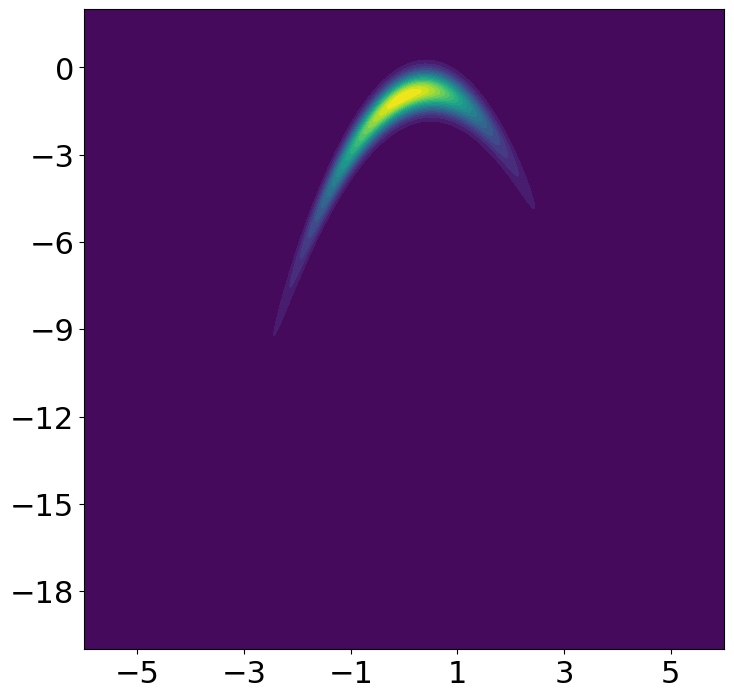}
	\hfill
	
	\subcaptionbox{gaussian}{\includegraphics[width=0.31\textwidth]{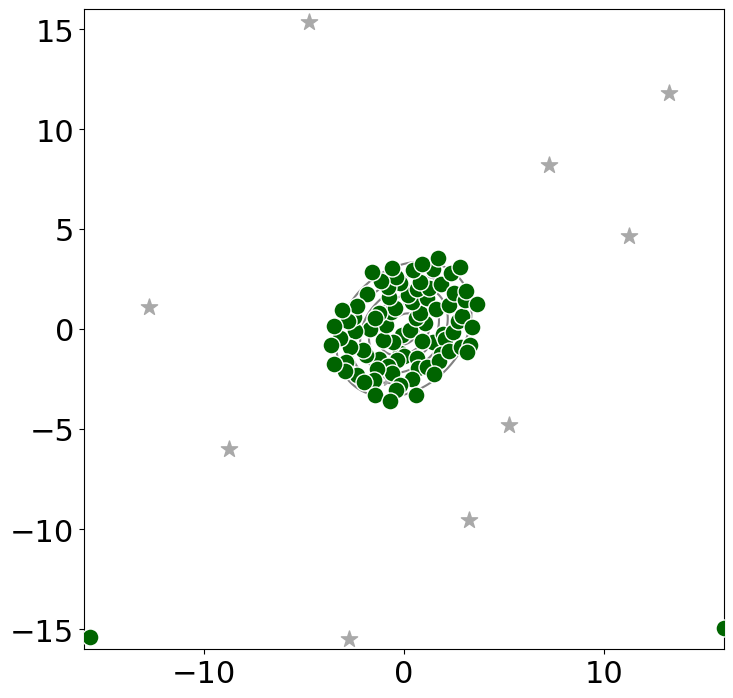}}
	\hfill
	\subcaptionbox{bimodal}{\includegraphics[width=0.3025\textwidth]{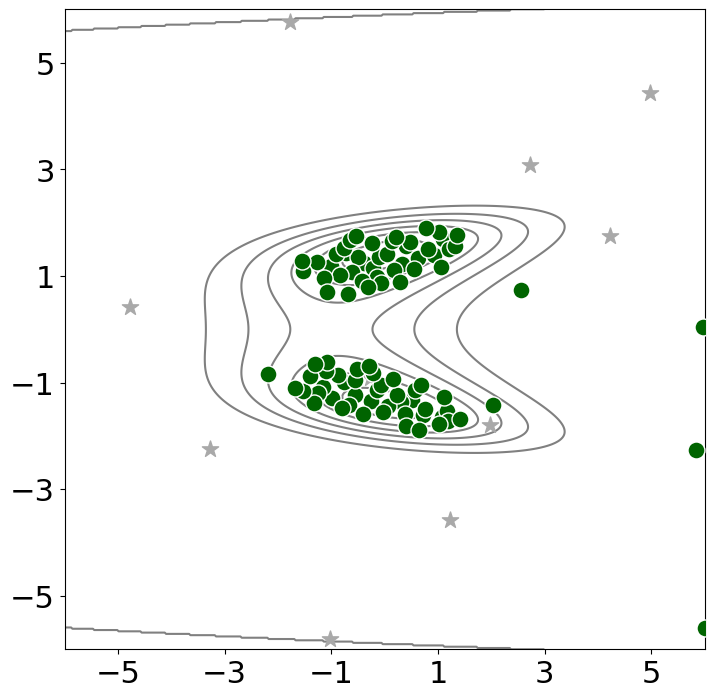}}
	\hfill
	\subcaptionbox{banana}{\includegraphics[width=0.31\textwidth]{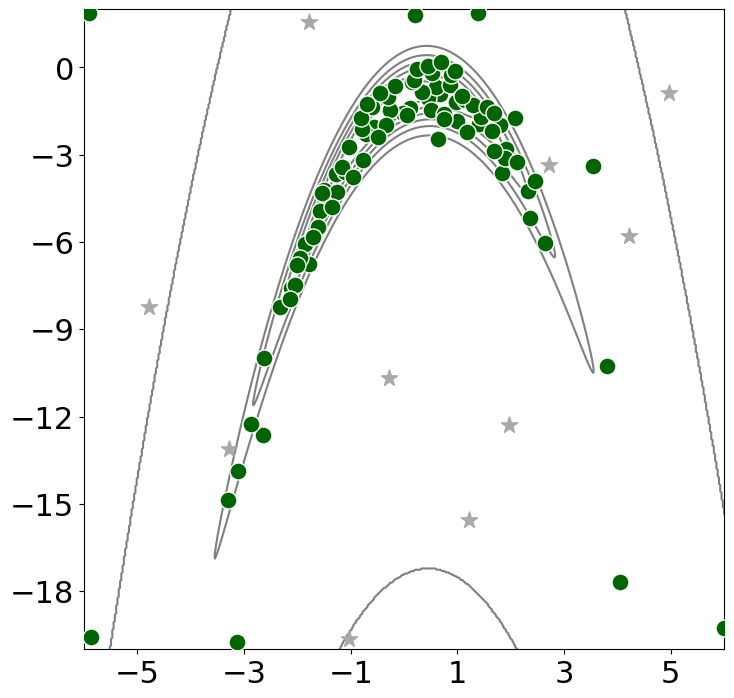}}
	\hfill
	\caption{
		Visualization of 100 samples generated by BIS for each density. The top panels show the contour colourmap of each density. The bottom panels show the samples, distinguishing between the initial 10 points (stars) and the subsequent 90 points selected by BIS (circles). Solid lines indicate density contours, where the values are raised to the power of $1/3$ to accentuate the tail geometry.
	} \label{fig:fig_41_1}
\end{figure}

For each target density, we generated 100 weighted samples using BIS.
The initial 10 points were selected from a Halton sequence to ensure the stability of the GP surrogate model, following the protocol in \cite{Jarvenpaa2021}.
The GP-UJB prior employed a zero mean function and a Gaussian kernel, treating the length scale and variance as tunable hyperparameters.
The candidate pool size was fixed at $M = 2048$, where a sensitivity analysis regarding this parameter is provided in \Cref{apx:pool_size}.
\Cref{fig:fig_41_1} presents a comparison between the target densities and the obtained samples, demonstrating that BIS effectively identifies the high-probability regions.
Notably, the method successfully resolved the complex geometries of the bimodal and banana-shaped densities.

We quantified the approximation quality of the obtained samples using the MMD.
\Cref{fig:fig_41_2} displays the MMD between the target density and the BIS samples.
For comparison, we evaluated a standard importance sampling baseline using the scaled Halton sequence with self-normalized weights.
\Cref{apx:experiment_1} investigated the number of samples required for the standard importance sampling baseline to match the approximation error of 100 BIS samples.
On average, BIS reduces the sample size required to achieve a comparable error by nearly 95\%.

\begin{figure}[t]
	\hfill
	\subcaptionbox{gaussian}{\includegraphics[width=0.32\textwidth]{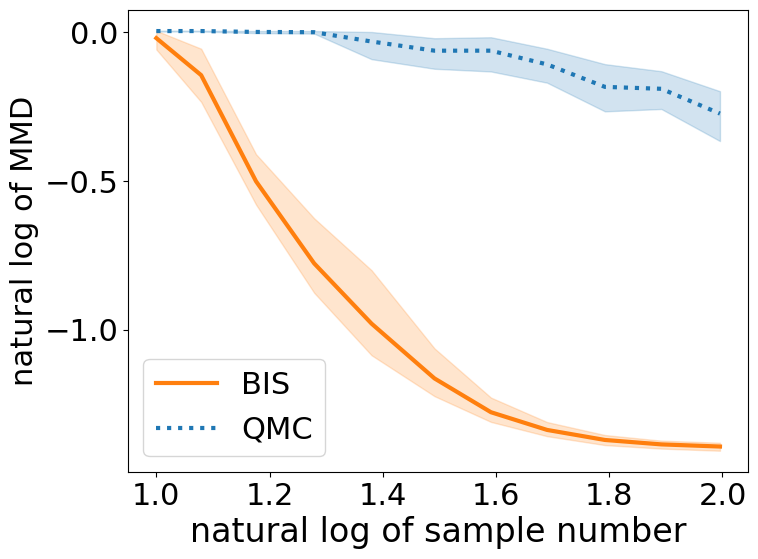}}
	\hfill
	\subcaptionbox{bimodal}{\includegraphics[width=0.32\textwidth]{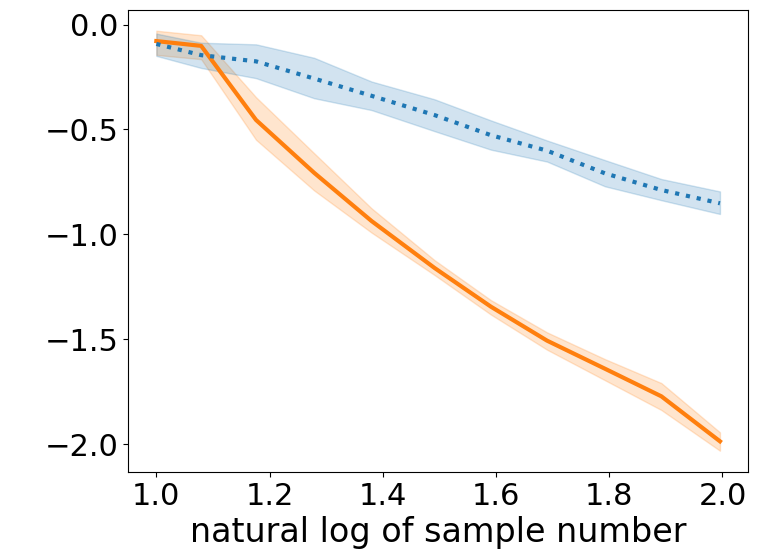}}
	\hfill
	\subcaptionbox{banana-shaped}{\includegraphics[width=0.32\textwidth]{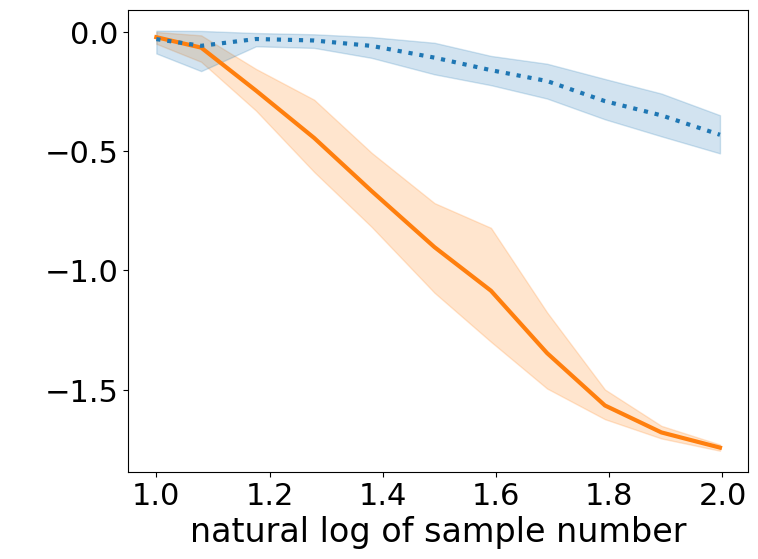}}
	\hfill
	\hfill
	\caption{
		Approximation error for BIS (solid line) and standard importance sampling (dotted line). Shaded regions indicate 95\% confidence intervals over 10 independent runs.
	} \label{fig:fig_41_2}
\end{figure}

\begin{figure}[t]
	\hfill
	\subcaptionbox{gaussian}{\includegraphics[width=0.32\textwidth]{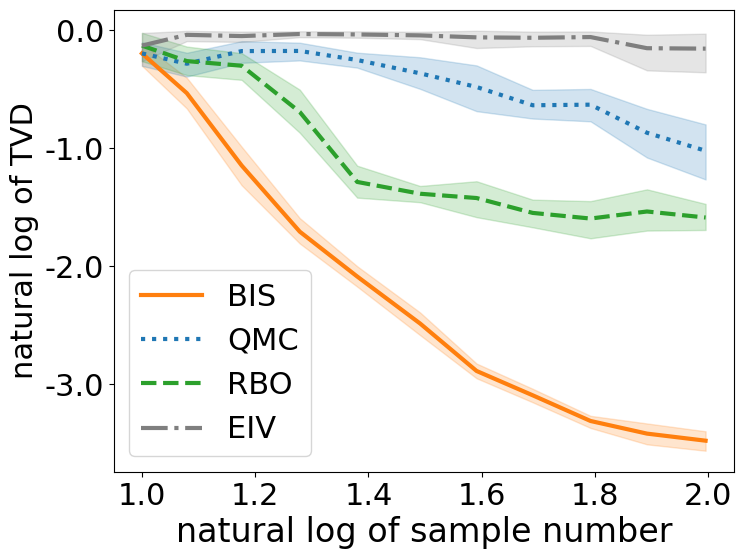}}
	\hfill
	\subcaptionbox{bimodal}{\includegraphics[width=0.32\textwidth]{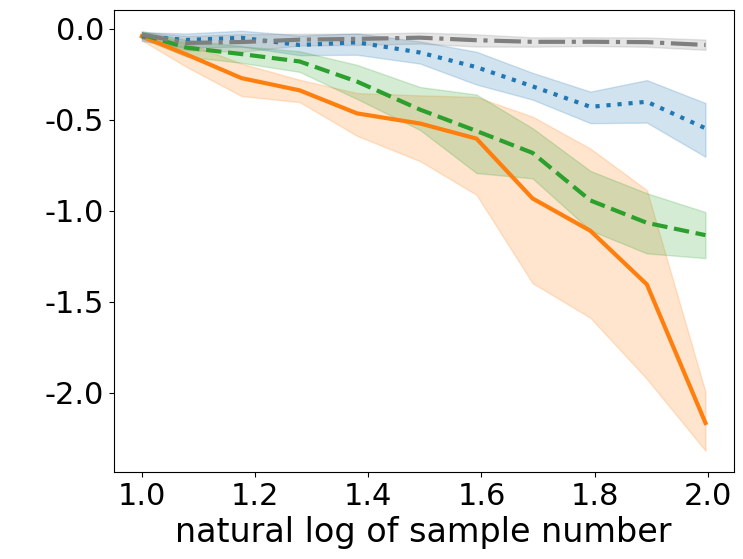}}
	\hfill
	\subcaptionbox{banana-shaped}{\includegraphics[width=0.32\textwidth]{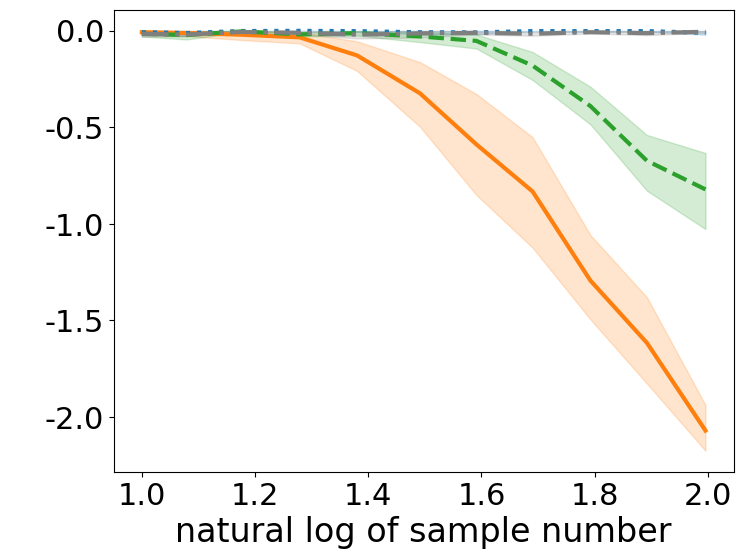}}
	\hfill
	\hfill
	\caption{
		Approximation error of GP-based density estimators trained on input points selected by BIS (solid line), QMC (dotted line), the randomized BO (dashed line), and the EIV design (dash-dotted line).
		Shaded regions indicate 95\% confidence intervals over 10 independent runs.
	} \label{fig:fig_41_5}
\end{figure}

BIS does \emph{not} require the GP surrogate to converge to the target density, where a reasonable approximation suffices for guiding sample selection.
Nonetheless, assessing the GP surrogate's fidelity provides valuable insight.
To this end, we evaluated the total variation distance (TVD) between the target density and the density estimator derived by plugging the GP mean $m_n(\theta)$ in the exponential map.
\Cref{fig:fig_41_5} illustrates the decay of the approximation error. 
We compared the estimator trained on the BIS samples against the estimator trained on three alternative designs.
The first is standard QMC samples.
The second is the expected integrated variance (EIV) design of \cite{Jarvenpaa2021}, which modifies \cite{Sinsbeck2017} for better numerical stability.
The last is the randomized BO of \cite{Kim2024}.
The distinction lies solely in the selection of input points; the GP hyperparameters were identical across all methods.
As demonstrated in \Cref{fig:fig_41_5}, the GP surrogate trained on BIS samples consistently yields the lowest TVD across all cases.

\subsection{Lorenz Weather Model} \label{sec:weather}

\cite{Wilks2005} proposed a stochastic modification of the Lorenz model for weather forecast \citep{Lorenz1995}.
The formulation distinguishes between `slow' observable and `fast' unobservable variables, where the influence of the latter is modeled as a stochastic noise term.
The system consists of 40 observable variables $( x_1(t), \dots, x_{40}(t) ) \in \R^{40}$, whose time evolution is described by 
\begin{align}
	\frac{d}{d t} x_k(t) = - x_{k-1}(t) \big( x_{k-2}(t) - x_{k+1}(t) \big) - x_{k}(t) + 10 - \theta_1 - \theta_2 x_k(t) + \eta_k(t) \label{eq:Lorenz}
\end{align}
where $\theta := (\theta_1, \theta_2) \in \R^2$ represents the system parameter, and $\eta_k(t)$ denotes the stochastic noise defined subsequently.
Here, the indices $k$ follow cyclic boundary conditions, meaning that $x_{0} = x_{40}$, $x_{-1} = x_{39}$, and $x_{41} = x_{1}$.
The model is defined over the time interval $[0, 4]$.

We assume the initial state of the 40 variables at $t = 0$ is known, followed by daily observations over a 20-day period.
In this model, one day corresponds to $0.2$ time units; thus, we observe a sample path at 20 equidistant time points within the interval $[0, 4]$.
Numerical integration of the differential equation \eqref{eq:Lorenz} is performed using the fourth-order Runge-Kutta method with a time step of $\delta t = 0.025$, following \cite{Thomas2022}.
Adopting the formulation of \cite{Wilks2005}, the stochastic noise $\eta_k(t)$ is modeled as a first-order autoregressive process:
\begin{align}
	\eta_k(t + \delta t) = \phi \eta_k(t) + \sqrt{1 - \phi^2} \epsilon_k(t),
\end{align}
where the autoregressive coefficient is fixed at $\phi = 0.4$, and $\epsilon_k(t)$ represents i.i.d.~standard normal noise.
The initial noise term $\eta_k(0)$ is set to $\sqrt{1 - \phi^2} \epsilon_k(0)$.
In our experiments, the initial states of the variables were drawn from a standard normal distribution, and the observed sample path was simulated using the ground-truth parameters $(\theta_1, \theta_2) = (2.0, 0.1)$.

\begin{figure}[t]
	\hfill
	\subcaptionbox{exact posterior}{\includegraphics[width=0.3\textwidth]{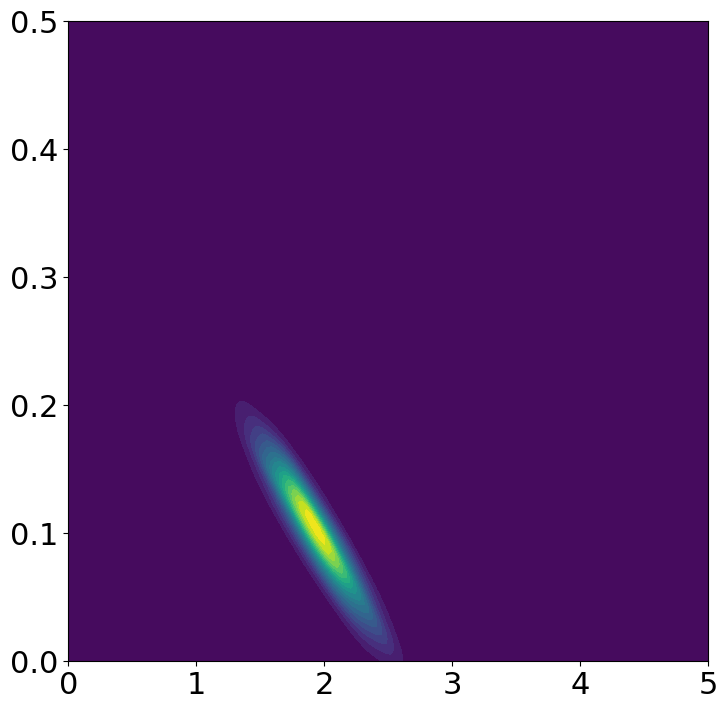}}
	\hfill
	\subcaptionbox{BIS samples}{\includegraphics[width=0.3\textwidth]{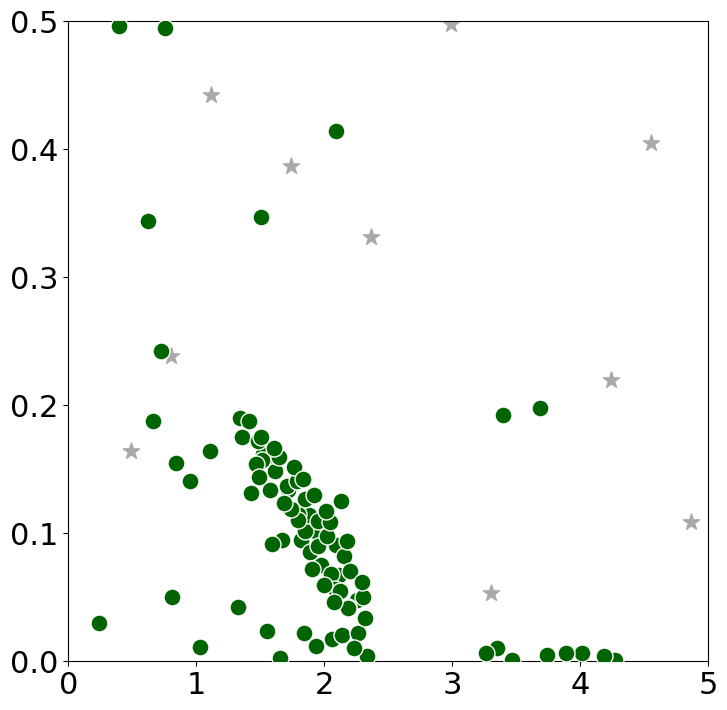}}
	\hfill
	\subcaptionbox{KDE based on BIS samples}{\includegraphics[width=0.3\textwidth]{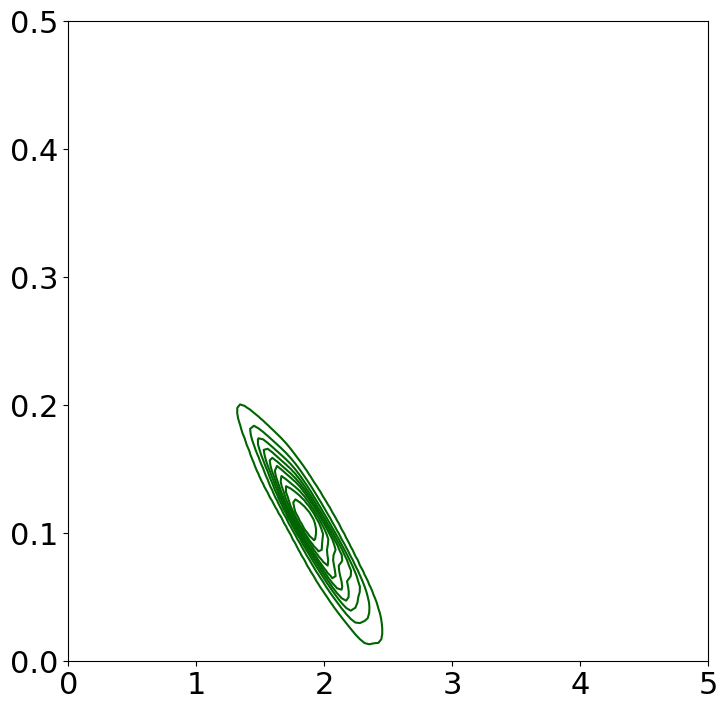}}
	\hfill
	\caption{Visualization of 100 samples generated by BIS for the Lorenz model. The left panel is the contour colourmap of the target posterior. The middle panels show the samples, distinguishing between the initial 10 points (stars) and the subsequent 90 points selected by BIS (circles). The right panel shows the contour colourmap of the KDE of the 100 samples with the associated weights.} \label{fig:fig_43_1}
\end{figure}

\cite{Hakkarainen2012} proposed using a synthetic likelihood to perform Bayesian inference of the Lorenz model; see \cite{Frazier2023} for a general introduction of synthetic likelihoods.
The summary statistics defining the synthetic likelihood are detailed in \Cref{apx:experiment_2}.
We employed a uniform prior over the parameter domain $\Theta = [0.0, 5.0] \times [0.0, 0.5]$, following \cite{Jarvenpaa2021}.
To ensure that the estimation error of the synthetic likelihood is negligible, we computed it using 10,000 simulation paths at each query point $\theta$.
The experimental configuration for BIS follows that of \Cref{sec:benchmark}.
\Cref{fig:fig_43_1} presents the 100 samples generated by BIS, alongside a kernel density estimator (KDE) derived from the weighted samples.
For comparison, we constructed a ground-truth proxy using standard importance sampling with 10,000 points from a scaled Halton sequence and self-normalized weights.
The figure demonstrates that the KDE obtained by BIS accurately reconstructs the target posterior.
This observation is quantitatively supported by the MMD between the target posterior and the 100 BIS samples, which was $0.008$.

\subsection{G-and-K Model} \label{sec:gandk}

The g-and-k model is an expressive parametric family capable of capturing location, scale, skewness, and kurtosis.
Notable applications include modeling of foreign exchange returns, which frequently exhibit heavy tails and asymmetry \citep{Prangle2020}.
The model is parameterized by a vector $\theta = (\theta_1, \theta_2, \theta_3, \theta_4) \in \mathbb{R}^4$.
It is defined explicitly through its quantile function $Q(u \mid \theta)$, as the density is unavailable in closed form.
For $u \in (0, 1)$, the quantile function is given by
\begin{align}
	Q(u \mid \theta) = \theta_1 + \theta_2 z(u) \big( 1 +  c \tanh( 0.5 \theta_3 z(u)) \big) \big( 1 + z(u)^2 \big)^{\theta_4} 
\end{align}
where $z(u)$ denotes the standard normal quantile function, and $c$ is a constant fixed at $0.8$.

\begin{figure}[t]
	\includegraphics[width=\textwidth]{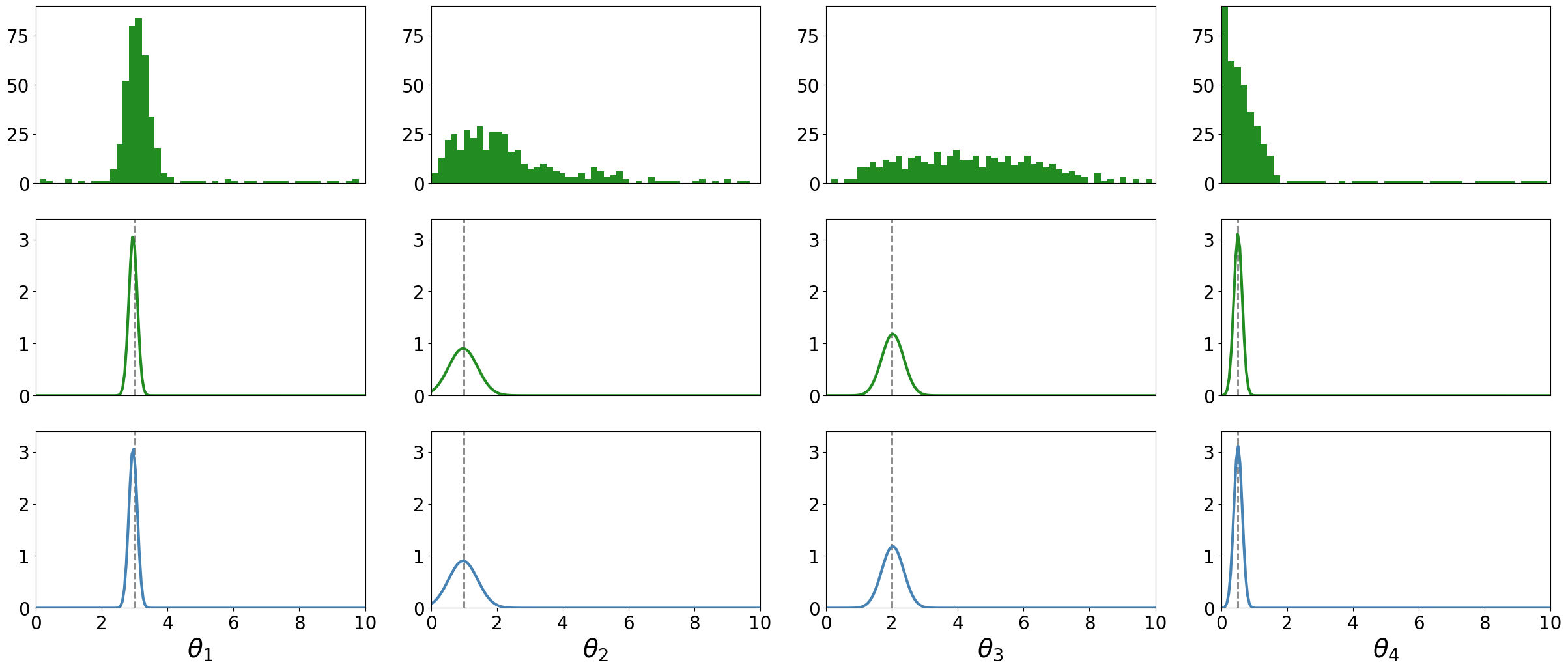}
	\caption{
		Marginal distributions for the g-and-k model parameters $\theta = (\theta_1, \theta_2, \theta_3, \theta_4)$. The top panels display the histograms of 400 samples generated by BIS without weights. The middle panels display the KDEs derived from the 400 samples with the associated importance weights. The bottom panels show ground-truth marginals. Vertical dotted lines indicate the true parameter values.
	} \label{fig:fig_42_1}
\end{figure}

Let $Q'(u \mid \theta)$ denote the derivative of the quantile function with respect to $u$.
\cite{Rayner2002} showed that the density $p(x \mid \theta)$ can be expressed in terms of the quantile function and its inverse, denoted as $u = Q^{-1}(x \mid \theta)$.
This expression is given by
\begin{align}
	p(x \mid \theta) = \frac{1}{ Q'( Q^{-1}(x \mid \theta) \mid \theta) } .
\end{align}
However, evaluating this density is computationally expensive because the inverse $Q^{-1}(x \mid \theta)$ has no closed-form expression and must be computed numerically for each data point $x$.
While approximate Bayesian computation (ABC) is a standard approach to bypass this intractability, it targets an approximation of the posterior, which often results in variance inflation.
In contrast, BIS performs exact Bayesian inference by directly targeting the true posterior density.

We generated a synthetic dataset of 1,000 observations $\{ x_i \}_{i=1}^{1000}$ from the g-and-k model with true parameters $\theta_0 = (3, 1, 2, 0.5)$.
The parameter domain was defined as the hyperrectangle $\Theta = [0, 10]^4$, with a uniform prior placed over this domain.
BIS was configured to produce 400 weighted samples, initialized with 40 points from the scaled Halton sequence to stabilize the GP surrogate.
To encode sufficient tail decay, we assigned a quadratic polynomial mean function with coefficient hyperparameters to the GP prior, mitigating the risk of probability mass escaping outward due to the concentration of measure.
We employed a Gaussian kernel with length-scale and variance hyperparameters for the covariance.
The candidate pool size was set to $M = 320,000$.
\Cref{fig:fig_42_1} displays the marginal histograms for each parameter coordinate derived from the 400 BIS samples, along with the corresponding KDE.
For benchmarking, we compared these results against ground-truth marginals constructed from 320,000 samples of the scaled Halton sequence with self-normalized importance weights.
The figure confirms that the BIS samples accurately reconstruct the target posterior, correctly concentrating probability mass around the true parameter $\theta_0$.

\subsection{US Precipitation Anomalies} \label{sec:mrfm}

Our final experiment demonstrates the practical utility of BIS through a large-scale real-world application involving MRFs.
The dataset, previously analyzed in \cite{Kaufman2008} and \cite{Bolin2024}, consists of annual precipitation anomalies for the year 1962, observed at $N = 7352$ weather stations across the United States (see \Cref{fig:fig_44_1}).
We employ a standard Gaussian MRF \citep{Rue2005} to model this dataset.
Notably, \cite{Lyne2015} highlighted this class of models as a potential failure case for their posterior sampling approach relying on stochastic approximation of the computationally expensive likelihood.
This challenge motivates the use of BIS, which evaluates the exact likelihood while minimizing the computational burden through efficient evaluation.

\begin{figure}[t]
	\hfill
	\includegraphics[width=0.5\textwidth]{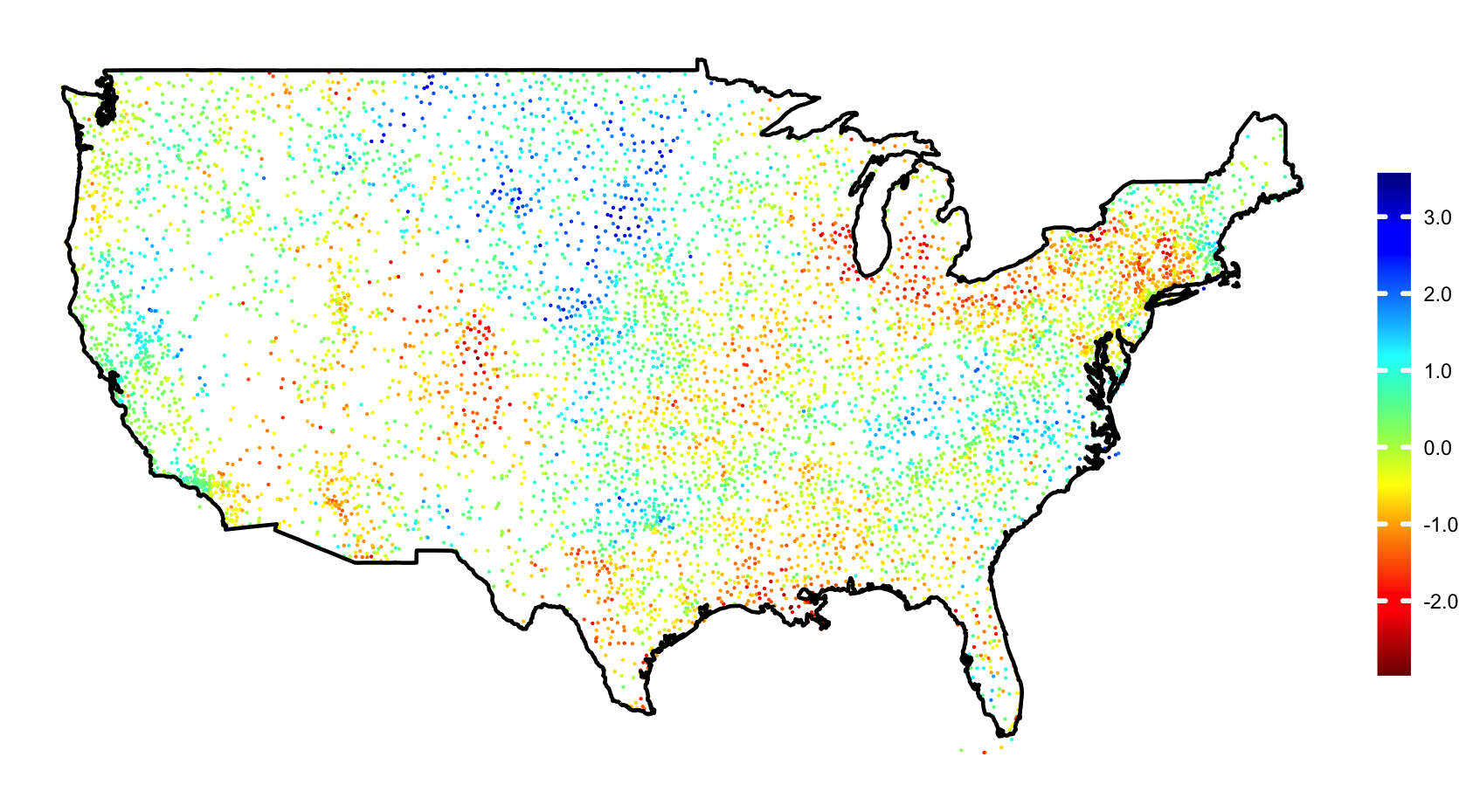}
	\hfill
	\hfill
	\caption{
		Visualization of the United States annual precipitation anomalies in 1962.
	} \label{fig:fig_44_1}
\end{figure}

Let $\mathbf{y} = (y_1, \dots, y_N) \in \mathbb{R}^N$ denote the observations.
Let $\mathbf{X}$ be a latent Gaussian MRF defined on a fixed triangulation of the continental United States.
The covariance structure of the MRF is governed by the Mat\'{e}rn kernel with range and variance parameters $\rho$ and $\sigma$, respectively.
The smoothness parameter of the Mat\'{e}rn kernel is fixed at $2.0$ to ensure that the latent field is at least once differentiable.
Let $\mathbf{Q}(\rho, \sigma)$ be the precision matrix of the stochastic partial differential equation (SPDE) associated with the Mat\'{e}rn kernel.
See, e.g., \cite{Lindgren2011} for the SPDE formulation of the Mat\'{e}rn kernel.
Let $\mathbf{A}$ be the projection matrix mapping the latent field $\mathbf{X}$ from the triangulation to the observation locations.
This hierarchical model is then expressed as:
\begin{align}
	\mathbf{y} \mid \mathbf{X}, \rho, \sigma, \sigma_0 \sim \mathcal{N}(\mathbf{A} \mathbf{X}, \sigma_0^{-2} \mathbf{I}); \qquad \mathbf{X} \mid \rho, \sigma \sim \mathcal{N}(\mathbf{0}, \mathbf{Q}(\rho, \sigma)^{-1}),
\end{align}
where $\sigma_0$ denotes the scale of the observation noise.

Our interest lies in Bayesian inference of the three-dimensional parameter vector $\theta = (\rho, \sigma, \sigma_0)$.
Since the posterior distribution of the latent field $\mathbf{X}$ is Gaussian, $\mathbf{X}$ can be marginalized out analytically.
The marginal log-likelihood $l(\theta)$ is explicitly available as:
\begin{align}
	2 l(\theta) = & N \log \left( \sigma_0^2 \right) + \log \det \left( \mathbf{Q}(\rho, \sigma) \right) - \log \det \left( \mathbf{Q}(\rho, \sigma) + \sigma_0^2 \mathbf{A}^{\mathrm{T}} \mathbf{A} \right) \\
	& - \sigma_0^2 \mathbf{y}^{\mathrm{T}} \mathbf{y} + \sigma_0^4 \mathbf{y}^{\mathrm{T}} \mathbf{A} \left( \mathbf{Q}(\rho, \sigma) + \sigma_0^2 \mathbf{A}^{\mathrm{T}} \mathbf{A} \right)^{-1} \mathbf{A}^{\mathrm{T}} \mathbf{y} - N \log \left( 2 \pi \right) .
\end{align}
Although the likelihood is available in closed form, its evaluation presents a significant computational challenge due to the determinant and inverse operations involving massive precision matrices.
While the sparsity of $\mathbf{Q}(\rho, \sigma)$ mitigates the cost, these linear algebra operations scale superlinearly as the dataset size increases.
This necessitates an evaluation-efficient strategy for posterior sampling.
We selected this application because, while the precision matrix is high-dimensional, performing MCMC remains feasible, allowing us to establish a robust baseline for comparison.
We used the efficient implementation of the likelihood provided by \cite{Bolin2020}.

For Bayesian inference, we assigned a uniform prior to $\theta$ over the compact domain $\Theta := [0.1, 5.0]^3$.
The interval for the range parameter $\rho$ corresponds to approximately 0\%--10\% of the diameter of the observation region.
The same interval was used for the variance parameter $\sigma$ and noise scale $\sigma_0$, as it was large enough to cover 2 to 3 times the standard deviation of the observed values.
BIS was configured to produce 200 weighted samples, initialized with 20 points from the scaled Halton sequence.
To encode sufficient tail decay, we assigned a quadratic polynomial mean function with coefficient hyperparameters to the GP prior.
We employed a Gaussian kernel with length-scale and variance hyperparameters for the covariance.
The candidate pool size was set to $M = 80,000$.
For the baseline, we performed $10,000$ iterations of MCMC, with a burn-in period of $5,000$ iterations and thinning applied every $10$ iterations to the remaining $5,000$ samples.
\Cref{fig:fig_44_2} presents the resulting posterior approximation, using the KDE of the MCMC samples as the ground truth.
The proposed BIS algorithm achieves an accurate approximation of the target posterior using only $200$ samples.
The ability of BIS to recover the posterior structure with a limited sample budget demonstrates its computational advantage.

\begin{figure}[t]
	\hfill
	\includegraphics[width=0.85\textwidth]{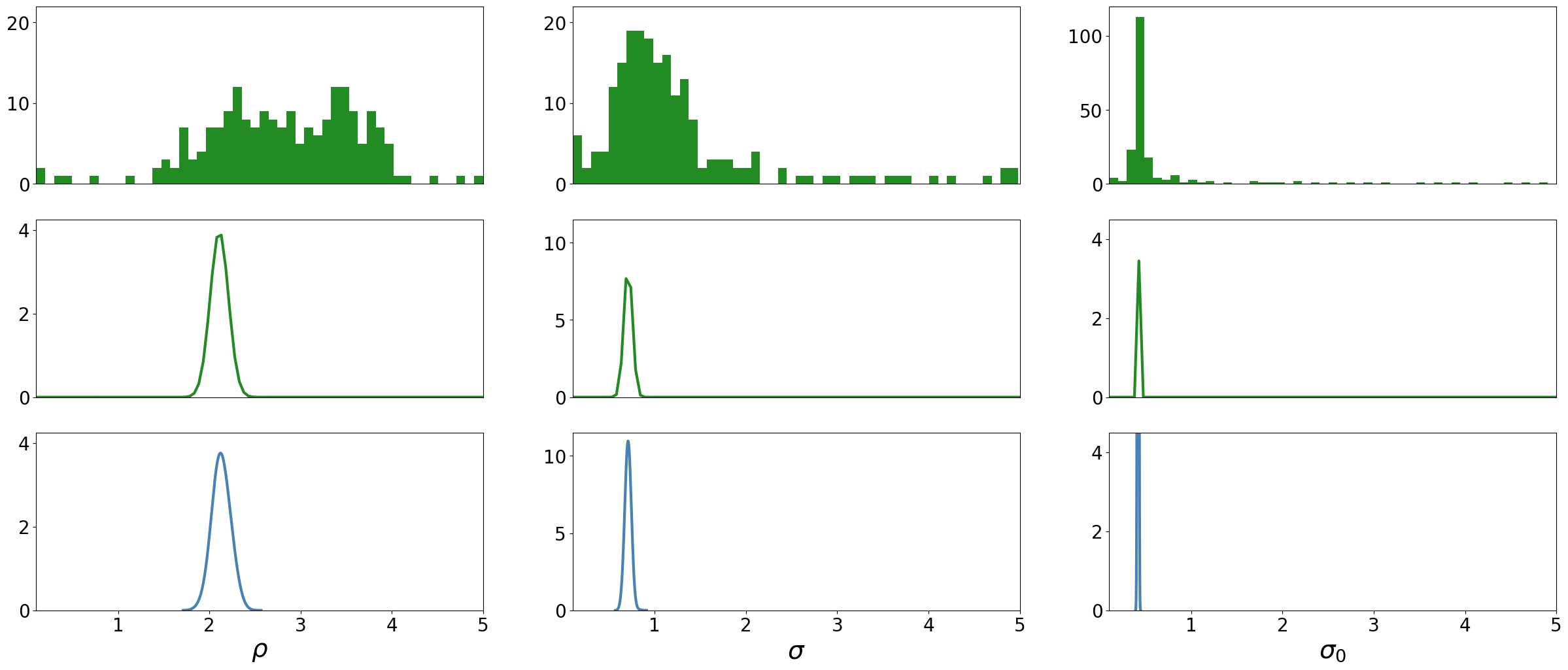}
	\hfill
	\hfill
	\caption{
		Marginal distributions for the MRF parameters $\theta = (\rho, \sigma, \sigma_0)$. The top panels display the histograms of 200 samples generated by BIS. The middle panels display the KDEs derived from the 200 samples with the associated weights. The bottom panels show baseline marginals by MCMC.
	} \label{fig:fig_44_2}
\end{figure}

\section{Related Work} \label{sec:related}

To the best of our knowledge, BIS is the first framework to formulate the sample design in importance sampling as a bandit problem.
In most existing experimental design approaches for GPs, selected points serve merely as training inputs for a surrogate model.
In contrast, BIS leverages the GP surrogate to guide the adaptive design of samples that converge to the target density.

\cite{Sinsbeck2017} developed a sequential design strategy to efficiently approximate the likelihood via a GP surrogate.
Their method utilizes the Bayes estimator of the log-likelihood that minimizes the Bayes risk associated with the GP surrogate, proposing the design that maximizes this Bayes estimator at each iteration.
Once the specified number of iterations is finished, the posterior is approximated via a plug-in estimator, where the log-likelihood is replaced by the GP posterior mean.
\cite{Jarvenpaa2021} extended this approach by modeling the posterior density directly and proposing a batched sequential design.
However, these methods rely on several heuristics and lack theoretical guarantees regarding the convergence of the surrogate model to the true target.

More recently, \cite{Kim2024} proposed a modification of BO for posterior approximation.
Their strategy augments the standard BO acquisition with an additional random point at each iteration to ensure exploration.
They demonstrated that, under a well-specified GP prior, the mean function of the GP posterior converges to the log target posterior.
While this result is expected, given that random sampling eventually covers the entire domain, it offers the important insight that imposing simple constraints to ensure domain coverage is sufficient to guarantee consistency of the GP surrogate.

Beyond density approximation, GP surrogate modeling has been widely adopted in Bayesian computation.
Notable applications include computing model evidence via active learning \citep{Osborne2012, Gunter2014}, accelerating ABC through summary statistic modeling \citep{Meeds2014} or discrepancy modeling \citep{Gutmann2016}, merging posteriors from data subsets \citep{Nemeth2018}, and acting as physics-informed surrogates for forward simulators in inverse problems \citep{Bai2024}.

\section{Conclusion} \label{sec:conclusion}

Sampling from black-box densities has become an increasingly common task, particularly in Bayesian inference for complex models.
Despite its importance, there is limited literature on particle-based sampling methods specifically tailored for such densities.
To address this gap, we introduced BIS, the first importance sampling framework that formulates the sample design as a bandit problem.

BIS offers a simple yet powerful framework, defined by a proposal sequence and a point-selection criterion.
In this work, we proposed GP-UJB as a default criterion, leveraging GP surrogate models of the target density.
The GP-UJB criterion naturally decomposes into two terms that balance the exploration of new regions with the exploitation of existing information.
Theoretically, we established the weak convergence of the weighted samples generated by BIS.
Empirically, we demonstrated the efficacy of the method on both benchmark densities and real-world applications.

While this paper focused on GP-UJB, the general BIS framework allows for flexible choices of point-selection criteria.
As aforementioned in \Cref{sec:introduction}, applying GPs in high-dimensional spaces remains a well-known open challenge, a limitation inherited by GP-UJB.
Consequently, designing criteria suitable for high-dimensional densities requires substantial further investigation.
Promising directions for future research include developing novel point-selection strategies that extend beyond the scope of GP surrogate modeling.

\bibliographystyle{abbrvnat}
\bibliography{bibliography}

@article{Jarvenpaa2021,
	author = {Marko J{\"a}rvenp{\"a}{\"a} and Michael U. Gutmann and Aki Vehtari and Pekka Marttinen},
	title = {Parallel {G}aussian Process Surrogate {B}ayesian Inference with Noisy Likelihood Evaluations},
	volume = {16},
	journal = {{B}ayesian Analysis},
	number = {1},
	publisher = {International Society for {B}ayesian Analysis},
	pages = {147--178},
	year = {2021},
}

@article{Kim2024,
	title={Enhancing {G}aussian Process Surrogates for Optimization and Posterior Approximation via Random Exploration}, 
	author={Hwanwoo Kim and Daniel Sanz-Alonso},
	year={2024},
	journal={arXiv:~2401.17037},
}

@article{Rayner2002,
	author = {Rayner, G. D. and MacGillivray, H. L.},
	journal = {Statistics and Computing},
	number = {1},
	pages = {57--75},
	title = {Numerical Maximum Likelihood Estimation for the G-and-K and Generalized G-and-H Distributions},
	volume = {12},
	year = {2002},
}

@book{Dick2010,
	title={Digital Nets and Sequences: Discrepancy Theory and Quasi–Monte Carlo Integration}, 
	publisher={Cambridge University Press}, 
	author={Dick, Josef and Pillichshammer, Friedrich}, 
	year={2010}
}

@article{Morokoff1994,
	author = {Morokoff, William J. and Caflisch, Russel E.},
	title = {Quasi-Random Sequences and Their Discrepancies},
	journal = {SIAM Journal on Scientific Computing},
	volume = {15},
	number = {6},
	pages = {1251-1279},
	year = {1994},
}

@article{Halton1964,
	author = {J.H. Halton},
	title = {Algorithm 247:~Radical-Inverse Quasi-Random Point Sequence},
	year = {1964},
	volume = {7},
	number = {12},
	journal = {Communications of the ACM},
	pages = {701–702},
}

@book{Paulsen2016, 
	title={An Introduction to the Theory of Reproducing Kernel {H}ilbert Spaces}, 
	publisher={Cambridge University Press}, 
	author={Paulsen, Vern I. and Raghupathi, Mrinal}, 
	year={2016}, 
}

@InProceedings{Smola2007,
	author={Smola, Alex and Gretton, Arthur and Song, Le and Sch{\"o}lkopf, Bernhard},
	title={A {H}ilbert Space Embedding for Distributions},
	booktitle={Algorithmic Learning Theory},
	year={2007},
	pages={13--31},
}

@article{Gretton2012,
	author  = {Arthur Gretton and Karsten M. Borgwardt and Malte J. Rasch and Bernhard Sch{{\"o}}lkopf and Alexander Smola},
	title   = {A Kernel Two-Sample Test},
	journal = {Journal of Machine Learning Research},
	year    = {2012},
	volume  = {13},
	number  = {25},
	pages   = {723-773},
}

@article{Simon-Gabriel2023,
	author  = {Carl-Johann Simon-Gabriel and Alessandro Barp and Bernhard Sch\"{o}lkopf and Lester Mackey},
	title   = {Metrizing Weak Convergence with Maximum Mean Discrepancies},
	journal = {Journal of Machine Learning Research},
	year    = {2023},
	volume  = {24},
	number  = {184},
	pages   = {1--20},
}

@book{Wainwright2019, 
	title={High-Dimensional Statistics: A Non-Asymptotic Viewpoint}, 
	publisher={Cambridge University Press}, 
	author={Wainwright, Martin J.}, 
	year={2019}
}

@article{Prangle2020,
	author = {Dennis Prangle},
	title = {Gk:~An R Package for the g-and-k and Generalised g-and-h Distributions},
	year = {2020},
	journal = {{The R Journal}},
	pages = {7--20},
	volume = {12},
	number = {1}
}

@article{Thomas2022,
	author = {Owen Thomas and Ritabrata Dutta and Jukka Corander and Samuel Kaski and Michael U. Gutmann},
	title = {Likelihood-Free Inference by Ratio Estimation},
	volume = {17},
	journal = {{B}ayesian Analysis},
	number = {1},
	pages = {1--31},
	year = {2022},
}

@article{Lorenz1995,
	author = {Edward N. Lorenz},
	title = {Predictability: a problem partly solved},
	year = {1995},
	journal = {Seminar on Predictability},
	volume = {1},
	pages = {1-18},
}

@article{Wilks2005,
	author = {Wilks, Daniel S.},
	title = {Effects of stochastic parametrizations in the {L}orenz '96 system},
	journal = {Quarterly Journal of the Royal Meteorological Society},
	volume = {131},
	number = {606},
	pages = {389-407},
	year = {2005}
}

@article{Hakkarainen2012,
	author = {{Hakkarainen}, J. and {Ilin}, A. and {Solonen}, A. and {Laine}, M. and {Haario}, H. and {Tamminen}, J. and {Oja}, E. and {J{\"a}rvinen}, H.},
	title = {On closure parameter estimation in chaotic systems},
	journal = {Nonlinear Processes in Geophysics},
	year = {2012},
	volume = {19},
	number = {1},
	pages = {127-143},
}

@book{Rasmussen2005,
	author = {Rasmussen, Carl Edward and Williams, Christopher K. I.},
	title = {{G}aussian Processes for Machine Learning},
	publisher = {The MIT Press},
	year = {2005},
}

@article{Kanagawa2018,
	title={{G}aussian Processes and Kernel Methods:~A Review on Connections and Equivalences}, 
	author={Motonobu Kanagawa and Philipp Hennig and Dino Sejdinovic and Bharath K Sriperumbudur},
	year={2018},
	journal={arXiv:~1807.02582},
}

@article{Sinsbeck2017,
	author = {Sinsbeck, Michael and Nowak, Wolfgang},
	title = {Sequential Design of Computer Experiments for the Solution of {B}ayesian Inverse Problems},
	journal = {SIAM/ASA Journal on Uncertainty Quantification},
	volume = {5},
	number = {1},
	pages = {640-664},
	year = {2017},
}

@inproceedings{Osborne2012,
	author = {Osborne, Michael and Garnett, Roman and Ghahramani, Zoubin and Duvenaud, David K and Roberts, Stephen J and Rasmussen, Carl},
	booktitle = {Advances in Neural Information Processing Systems},
	title = {Active Learning of Model Evidence Using {B}ayesian Quadrature},
	volume = {25},
	year = {2012}
}

@inproceedings{Gunter2014,
	author = {Gunter, Tom and Osborne, Michael A. and Garnett, Roman and Hennig, Philipp and Roberts, Stephen J.},
	title = {Sampling for inference in probabilistic models with fast {B}ayesian quadrature},
	year = {2014},
	booktitle = {Proceedings of the 28th International Conference on Neural Information Processing Systems},
	volume = {2},
	pages = {2789–2797},
}

@book{Mockus1989,
	author = {Jonas Mockus},
	publisher = {Springer},
	title = {{B}ayesian Approach to Global Optimization:~Theory and Applications},
	year = {1989},
}

@article{Gutmann2016,
	author  = {Michael U. Gutmann and Jukka Corander},
	title   = {{B}ayesian Optimization for Likelihood-Free Inference of Simulator-Based Statistical Models},
	journal = {Journal of Machine Learning Research},
	year    = {2016},
	volume  = {17},
	number  = {125},
	pages   = {1--47},
}

@inproceedings{Meeds2014,
	author = {Meeds, Edward and Welling, Max},
	title = {{GPS-ABC}:~{G}aussian process surrogate approximate {B}ayesian computation},
	year = {2014},
	booktitle = {Proceedings of the Thirtieth Conference on Uncertainty in Artificial Intelligence},
	pages = {593–602},
}

@article{Bai2024,
	author = {Bai, Tianming and Teckentrup, Aretha L. and Zygalakis, Konstantinos C.},
	journal = {Statistics and Computing},
	number = {4},
	pages = {139},
	title = {{G}aussian processes for {B}ayesian inverse problems associated with linear partial differential equations},
	volume = {34},
	year = {2024},
}

@inproceedings{Srinivas2010,
	author = {Srinivas, Niranjan and Krause, Andreas and Kakade, Sham and Seeger, Matthias},
	title = {{G}aussian Process Optimization in the Bandit Setting:~No Regret and Experimental Design},
	year = {2010},
	booktitle = {Proceedings of the 27th International Conference on International Conference on Machine Learning},
	pages = {1015–1022},
}

@article{Bull2011,
	author  = {Adam D. Bull},
	title   = {Convergence Rates of Efficient Global Optimization Algorithms},
	journal = {Journal of Machine Learning Research},
	year    = {2011},
	volume  = {12},
	number  = {88},
	pages   = {2879--2904},
}

@article{Girolami2008,
	title = {{B}ayesian Inference for Differential Equations},
	journal = {Theoretical Computer Science},
	volume = {408},
	number = {1},
	pages = {4-16},
	year = {2008},
	author = {Mark Girolami},
}

@article{Warne2020,
	title = {A practical guide to pseudo-marginal methods for computational inference in systems biology},
	journal = {Journal of Theoretical Biology},
	volume = {496},
	pages = {110255},
	year = {2020},
	author = {David J. Warne and Ruth E. Baker and Matthew J. Simpson},
}

@article{Beaumont2002,
	author = {Beaumont, Mark A and Zhang, Wenyang and Balding, David J},
	title = {Approximate {B}ayesian Computation in Population Genetics},
	journal = {Genetics},
	volume = {162},
	number = {4},
	pages = {2025-2035},
	year = {2002},
}

@article{Marin2012,
	author = {Marin, Jean-Michel and Pudlo, Pierre and Robert, Christian P. and Ryder, Robin J.},
	journal = {Statistics and Computing},
	number = {6},
	pages = {1167--1180},
	title = {Approximate {B}ayesian computational methods},
	volume = {22},
	year = {2012},
}

@article{Besag1986,
	author = {Julian Besag},
	journal = {Journal of the Royal Statistical Society. Series B (Methodological)},
	number = {3},
	pages = {259--302},
	title = {On the Statistical Analysis of Dirty Pictures},
	urldate = {2025-06-03},
	volume = {48},
	year = {1986}
}

@book{Rue2005,
	author = {H. Rue and L. Held},
	title = {{G}aussian {M}arkov Random Fields:~Theory and Applications},
	edition = {1st Edition},
	year = {2005},
	publisher = {Chapman and Hall/CRC}
}

@article{Goodreau2009,
	author = {Goodreau, Steven M. and Kitts, James A. and Morris, Martina},
	journal = {Demography},
	number = {1},
	pages = {103--125},
	title = {Birds of a feather, or friend of a friend?~using exponential random graph models to investigate adolescent social networks},
	volume = {46},
	year = {2009},
}

@book{Brooks2011,
	author = {S. Brooks and A. Gelman and G. Jones and X.L. Meng},
	title = {{Handbook of {M}arkov Chain Monte Carlo}},
	edition = {1st Edition},
	year = {2011},
	publisher = {Chapman and Hall/CRC}
}

@book{Chen2010,
	author = {Ming-Hui Chen and Peter Müller and Dongchu Sun and Keying Ye and Dipak K. Dey},
	title = {{Frontiers of Statistical Decision Making and {B}ayesian Analysis}},
	subtitle = {{In Honor of James O. Berger}},
	year = {2010},
	publisher = {Springer}
}

@book{Gramacy2020,
	author = {Robert B. Gramacy},
	title = {Surrogates},
	subtitle = {{G}aussian Process Modeling, Design, and Optimization for the Applied Sciences},
	edition = {1st Edition},
	year = {2020},
	publisher = {Chapman and Hall/CRC}
}

@ARTICLE{Shahriari2016,
	author={Shahriari, Bobak and Swersky, Kevin and Wang, Ziyu and Adams, Ryan P. and de Freitas, Nando},
	journal={Proceedings of the IEEE}, 
	title={Taking the Human Out of the Loop:~A Review of {B}ayesian Optimization}, 
	year={2016},
	volume={104},
	number={1},
	pages={148-175},
}

@article{Tokdar2010,
	author = {Tokdar, Surya T. and Kass, Robert E.},
	title = {Importance Sampling:~A Review},
	journal = {WIREs Computational Statistics},
	volume = {2},
	number = {1},
	pages = {54-60},
	year = {2010}
}

@article{Hammersley1954,
	author = {J. M. Hammersley and K. W. Morton},
	journal = {Journal of the Royal Statistical Society. Series B (Methodological)},
	number = {1},
	pages = {23--38},
	title = {Poor Man's Monte Carlo},
	volume = {16},
	year = {1954}
}

@ARTICLE{Bugallo2017,
	author={Bugallo, Monica F. and Elvira, Victor and Martino, Luca and Luengo, David and Miguez, Joaquin and Djuric, Petar M.},
	journal={IEEE Signal Processing Magazine}, 
	title={Adaptive Importance Sampling:~The Past, the Present, and the Future}, 
	year={2017},
	volume={34},
	number={4},
	pages={60-79},
}

@article{Lyne2015,
	author = {Anne-Marie Lyne and Mark Girolami and Yves Atchad{\'e} and Heiko Strathmann and Daniel Simpson},
	number = {4},
	pages = {443--467 },
	title = {On {R}ussian Roulette Estimates for {B}ayesian Inference with Doubly-Intractable Likelihoods},
	volume = {30},
	year = {2015},
	journal = {Statistical Science}
}

@book{Steinwart2008,
	author = {Ingo Steinwart and Andreas Christmann},
	title = {{Support Vector Machines}},
	year = {2008},
	publisher = {Springer}
}

@article{Sriperumbudur2010,
	author  = {Bharath K. Sriperumbudur and Arthur Gretton and Kenji Fukumizu and Bernhard Sch{{\"o}}lkopf and Gert R.G. Lanckriet},
	title   = {{H}ilbert Space Embeddings and Metrics on Probability Measures},
	journal = {Journal of Machine Learning Research},
	year    = {2010},
	volume  = {11},
	number  = {50},
	pages   = {1517--1561},
}

@article{Nemeth2018,
	author = {Christopher Nemeth and Chris Sherlock},
	title = {Merging {MCMC} Subposteriors through {G}aussian-Process Approximations},
	volume = {13},
	journal = {{B}ayesian Analysis},
	number = {2},
	publisher = {International Society for {B}ayesian Analysis},
	pages = {507--530},
	year = {2018},
}

@book{Niederreiter1992,
	author = {Niederreiter, Harald},
	title = {Random Number Generation and Quasi-Monte Carlo Methods},
	publisher = {Society for Industrial and Applied Mathematics},
	year = {1992},
}

@article{Cotter2019,
	title={Ensemble transport adaptive importance sampling},
	author={Cotter, Colin and Cotter, Simon and Russell, Paul},
	journal={SIAM/ASA Journal on Uncertainty Quantification},
	volume={7},
	number={2},
	pages={444--471},
	year={2019},
	publisher={SIAM}
}

@article{Cotter2020,
	title={Transport map accelerated adaptive importance sampling, and application to inverse problems arising from multiscale stochastic reaction networks},
	author={Cotter, Simon L and Kevrekidis, Ioannis G and Russell, Paul T},
	journal={SIAM/ASA Journal on Uncertainty Quantification},
	volume={8},
	number={4},
	pages={1383--1413},
	year={2020},
	publisher={SIAM}
}

@article{Frazier2023,
	author = {David T. Frazier and David J. Nott and Christopher Drovandi and Robert Kohn},
	title = {{B}ayesian Inference Using Synthetic Likelihood:~Asymptotics and Adjustments},
	journal = {Journal of the American Statistical Association},
	volume = {118},
	number = {544},
	pages = {2821--2832},
	year = {2023},
}

@article{Lindgren2011,
	author = {Lindgren, Finn and Rue, Håvard and Lindström, Johan},
	title = {An explicit link between Gaussian fields and Gaussian Markov random fields: the stochastic partial differential equation approach},
	journal = {Journal of the Royal Statistical Society: Series B (Statistical Methodology)},
	volume = {73},
	number = {4},
	pages = {423-498},
	year = {2011}
}

@article{Kaufman2008,
	author = {Cari G. Kaufman and Mark J. Schervish and Douglas W. Nychka},
	title = {Covariance Tapering for Likelihood-Based Estimation in Large Spatial Data Sets},
	journal = {Journal of the American Statistical Association},
	volume = {103},
	number = {484},
	pages = {1545--1555},
	year = {2008},
}

@article{Bolin2020,
	title = {The rational {SPDE} approach for {Gaussian} random fields with general smoothness},
	author = {David Bolin and Kristin Kirchner},
	journal = {Journal of Computational and Graphical Statistics},
	year = {2020},
	volume = {29},
	number = {2},
	pages = {274--285},
}

@article{Bolin2024,
	author = {David Bolin and Alexandre B. Simas and Zhen Xiong},
	title = {Covariance–Based Rational Approximations of Fractional SPDEs for Computationally Efficient Bayesian Inference},
	journal = {Journal of Computational and Graphical Statistics},
	volume = {33},
	number = {1},
	pages = {64--74},
	year = {2024},
}

\appendix

\setcounter{figure}{0}
\setcounter{table}{0}
\setcounter{equation}{0}

\renewcommand{\thefigure}{S\arabic{figure}}
\renewcommand{\thetable}{S\arabic{table}}
\renewcommand{\theequation}{S\arabic{equation}}

\begin{center}
	\LARGE \textbf{Supplementary Material}
\end{center}

\vspace{40pt}

This supplementary material contains details of the results presented in the main text.
\Cref{apx:preliminary} presents preliminary lemmas useful for subsequent proofs.
\Cref{apx:proof} provides proofs of all the theoretical results presented in the main text.
\Cref{apx:simulation} delivers simulation studies on the settings of BIS.
\Cref{apx:experiment} describes additional details of the experiments.
It is convenient to introduce the following multi-index notations that will be used throughout.

\paragraph*{Notations}
Let $\alpha = (\alpha_1, \dots, \alpha_d)$ be a $d$-dimensional vector of binary values, that is, $\alpha \in \{ 0, 1 \}^d$. 
We call $\alpha$ a binary multi-index.
Let $| \alpha |$ be the number of non-zero elements of $\alpha$.
Recall that $\partial_i$ denotes the partial derivative of a function with respect to the i-th coordinate of the argument.
Let $\partial_i^0$ be an identity operator and let $\partial_i^1$ be the partial derivative $\partial_i$.
Denote by $\partial^\alpha$ the mixed partial derivative $\partial^\alpha = \partial_1^{\alpha_1} \cdots \partial_d^{\alpha_d}$.
Denote by $\theta^\alpha$ a $|\alpha|$-dimensional vector consisting of all i-th coordinates of $\theta$ such that $\alpha_i = 1$.
Denote by $\Theta^\alpha$ a set of all values of $\theta^\alpha$.
Denote by $\theta_*^\alpha$ a $d$-dimensional vector s.t.~the i-th coordinate is equal to that of $\theta$ if $\alpha_i = 1$ or fixed to a constant $b_i$ otherwise, where $b_i$ denotes the right boundary of the i-th coordinate of the hyperrectangular domain $\Theta$.

\section{Preliminary Lemmas} \label{apx:preliminary}

This section contains preliminary lemmas on RKHS and the star discrepancy, which will be useful for the subsequent proofs in \Cref{apx:proof}.

\subsection{Preliminary Lemmas on RKHS} \label{apx:premilinary_rkhs}

Recall that $\kappa$ is the kernel used for MMD, which meets the kernel condition presented in \Cref{sec:theory}.
Denote by $\langle \cdot, \cdot \rangle_{\H_\kappa}$ the inner product of the RKHS $\H_\kappa$ of the kernel $\kappa$.
Any function $f$ in the RKHS $\H_\kappa$ satisfies an identity called the \emph{reproducing property} s.t.
\begin{align}
	f(\theta) = \langle f(\cdot), \kappa(\cdot, \theta) \rangle_{\H_\kappa}
\end{align}
at each $\theta$ pointwise \citep[][Section 4.2]{Steinwart2008}.
The reproducing property also applies to the kernel $\kappa$ itself, where it holds that $\kappa(\theta, \theta') = \langle \kappa(\cdot, \theta), \kappa(\cdot, \theta') \rangle_{\H_\kappa}$ at each $\theta, \theta'$ pointwise.
We have the following useful lemma on a uniform bound of functions in $\H_\kappa$.

\begin{lemma} \label{lem:rk_bound}
	Any function $f \in \H_\kappa$ s.t.~$\| f \|_{\H_\kappa} \le 1$ satisfies that
	\begin{align}
		\sup_{\theta \in \Theta} | f(\theta) | \le 1 .
	\end{align}
\end{lemma}

\begin{proof}
	By the reproducing property and the Cauchy-Schwartz inequality, we have
	\begin{align}
		| f(\theta) | \le \| f \|_{\H_\kappa} \| \kappa(\cdot, \theta) \|_{\H_\kappa} = \| f \|_{\H_\kappa} \sqrt{ \kappa(\theta, \theta) } \le 1
	\end{align}
	since the kernel $\kappa$ is bounded by $1$ under the kernel condition in \Cref{sec:theory}.
\end{proof}

We have another useful lemma on a uniform bound of the mixed partial derivative of functions in $\H_\kappa$.
Recall that $\partial_i$ denotes the partial derivative with respect to the i-th coordinate of the argument $\theta$, and that $\partial^\alpha$ denotes the mixed partial derivative specified by a given binary multi-index $\alpha \in \{ 0, 1 \}^d$.
Denote by $\partial_i \kappa(\theta, \theta')$ and $\partial'_i \kappa(\theta, \theta')$ the partial derivative of $\kappa(\theta, \theta)$ with respect to the i-th coordinate of, respectively, the first argument $\theta$ and the second argument $\theta'$.

\begin{lemma} \label{lem:deriv}
	Any function $f \in \H_\kappa$ s.t.~$\| f \|_{\H_\kappa} \le 1$ satisfies that 
	\begin{align}
		\sup_{\theta \in \Theta} | \partial^\alpha f(\theta) | \le 1
	\end{align}
	for all binary multi-index $\alpha \in \{ 0, 1 \}^d$.
\end{lemma}

\begin{proof}
	Let $I_\alpha$ be a set of indices $i$ s.t.~$\alpha_i = 1$.
	Similarly, let $I_{\setminus \alpha}$ be a set of indices $i$ s.t.~$\alpha_i = 0$.
	The mixed partial derivative $\partial^\alpha$ can be expressed as $\partial^\alpha = \prod_{i \in I_\alpha} \partial_i$ using the set $I_\alpha$.
	Corollary 4.36 of \cite{Steinwart2008} implies that
	\begin{align}
		| \partial^\alpha f(\theta) | \le \| f \|_{\H_\kappa} \sqrt{ \prod_{i \in I_\alpha} \partial_i \partial_i' \kappa(\theta, \theta) } .
	\end{align}
	It follows from $\| f \|_{\H_\kappa} \le 1$ and the factorization $\kappa(\theta, \theta') = \prod_{i=1}^{d} \kappa_i(\theta_{(i)}, \theta'_{(i)})$ that
	\begin{align}
		| \partial^\alpha f(\theta) | \le \sqrt{ \prod_{i \in I_{\setminus \alpha}} \kappa_i(\theta_{(i)}, \theta_{(i)}) } \sqrt{ \prod_{i \in I_\alpha} \partial_i \partial_i' \kappa_i(\theta_{(i)}, \theta_{(i)}) } .
	\end{align}
	Recall the kernel condition that every $\kappa_i$ is bounded by $1$ and admits the derivative $\partial_i \partial_i' \kappa_i(\theta_{(i)}, \theta_{(i)}')$ uniformly bounded by $1$.
	Therefore we have
	\begin{align}
		\sup_{\theta \in \Theta} | \partial^\alpha f(\theta) | \le \sup_{\theta \in \Theta} \sqrt{ \prod_{i \in I_{\setminus \alpha}} \kappa_i(\theta_{(i)}, \theta_{(i)}) } \sqrt{ \prod_{i \in I_\alpha} \partial_i \partial_i' \kappa_i(\theta_{(i)}, \theta_{(i)}) } \le 1 .
	\end{align}
	This concludes the proof since the upper bound does not depend on a choice of $\alpha$.
\end{proof}

\subsection{Preliminary Lemmas on Star Discrepancy} \label{apx:premilinary_sd}

We defined the star discrepancy $\operatorname{D}^*$ for points in the domain $\Theta$ in the main text.
Typically, the star discrepancy is defined for points in the unit cube $[0, 1]^d$.
We shall see that the star discrepancy $\operatorname{D}^*$ defined for the domain $\Theta$ is equivalent to that defined for the unit cube $[0, 1]^d$.
To distinguish the two star discrepancies, let $\operatorname{D}^\dagger$ be the star discrepancy for the unit cube $[0, 1]^d$.
For a given sequence of $K$ points $\{ \eta_i \}_{i=1}^{K}$ in the unit cube $[0, 1]^d$, the latter star discrepancy $\operatorname{D}^\dagger$ is defined as follows:
\begin{align}
	\operatorname{D}^\dagger\left( \{ \eta_n \}_{n=1}^{K} \right) := \sup_{ B_0 \in \mathcal{B}_0 } \left| \frac{\# \{ \eta_n \in B_0 \} }{K} - \operatorname{Vol}(B_0) \right| \label{eq:star_discrepancy_0}
\end{align}
where $\mathcal{B}_0$ is a set of all subboxes contained in $[0, 1]^d$ s.t.~the left-boundary at every i-th coordinate is $0$, and $\# \{ \eta_n \in B_0 \}$ denotes the number of points contained in a box $B_0 \in \mathcal{B}_0$.
Notice that the former star discrepancy $\operatorname{D}^*$ introduced in the main text is a generalization of the above definition of the latter $\operatorname{D}^\dagger$, where they coincide with each other when $\Theta = [0, 1]^d$.

Consider a bijective affine map $\tau: [0, 1]^d \to \Theta$ between the unit cube $[0, 1]^d$ and the domain $\Theta$.
Since $\Theta$ is hyperrectangular, the bijective affine map $\tau$ always exists uniquely.
Now we formally state the equivalence of one star discrepancy $\operatorname{D}^*$ to the other $\operatorname{D}^\dagger$ under the map $\tau$.

\begin{lemma} \label{lem:sd}
	Let $\{ \eta_n \}_{n=1}^{K}$ be the image of a given sequence of points $\{ \theta_n \}_{n=1}^{K}$ in $\Theta$ under the map $\tau^{-1}$, that is, each point satisfies $\eta_n = \tau^{-1}(\theta_n)$.
	It holds for any sequence $\{ \theta_n \}_{n=1}^{K}$ that
	\begin{align}
		\operatorname{D}^*\left( \{ \theta_n \}_{n=1}^{K} \right) = \operatorname{D}^\dagger\left( \{ \eta_n \}_{n=1}^{K} \right) .
	\end{align}
\end{lemma}

\begin{proof}
	Recall that $\# \{ \eta_n \in B_0 \}$ denotes the number of points contained in a box $B_0 \in \mathcal{B}_0$.
	The number is invariant to the bijective transform $\tau$, meaning that $\# \{ \eta_n \in B_0 \} = \# \{ \tau(\eta_n) \in \tau(B_0) \}$, where we denote by $\tau(C)$ the image of a subset $C$ of $[0, 1]^d$ under the map $\tau$.
	Thus, the star discrepancy $\operatorname{D}^\dagger( \{ \eta_n \}_{n=1}^{K} )$ on the right-hand side can be expressed as 
	\begin{align}
		\operatorname{D}^\dagger\left( \{ \eta_n \}_{n=1}^{K} \right) & = \sup_{ B_0 \in \mathcal{B}_0 } \left| \frac{ \# \{ \tau(\eta_n) \in \tau(B_0) \} }{K} - \frac{ \operatorname{Vol}(B_0) }{ \operatorname{Vol}( [0, 1]^d ) } \right| .
	\end{align}
	The ratio between the volume $\operatorname{Vol}(B_0)$ and $\operatorname{Vol}( [0, 1]^d )$ is invariant to the transform $\tau$, meaning that
	\begin{align}
		\frac{ \operatorname{Vol}(B_0) }{ \operatorname{Vol}( [0, 1]^d ) } = \frac{ \operatorname{Vol}( \tau( B_0 ) ) }{ \operatorname{Vol}( \tau( [0, 1]^d ) ) } = \frac{ \operatorname{Vol}( \tau( B_0 ) ) }{ \operatorname{Vol}( \Theta ) }
	\end{align}
	Let $\mathcal{B}$ be all subboxes contained in $\Theta$ s.t.~the left-boundary at every i-th coordinate is equal to that of $\Theta$.
	Notice that the map $\tau$ defines the bijection between the set $\mathcal{B}_0$ of subboxes in $[0, 1]^d$ and the set $\mathcal{B}$ of subboxes in $\Theta$.
	This means that we have $\sup_{ B_0 \in \mathcal{B}_0 } f( \tau(B_0) ) = \sup_{ B \in \mathcal{B} } f( B )$ for any function $f$ that takes a box in $\mathcal{B}$ as an argument.
	Therefore, we have
	\begin{align}
		\operatorname{D}^\dagger\left( \{ \eta_n \}_{n=1}^{K} \right) & = \sup_{ B_0 \in \mathcal{B}_0 } \left| \frac{\# \{ \tau(\eta_n) \in \tau(B_0) \} }{K} - \frac{ \operatorname{Vol}( \tau(B_0) ) }{ \operatorname{Vol}(\Theta) } \right| = \sup_{ B \in \mathcal{B} } \left| \frac{\# \{ \theta_n \in B \} }{K} - \frac{ \operatorname{Vol}( B ) }{ \operatorname{Vol}(\Theta) } \right|
	\end{align}
	where we used $\theta_n = \tau(\eta_n)$.
	This concludes the proof.
\end{proof}

\section{Proofs} \label{apx:proof}

This section contains proofs of all the theoretical results presented in the main text.

\subsection{Proof of \Cref{thm:convergence}} \label{apx:convergence}

We introduce two technical lemmas that will be useful.
The first lemma is on a decomposition of the MMD.
Recall that $\delta_N$ denotes the empirical distribution of the weighted samples $\{ w_n^*, \theta_n^* \}_{n=1}^{N}$, and $d(p, \delta_N)$ denotes the MMD between the target density $p$ and the empirical distribution $\delta_N$.

\begin{lemma} \label{lem:second}
	Suppose that the target density $p(\cdot) = q(\cdot) / Z$ and the proposal density $u$ are positive and continuous in Algorithm \ref{alg:tmp}. 
	Given an arbitrary proposal sequence $\{ \theta_n \}_{n=1}^{M+N}$, define a functional
	\begin{align}
		I_{M+N}(f) := \left| \int_\Theta f(\theta) \frac{ q(\theta) }{ u(\theta) } u(\theta) d \theta - \frac{1}{N+M} \sum_{n=1}^{N + M} f(\theta_n) \frac{ q(\theta_n) }{ u(\theta_n) } \right|
	\end{align}
	for functions $f$ on $\Theta$.
	Then the weighted samples $\{ w_n^*, \theta_n^* \}$ of Algorithm \ref{alg:tmp} satisfies
	\begin{align}
		d(p, \delta_N) \le B_1 \sup_{\| f \|_{\mathcal{H}_k} \le 1} I_{M+N}(f) + B_2 I_{M+N}(1) + B_3 \frac{M}{N+M}
	\end{align}
	for some constants $B_1, B_2, B_3 > 0$ dependent only on $p$ and $u$.
\end{lemma}

\begin{proof}
	Recall that the target density $p(\theta)$ is defined as $p(\theta) = q(\theta) / Z$.
	Using the proposal density $u(\theta)$, we rewrite the MMD of interest $d(p, \delta_N)$ as follows:
	\begin{align}
		d(p, \delta_N) = \sup_{\| f \|_{\H_\kappa} \le 1} \underbrace{ \left| \frac{ 1 }{ Z } \int_\Theta f(\theta) \frac{ q(\theta) }{ u(\theta) } u(\theta) d \theta - \sum_{i=1}^{N} w_n^* f(\theta_n^*) \right| }_{ =: (*) } \quad \text{where} \quad w_n^* = \frac{ \frac{1}{N} \frac{ q(\theta_n^*) }{ u(\theta_n^*) } }{ \frac{1}{N} \sum_{n=1}^{N} \frac{ q(\theta_n^*) }{ u(\theta_n^*) } } .
	\end{align} 
	For better presentation, define a function $h_f(\theta) := f(\theta) q(\theta) / u(\theta)$ for each $f$ to see that
	\begin{align}
		(*) = \left| \frac{ \int_\Theta h_f(\theta) u(\theta) d \theta }{ Z } - \frac{\frac{1}{N} \sum_{n=1}^{N} h_f(\theta_n^*) }{ \frac{1}{N} \sum_{n=1}^{N} \frac{ q(\theta_n^*) }{ u(\theta_n^*) } } \right|
	\end{align}
	In addition, define two constants $Z_{M+N}$ and $Z_N^*$ by, respectively,
	\begin{align}
		Z_{M+N} := \frac{1}{N+M} \sum_{n=1}^{M+N} \frac{ q(\theta_n) }{ u(\theta_n) } \qquad \text{and} \qquad Z_N^* := \frac{1}{N} \sum_{n=1}^{N} \frac{ q(\theta_n^*) }{ u(\theta_n^*) }  .
	\end{align} 
	By the triangle inequality, the term $(*)$ is upper-bounded as follows:
	\begin{align}
		(*) & \le \underbrace{ \left| \frac{ \int_\Theta h_f(\theta) u(\theta) d \theta }{Z} - \frac{ \frac{1}{N+M} \sum_{n=1}^{N + M} h_f(\theta_n) }{Z_{N+M}} \right| }_{ =: (*_1) } + \underbrace{ \left| \frac{\frac{1}{N+M} \sum_{n=1}^{N + M} h_f(\theta_n) }{Z_{N+M}} - \frac{\frac{1}{N} \sum_{n=1}^{N} h_f(\theta_n^*) }{ Z_N^* } \right| }_{ =: (*_2) }
	\end{align}
	where $\{ \theta_n \}_{n=1}^{N+M}$ is the proposal sequence.
	For the term $(*_1)$, by the triangle inequality again,
	\begin{align}
		(*_1) & \le \underbrace{ \left| \frac{ \int_\Theta h_f(\theta) u(\theta) d \theta }{ Z } - \frac{ \frac{1}{N+M} \sum_{n=1}^{N + M} h_f(\theta_n) }{ Z } \right| }_{ = \frac{1}{Z} I_{M+N}(f) =: (*_{11}) } + \underbrace{ \left| \frac{ \frac{1}{N+M} \sum_{n=1}^{N + M} h_f(\theta_n) }{ Z } - \frac{ \frac{1}{N+M} \sum_{n=1}^{N + M} h_f(\theta_n) }{Z_{M+N}} \right| }_{ =: (*_{12}) } .
	\end{align}
	We shall find a further upper bound of the latter term $(*_{12})$.
	The term $(*_{12})$ is upper bounded as
	\begin{align}
		(*_{12}) & = \left| Z - Z_{N+M} \right| \left| \frac{ \frac{1}{N+M} \sum_{n=1}^{N + M} h_f(\theta_n) }{ Z Z_{N+M} } \right| \le \frac{1}{Z} \left| Z - Z_{N+M} \right| \underbrace{ \frac{ \frac{1}{N+M} \sum_{n=1}^{N + M} | h_f(\theta_n) | }{ \frac{1}{N+M} \sum_{n=1}^{N + M} \frac{ q(\theta_n) }{ u(\theta_n) } } }_{ =: (\star) } .
	\end{align}
	Here, it follows from the definition of $h_f$ and \Cref{lem:rk_bound} that we have $| h_f(\theta) | \le q(\theta) / u(\theta)$, where no absolute value is needed in the right-hand side since $q$ and $u$ are positive.
	This immediately implies that the term $(\star)$ satisfies the upper bound $(\star) \le 1$.
	Therefore, we have
	\begin{align}
		(*_{12}) & \le \frac{1}{Z} \left| Z - Z_{N+M} \right| = \frac{1}{Z} \left| \int_\Theta \frac{ q(\theta) }{ u(\theta) } u(\theta) d\theta - \frac{1}{N+M} \sum_{n=1}^{M+N} \frac{ q(\theta_n) }{ u(\theta_n) } \right| = \frac{1}{Z} I_{M+N}(1) .
	\end{align}
	Finally, we shall find an upper bound of the term $(*_2)$.
	By the triangle inequality, 
	\begin{align}
		(*_2) & \le \underbrace{ \left| \frac{\frac{1}{N+M} \sum_{n=1}^{N + M} h_f(\theta_n) }{Z_{N+M}} - \frac{ \frac{1}{M + N} \sum_{n=1}^{N} h_f(\theta_n^*) }{Z_{N+M}} \right| }_{=:(*_{21})} + \underbrace{ \left| \frac{ \frac{1}{M + N} \sum_{n=1}^{N} h_f(\theta_n^*) }{Z_{N+M}} - \frac{ \frac{1}{N} \sum_{n=1}^{N} h_f(\theta_n^*) }{ Z_N^* } \right| }_{=:(*_{22})} .
	\end{align}
	Notice that, by construction, the proposal sequence $\{ \theta_n \}_{n=1}^{M + N}$ can be disjointly decomposed into the samples $\{ \theta_n^* \}_{n=1}^{N}$ obtained after the $N$-th iteration and the selection pool $\mathcal{S}_{N+1}$ at the $(N+1)$-th iteration.
	In other words, the proposal sequence $\{ \theta_n \}_{n=1}^{M + N}$ is a union of the $N$ samples $\{ \theta_n^* \}_{n=1}^{N}$ and the selection pool $\mathcal{S}_{N+1}$ at the $(N+1)$-th iteration. 
	Hence, the term $(*_{21})$ can be expressed as
	\begin{align}
		(*_{21}) & = \left| \frac{ \sum_{n=1}^{N + M} h_f(\theta_n) - \sum_{n=1}^{N} h_f(\theta_n^*) }{ (M+N) Z_{M+N} } \right| = \left| \frac{ \sum_{\theta \in \mathcal{S}_{N+1}} h_f(\theta) }{ \sum_{n=1}^{M+N} q(\theta_n) / u(\theta_n) } \right| \le \frac{ \sum_{\theta \in \mathcal{S}_{N+1}} q(\theta) / u(\theta) }{ \sum_{n=1}^{M+N} q(\theta_n) / u(\theta_n) }
	\end{align}
	where we used the fact that $| h_f(\theta) | \le q(\theta) / u(\theta)$ again.
	Since $q$ and $u$ are both positive and continuous, the ratio $q(\cdot) / u(\cdot)$ is also positive and continuous.
	Since $\Theta$ is compact, the extreme value theorem implies that $q(\cdot) / u(\cdot)$ admits some maximum value $U$ and minimum value $L$ over $\Theta$, where $U \ge L > 0$ due to the positivity. 
	Notice also that, by construction, the selection pool $\S_{N+1}$ always consists of $M$ points across the iterations.
	These two facts provide an upper bound $\sum_{\theta \in S_{N+1}} q(\theta) / u(\theta) \le U M$ and a lower bound $\sum_{n=1}^{M+N} q(\theta_n) / u(\theta_n) \ge L (M + N)$.
	Therefore
	\begin{align}
		(*_{21}) & \le \frac{U}{L} \frac{M}{M + N} .
	\end{align}
	We use the same argument for the term $(*_{22})$.
	The term $(*_{22})$ can be expressed as
	\begin{align}
		(*_{22}) & = \left| \frac{ \sum_{n=1}^{N} h_f(\theta_n^*) }{ (M + N) Z_{N+M} } - \frac{ \sum_{n=1}^{N} h_f(\theta_n^*) }{ N Z_N^* } \right| \\
		& = \frac{\left| (M + N) Z_{N+M} - N Z_N^* \right|}{(M+N) Z_{N+M}} \left| \frac{ \sum_{n=1}^{N} h_f(\theta_n^*) }{ N Z_n } \right| \\
		& = \frac{\left| \sum_{n=1}^{M+N} q(\theta_n) / u(\theta_n) - \sum_{n=1}^{N} q(\theta_n^*) / u(\theta_n^*) \right| }{ \sum_{n=1}^{M+N} q(\theta_n) / u(\theta_n) } \left| \frac{ \sum_{n=1}^{N} h_f(\theta_n^*) }{ \sum_{n=1}^{N} q(\theta_n^*) / u(\theta_n^*) } \right| .
	\end{align}
	Since the proposal sequence $\{ \theta_n \}_{n=1}^{M + N}$ can be disjointly decomposed into the samples $\{ \theta_n^* \}_{n=1}^{N}$ and the selection pool $\mathcal{S}_{N+1}$, the same argument used for the term $(*_{21})$ provides that
	\begin{align}
		(*_{22}) & \le \frac{ \sum_{\theta \in \S_{N+1}} q(\theta) / u(\theta) }{ \sum_{n=1}^{M+N} q(\theta_n) / u(\theta_n) } \le \frac{U}{L} \frac{M}{M + N} .
	\end{align}
	By all the arguments above, the MMD of interest $d(p, \delta_n)$ is bounded as
	\begin{align}
		d(p, \delta_n) & \le \sup_{\| f \|_{\H_\kappa} \le 1} (*_{11}) + \sup_{\| f \|_{\H_\kappa} \le 1} (*_{12}) + \sup_{\| f \|_{\H_\kappa} \le 1} (*_{21}) + \sup_{\| f \|_{\H_\kappa} \le 1} (*_{22}) \\
		& \le \frac{1}{Z} \sup_{\| f \|_{\H_\kappa} \le 1} I_{M+N}(f) + \frac{1}{Z} I_{M+N}(1) + 2 \frac{U}{L} \frac{M}{N + M} .
	\end{align}
	Setting $B_1 = 1 / Z$, $B_2 = 1 / Z$, and $B_3 = 2 U / L$ completes the proof.
\end{proof}

The second technical lemma is an extension of the celebrated Koksma-Hlawka inequality to the domain $\Theta$.
Recall that $\operatorname{D}^*$ is the star discrepancy introduced in \Cref{sec:theory} and $\partial_{1:d}$ denotes the mixed partial derivative.
Recall also the multi-index notation.

\begin{lemma} \label{lem:kh}
	Let $g$ be a function on $\Theta$ s.t.~the first mixed partial derivative $\partial_{1:d} g(\theta)$ exists.
	Let $\{ \theta_n \}_{n=1}^{K}$ be a sequence of $K$ arbitrary points in $\Theta$.
	We have
	\begin{align}
		 \left| \int_\Theta g(\theta) d \theta - \frac{\operatorname{Vol}(\Theta)}{K} \sum_{n=1}^{K} g(\theta_n) \right| \le \operatorname{Vol}(\Theta) \operatorname{V}( g ) \operatorname{D}^*( \{ \theta_n \}_{n=1}^{K} ) 
	\end{align}
	where $V(g)$ is the Hardy-Krause variation of the function $g$ defined as follows:
	\begin{align}
		V(g) := \sum_{\alpha \ne 0} \int_{\Theta^\alpha} \left| \partial^\alpha g(\theta_*^\alpha) \right| d \theta^\alpha 
	\end{align}
	where $\alpha \ne 0$ is the binary multi-index $\alpha$ whch is not zero vector $(0, \dots, 0)$.
\end{lemma}

\begin{proof}
	Since $\Theta$ is hyperrectangular, $\Theta$ can be expressed as the product $[a_1, b_1] \times \cdots \times [a_d, b_d]$ of some coordinatewise intervals $[ a_i, b_i ]$ for each $i = 1, \dots, d$.
	There exists a bijective affine map $\tau$ between the unit cube $[0, 1]^d$ and the domain $\Theta$.
	Any point $\theta$ in $\Theta$ can be expressed as the image of some point $\eta$ in $[0, 1]^d$ under the map $\tau$.
	For each $i = 1, \dots, d$, let $\tau_i: [0, 1] \to [a_i, b_i]$ be a map s.t.~$\theta_{(i)} = \tau_i( \eta_{(i)} ) = a_i + ( b_i - a_i ) \eta_{(i)}$, where $\theta_{(i)}$ and $\eta_{(i)}$ denote the $i$-th coordinate of $\theta$ and $\eta$.
	Then, clearly, the bijective map $\tau$ is given by $\theta = \tau(\eta) = ( \tau_1(\eta_{(1)}), \dots, \tau_d(\eta_{(d)}) )$.
	Let $\{ \eta_n \}_{n=1}^{K}$ be a sequence in $[0, 1]^d$ s.t.~the image under the map $\tau$ corresponds to the original sequence $\{ \theta_n \}_{n=1}^{K}$ in $\Theta$.
	This means that $\theta_n = \tau(\eta_n)$ for each $n = 1, \dots, K$.
	By the change of variables, we have
	\begin{align}
		\left| \int_\Theta g(\theta) d \theta - \frac{\operatorname{Vol}(\Theta)}{K} \sum_{n=1}^{K} g(\theta_n) \right| = \operatorname{Vol}(\Theta) \underbrace{ \left| \int_{[0, 1]^d} g \circ \tau(\eta) d \eta - \frac{1}{K} \sum_{n=1}^{K} g \circ \tau(\eta_n) \right| }_{ =: (*) }
	\end{align}
	where we used that $\prod_{i=1}^{d} | (\partial / \partial \eta_{(i)} ) \tau_i(\eta_{(i)}) | = \prod_{i=1}^{d} | b_i - a_i | = \operatorname{Vol}(\Theta)$ for the change of variables.
	Recall the star discrepancy $\operatorname{D}^\dagger$ for the unit cube $[0, 1]^d$ defined in \eqref{eq:star_discrepancy_0} in \Cref{apx:premilinary_sd}.
	It follows from the Koksma-Hlawka inequality \citep[][Proposition 2.18]{Dick2010} that
	\begin{align}
		(*) \le V^\dagger(g \circ \tau) \operatorname{D}^\dagger( \{ \eta_n \}_{n=1}^{K} ) . \label{eq:kh_bound_v}
	\end{align}
	where $V^\dagger(g \circ \tau)$ is the Hardy-Krause variation of the function $g \circ \tau$ defined by
	\begin{align}
		V^\dagger(g \circ \tau) := \sum_{\alpha \ne 0} \int_{[0, 1]^\alpha} \left| \partial^\alpha g \circ \tau(\eta_1^\alpha) \right| d \eta^\alpha . \label{eq:v_dagger}
	\end{align}
	Here $\eta_1^\alpha$ denotes a $d$-dimensional vector whose $i$-th coordinate is equal to that of $\eta$ if $\alpha_i = 1$ and otherwise is fixed to a constant $1$.
	We shall show that $V^\dagger(g \circ \tau) = V(g)$.
	Let $I_\alpha$ be a set of indices $i$ s.t.~$\alpha_i = 1$, by which the mixed partial derivative $\partial^\alpha$ can be expressed as $\partial^\alpha = \prod_{i \in I_\alpha} \partial_i$.
	By the chain rule and by that the map $\tau(\eta) = ( \tau_1(\eta_{(1)}), \dots, \tau_d(\eta_{(d)}) )$ is coordinatewise, we have
	\begin{align}
		\partial^\alpha g \circ \tau(\eta_1^\alpha) = \partial^\alpha g(\theta) \big|_{\theta = \tau(\eta_1^\alpha)} \left( \prod_{i \in I_\alpha} \partial_i \tau_i(\eta_{(i)}) \right) = \partial^\alpha g(\theta) \big|_{\theta = \tau(\eta_1^\alpha)} \cdot \operatorname{Vol}(\Theta^\alpha) \label{eq:g_tau_eq}
	\end{align}
	where we used $\prod_{i \in I_\alpha} \partial_i \tau_i(\eta_{(i)}) = \prod_{i \in I_\alpha} b_i - a_i = \operatorname{Vol}(\Theta^\alpha)$.
	Denote by $h(\theta)$ the mixed partial derivative $\partial^\alpha g(\theta)$ to see that $h(\tau(\eta_1^\alpha)) = \partial^\alpha g(\theta) |_{\theta = \tau(\eta_1^\alpha)}$.
	Combining \eqref{eq:v_dagger} and \eqref{eq:g_tau_eq} gives that
	\begin{align}
		V^\dagger(g \circ \tau) = \sum_{\alpha \ne 0} \int_{[0, 1]^\alpha} \operatorname{Vol}(\Theta^\alpha) \left| h(\tau(\eta_1^\alpha)) \right| d \eta^\alpha .
	\end{align}
	Now recall that $\theta_*^\alpha$ is a $d$-dimensional vector whose $i$-th coordinate is equal to that of $\theta$ if $\alpha_i = 1$ and otherwise is fixed to the constant $b_i$ that is the right boundary of the hyperrectangular domain $\Theta$ at the $i$-th axis.
	We can express $\theta_*^\alpha$ as $\theta_*^\alpha = \tau(\eta^\alpha_1)$ given that $\theta = \tau(\eta)$.
	We apply the change of variables to $\eta^\alpha$ to show that $V^\dagger(g \circ \tau) = V(g)$ as intended.
	Specifically, we apply the transform $\theta_{(i)} = \tau_i( \eta_{(i)} )$ for each coordinate $i$ s.t.~$\alpha_i = 1$, to see that
	\begin{align}
		V^\dagger(g \circ \tau) = \sum_{\alpha \ne 0} \int_{\Theta^\alpha} \operatorname{Vol}(\Theta^\alpha) \left| h(\tau(\eta_1^\alpha)) \right| d \eta^\alpha = \sum_{\alpha \ne 0} \int_{\Theta^\alpha} \left| \partial^\alpha g(\theta_*^\alpha) \right| d \theta^\alpha = V(g) .
	\end{align}
	where $\operatorname{Vol}(\Theta^\alpha)$ is canceled by the change of variables.
	Therefore, the inequality \eqref{eq:kh_bound_v} is rewritten as
	\begin{align}
		(*) \le V(g) \operatorname{D}^\dagger( \{ \eta_n \}_{n=1}^{K} )
	\end{align}
	By \Cref{lem:sd}, we have $\operatorname{D}^\dagger\left( \{ \eta_n \}_{n=1}^{K} \right) = \operatorname{D}^*\left( \{ \theta_n \}_{n=1}^{K} \right)$, which concludes the proof.
\end{proof}

Now we present the main proof:

\begin{proof}
	First, we apply the upper bound of \Cref{lem:second} to see that
	\begin{align}
		d(p, \delta_N) \le B_1 \underbrace{ \sup_{\| f \|_{\mathcal{H}_k} \le 1} I_{M+N}(f) }_{ =:(*_1) } + B_2 \underbrace{\vphantom{\sup_{\| f \|_{\mathcal{H}_k} \le 1}} I_{M+N}(1) }_{ =:(*_2) } + B_3 \frac{M}{N+M} .
	\end{align}
	We shall upper bound the term $(*_1)$ and $(*_2)$ in what follows.
	
	\vspace{5pt}
	\noindent
	\textbf{Upper Bounding $(*_1)$:}
	Since $u$ is uniform, the term $(*_1)$ can be bounded as follows:
	\begin{align}
		(*_1) & = \sup_{\| f \|_{\H_\kappa \le 1}} \left| \int_\Theta f(\theta) q(\theta) d \theta - \frac{\operatorname{Vol}(\Theta)}{N+M} \sum_{n=1}^{N + M} f(\theta_n) q(\theta_n) \right| \le \operatorname{Vol}(\Theta) \underbrace{ \sup_{\| f \|_{\H_\kappa \le 1}} V( f \cdot q ) }_{ =: (*_{11}) } \operatorname{D}^*( \{ \theta_n \}_{n=1}^{M+N} ) 
	\end{align}
	where the inequality follows from \Cref{lem:kh}.
	We shall upper bound the term $(*_{11})$, i.e., the supremum of the Hardy-Krause variation $V(f \cdot q)$.
	It is easy to see from the definition of $V$ that
	\begin{align}
		V( f \cdot q ) & \le \sum_{\alpha \ne 0} \operatorname{Vol}(\Theta^\alpha) \sup_{\theta \in \Theta^\alpha} \left| \partial^\alpha ( f(\theta_*^\alpha) q(\theta_*^\alpha) ) \right| \le \sum_{\alpha \ne 0} \operatorname{Vol}(\Theta^\alpha) \underbrace{ \sup_{\theta \in \Theta} \left| \partial^\alpha ( f(\theta) q(\theta) ) \right| }_{ =: (*_{12}) } .
	\end{align}
	Similarly to the product rule of derivatives, mixed partial derivatives can be extended by the general Leibniz rule.
	The general Leibniz rule slightly simplifies for binary multi-indices, in which case the mixed partial derivative of the product in the term $(*_{12})$ can be expressed as follows:
	\begin{align}
		\partial^\alpha ( f(\theta) q(\theta) ) = \sum_{\beta \le \alpha} ( \partial^{\alpha - \beta} f(\theta) ) ( \partial^{\beta} q(\theta) )
	\end{align}
	where $\beta \le \alpha$ means a binary multi-index $\beta$ s.t.~$\beta_i \le \alpha_i$ for all $i = 1, \dots, d$.
	By \Cref{lem:deriv}, we have $\sup_{\theta \in \Theta} \left| \partial^{\alpha-\beta} f(\theta) \right| \le 1$ for all $\alpha - \beta$ and all $\| f \|_{\H_\kappa} \le 1$.
	By the assumption that $\partial_{1:d} q(\theta)$ is uniformly continuous, $\partial^\beta q(\theta)$ is uniformly continuous for any binary multi-index $\beta$.
	By the continuous extension theorem, $\partial^\beta q(\theta)$ can be continuously extended to the boundary of the domain $\Theta$.
	This means that, by the extreme value theorem, there exists some constant $C_q^\beta$ s.t.~$\sup_{\theta \in \Theta} | \partial^\beta q(\theta) | \le C_q^\beta$ for every binary multi-index $\beta$.
	By these two facts, we have
	\begin{align}
		(*_{12}) \le \sup_{\theta \in \Theta} \sum_{\beta \le \alpha} \left| ( \partial^{\alpha - \beta} f(\theta) ) ( \partial^{\beta} q(\theta) ) \right| \le \sum_{\beta \le \alpha} C_q^\beta .
	\end{align}
	Since this bound of $(*_{12})$ holds for all $\| f \|_{\H_\kappa} \le 1$, the original term $(*_{11})$ is bounded as
	\begin{align}
		(*_{11}) \le \sum_{\alpha \ne 0} \operatorname{Vol}(\Theta^\alpha) \sum_{\beta \le \alpha} C_q^\beta =: A_1
	\end{align}
	where we denote the constant in the right-hand side by $A_1$.
	
	\vspace{5pt}
	\noindent
	\textbf{Upper Bounding $(*_2)$:}
	Since $u$ is uniform, the term $(*_2)$ can be bounded as follows:
	\begin{align}
		(*_2) = \left| \int_\Theta q(\theta) d \theta - \frac{\operatorname{Vol}(\Theta)}{N+M} \sum_{n=1}^{N + M} q(\theta_n) \right| \le \operatorname{Vol}(\Theta) V( q ) \operatorname{D}^*( \{ \theta_n \}_{n=1}^{M+N} ) 
	\end{align}
	where the inequality follows from \Cref{lem:kh}.
	It is easy to see from the definition of $V$ that
	\begin{align}
		V( q ) & \le \sum_{\alpha \ne 0} \operatorname{Vol}(\Theta^\alpha) \sup_{\theta \in \Theta^\alpha} \left| \partial^\alpha q(\theta_*^\alpha) \right|  \le \sum_{\alpha \ne 0} \operatorname{Vol}(\Theta^\alpha) \sup_{\theta \in \Theta} \left| \partial^\alpha q(\theta) \right| .
	\end{align}
	By the same argument as that used for upper-bounding the first term $(*_1)$, there exists some constant $C_q^\alpha$ s.t.~$\sup_{\theta \in \Theta} | \partial^\alpha q(\theta) | \le C_q^\alpha $ for every binary multi-index $\alpha$.
	By this facts, we have
	\begin{align}
		V( q ) & \le \sum_{\alpha \ne 0} \operatorname{Vol}(\Theta^\alpha) C_q^\alpha =: A_2
	\end{align}
	where we denote the constant in the right-hand side by $A_2$.
	
	\vspace{5pt}
	\noindent
	\textbf{Overall Bound:}
	Plugging the derived upper bounds of $(*_1)$ and $(*_2)$ in, we have 
	\begin{align}
		(*) & \le \operatorname{Vol}(\Theta) A_1 B_1 \operatorname{D}^*( \{ \theta_n \}_{n=1}^{M+N} ) + \operatorname{Vol}(\Theta) A_2 B_2 \operatorname{D}^*( \{ \theta_n \}_{n=1}^{M+N} ) + B_3 \frac{M}{N+M} .
	\end{align}
	Setting $C_1 = \operatorname{Vol}(\Theta) A_1 B_1 +  \operatorname{Vol}(\Theta) A_2 B_2$ and $C_2 = B_3$ concludes the proof.
\end{proof}

\subsection{Proof of \Cref{col:convergence_Halton}} \label{apx:proof_convergence_Halton}

\begin{proof}
	Recall that $\operatorname{D}^\dagger$ is the star discrepancy for the unit cube $[0, 1]^d$ defined in \eqref{eq:star_discrepancy_0} in \Cref{apx:premilinary_sd}.
	The scaled Halton sequence $\{ \theta_i \}_{i=1}^{\infty}$ is defined as an image of the original Halton sequence $\{ \eta_i \}_{i=1}^{\infty}$ under the bijective affine map.
	By \Cref{lem:sd}, we have $\operatorname{D}^*\left( \{ \theta_n \}_{n=1}^{K} \right) = \operatorname{D}^\dagger\left( \{ \eta_n \}_{n=1}^{K} \right)$ for any number of points $K$.
	By Theorem 3.6 of \cite{Niederreiter1992}, there exists some constant $C > 0$ s.t.~the star discrepancy $\operatorname{D}^\dagger$ of the original Halton sequence $\{ \eta_n \}_{n=1}^{K}$ satisfies that
	\begin{align}
		\operatorname{D}^\dagger\left( \{ \eta_n \}_{n=1}^{K} \right) \le C \frac{ \log(K)^d }{ K } .
	\end{align}
	Combining this bound with the result of \Cref{thm:convergence} concludes the proof.
\end{proof}

\subsection{Proof of \Cref{col:weak_convergence}} \label{apx:proof_weak_convergence}

\begin{proof}
	By assumption, the kernel $\kappa$ is bounded, continuous, and integrally strictly positive definite.
	Hence, the resulting MMD metrizes the weak convergence by Theorem 7 of \cite{Simon-Gabriel2023}.
	By \Cref{col:convergence_Halton}, the MMD between $\delta_N$ and $p$ converges to zero.
	The consistency of the Monte Carlo estimator immediately follows from the definition of the weak convergence.
	Here, since the domain $\Theta$ is compact, the continuous integrand $f$ is bounded by extreme value theorem.
\end{proof}

\subsection{Proof of \Cref{thm:concentration}} \label{apx:concetration}

The proof uses \Cref{lem:second} presented in \Cref{apx:convergence}, which holds for any proposal sequence.

\begin{proof}
	First, we apply the upper bound of \Cref{lem:second} to see
	\begin{align}
		d(p, \delta_N) \le B_1 \underbrace{ \sup_{\| f \|_{\mathcal{H}_k} \le 1} I_{M+N}(f) }_{ =:(*_1) } + B_2 \underbrace{\vphantom{\sup_{\| f \|_{\mathcal{H}_k} \le 1}} I_{M+N}(1) }_{ =:(*_2) } + B_3 \frac{M}{N+M} .
	\end{align}
	Define $g(\theta) := q(\theta) / u(\theta)$ to see that
	\begin{align}
		I_{M+N}(f) & = \left| \int_\Theta f(\theta) g(\theta) u(\theta) d \theta - \frac{1}{N+M} \sum_{n=1}^{N + M} f(\theta_n) g(\theta_n) \right| .
	\end{align}
	We shall upper bound the term $(*_1)$ and $(*_2)$ in what follows.
	Denote by $u_{M+N}$ the empirical distribution of the proposal sequence $\{ \theta_n \}_{n=1}^{M+N}$ for better presentation.
	
	\vspace{5pt}
	\noindent
	\textbf{Upper Bounding $(*_1)$:}
	Let $\mathcal{H}$ be a set of functions s.t.~$\mathcal{H} := \{ \theta \mapsto f(\theta) g(\theta) \mid \| f \|_{\H_\kappa} \le 1 \} $ for the function $g$.
	Then, the term $(*_1)$ can be expressed as
	\begin{align}
		(*_1) & = \sup_{\| f \|_{\H_\kappa} \le 1} \left| \E_{\theta \sim u}\left[ f(\theta) g(\theta) \right] - \E_{\theta \sim u_{M+N}}\left[ f(\theta) g(\theta) \right] \right| = \sup_{h \in \mathcal{H}} \left| \E_{\theta \sim u}\left[ h(\theta) \right] - \E_{\theta \sim u_{M+N}}\left[ h(\theta) \right] \right| .
	\end{align}
	Recall that the proposal sequence $\{ \theta_n \}_{n=1}^{M+N}$ is i.i.d.~random samples from the proposal density $u$.
	Let $\{ \epsilon_n \}_{n=1}^{M+N}$ be i.i.d.~random variables taking the value in $\{ +1, -1 \}$ with equiprobability $1 / 2$.
	The Rademacher complexity $\mathcal{R}_{M+N}( \mathcal{H} )$ of the set $\H$ given the random variables $\{ \theta_n \}_{n=1}^{M+N}$ is defined by
	\begin{align}
		\mathcal{R}_{M+N}( \mathcal{H} ) & := \E_{\theta_1, \dots, \theta_{M+N}} \E_{\epsilon_1, \dots, \epsilon_{M+N}} \left[ \sup_{h \in \H} \left| \frac{1}{M+N} \sum_{n=1}^{M+N} \epsilon_n h(\theta_n) \right| \right] \\
		& = \E_{\theta_1, \dots, \theta_{M+N}} \E_{\epsilon_1, \dots, \epsilon_{M+N}} \left[ \sup_{\| f \|_{\H_\kappa} \le 1} \left| \frac{1}{M+N} \sum_{n=1}^{M+N} \epsilon_n f(\theta_n) g(\theta_n) \right| \right] .
	\end{align}
	\Cref{lem:rk_bound} implies that $| f(\theta) | \le 1$ uniformly for all functions $f$ s.t.~$\| f \|_{\H_\kappa} \le 1$.
	It is thus straightforward to see that any function $h \in \H$ is uniformly bounded by $b := \sup_{\theta \in \Theta} | g(\theta) |$ because 
	\begin{align}
		\sup_{\theta \in \Theta}| h(\theta) | \le \sup_{\theta \in \Theta} | f(\theta) | \sup_{\theta \in \Theta}| g(\theta) | \le \sup_{\theta \in \Theta} | g(\theta) | = b .
	\end{align}
	It then follows from Theorem 4.10 of \cite{Wainwright2019} that
	\begin{align}
		\mathbb{P}\left( (*_1) \le 2 \mathcal{R}_{M+N}( \mathcal{H} ) + \delta \right) \ge 1 - \exp\left( - \frac{ (M+N) \delta^2 }{ 2 b^2 } \right) \label{eq:bounded_concentration}
	\end{align}
	where the probability $\mathbb{P}$ is taken with respect to the i.i.d.~random samples $\{ \theta_n \}_{n=1}^{M+N}$.
	We shall upper bound the Rademacher complexity next.
	We apply the reproducing property of $f \in \H_\kappa$ to see
	\begin{align}
		\mathcal{R}_{M+N}( \mathcal{H} ) & = \E_{\theta_1, \dots, \theta_{M+N}} \E_{\epsilon_1, \dots, \epsilon_{M+N}} \left[ \sup_{ \| f \|_{\H_\kappa} \le 1 } \left| \frac{1}{M+N} \sum_{n=1}^{M+N} \epsilon_n \langle f(\cdot), \kappa(\cdot, \theta_n) \rangle_{\H_\kappa} g(\theta_n) \right| \right] \\
		& = \frac{1}{M+N} \E_{\theta_1, \dots, \theta_{M+N}} \E_{\epsilon_1, \dots, \epsilon_{M+N}} \left[ \sup_{ \| f \|_{\H_\kappa} \le 1 } \left| \left\langle f(\cdot), \sum_{n=1}^{M+N} \epsilon_n \kappa(\cdot, \theta_n) g(\theta_n) \right\rangle_{\H_\kappa} \right| \right] .
	\end{align}
	We apply the fact that $\sup_{ \| f \|_{\H_\kappa} \le 1 } | \langle f(\cdot), v(\theta) \rangle_{\H_k} | = \| v(\theta) \|_{\H_k}$ for any $v \in \H_\kappa$, which follows from the condition of the Cauchy-Schwartz inequality under which equality holds, to see that
	\begin{align}
		\mathcal{R}_{M+N}( \mathcal{H} ) & = \frac{1}{M+N} \E_{\theta_1, \dots, \theta_{M+N}} \E_{\epsilon_1, \dots, \epsilon_{M+N}} \left[ \left\| \sum_{n=1}^{M+N} \epsilon_n \kappa(\cdot, \theta_n) g(\theta_n) \right\|_{\H_\kappa} \right] \\
		& \le \frac{1}{M+N} \sqrt{ \E_{\theta_1, \dots, \theta_{M+N}} \E_{\epsilon_1, \dots, \epsilon_{M+N}} \left[ \left\| \sum_{n=1}^{M+N} \epsilon_n \kappa(\cdot, \theta_n) g(\theta_n) \right\|_{\H_\kappa}^2 \right] } \\
		& = \frac{1}{M+N} \sqrt{ \E_{\theta_1, \dots, \theta_{M+N}} \E_{\epsilon_1, \dots, \epsilon_{M+N}} \left[ \sum_{n=1}^{M+N} \sum_{m=1}^{M+N} \epsilon_n \epsilon_m g(\theta_n) \kappa(\theta_n, \theta_m) g(\theta_m) \right] } .
	\end{align}
	In the above, the second inequality follows from the Jensen's inequality and the last equality follows from the reproducing property of the kernel $\kappa$.
	Note that the random variables $\epsilon_n$ and $\epsilon_m$ are independent when $n \ne m$, and that $\epsilon_n \epsilon_n = \epsilon_n^2 = 1$ with probability $1$ by definition.
	Therefore
	\begin{align}
		\mathcal{R}_{M+N}( \mathcal{H} ) & \le \frac{1}{M+N} \sqrt{ \E_{\theta_1, \dots, \theta_{M+N}} \left[ \sum_{n=1}^{M+N} g(\theta_n) \kappa(\theta_n, \theta_n) g(\theta_n) \right] } \\
		& \le \frac{1}{M+N} \sqrt{ (M + N) b^2 } \le \frac{b}{\sqrt{M + N}} 
	\end{align}
	where the second inequality follows from that $\kappa$ is bounded by $1$ and that $g(\theta) \le \sup_{\theta \in \Theta} | g(\theta) | = b$.
	Plugging this upper bound of $\mathcal{R}_{M+N}( \mathcal{H} )$ in the concentration inequality \eqref{eq:bounded_concentration} gives that
	\begin{align}
		\mathbb{P}\left( (*_1) \le \frac{2 b}{\sqrt{M + N}} + \delta \right) \ge 1 - \exp\left( - \frac{ (M+N) \delta^2 }{ 2 b^2 } \right) \ge 1 - 2 \exp\left( - \frac{ (M+N) \delta^2 }{ 2 b^2 } \right)
	\end{align}
	where the trivial inequality in the right-hand side is applied for convenience of subsequent analysis.
	This concentration inequality is equivalent to the following alternative expression
	\begin{align}
		\mathbb{P}\left( (*_1) \le \frac{2 b + \sqrt{2 \log(2 / \epsilon)} b }{\sqrt{M + N}} \right) \ge 1 - \epsilon 
	\end{align}
	where we set $\epsilon := 2 \exp( - (M + N) \delta^2 / (2 b^2) ) > 0$.
	
	\vspace{5pt}
	\noindent
	\textbf{Upper Bounding $(*_2)$:}
	Define a function $F(\theta_1, \dots, \theta_{M+N}) := ( 1 / ( M+N ) ) \sum_{n=1}^{N + M} g(\theta_n)$.
	Since the proposal sequence $\{ \theta_n \}_{n=1}^{M+N}$ is i.i.d.~random variables, the term $(*_2)$ can be expressed as
	\begin{align}
		(*_2) & = \left| \E_{\theta \sim u}\left[ g(\theta) \right] - \E_{\theta \sim u_{M+N}}\left[ g(\theta) \right] \right| = \left| \E_{\theta_1, \dots, \theta_{M+N} \overset{i.i.d.}{\sim} u}[ F(\theta_1, \dots, \theta_{M+N}) ] - F(\theta_1, \dots, \theta_{M+N}) \right| .
	\end{align}
	Our aim is to apply the bounded difference inequality in Corollary 2.21 of \cite{Wainwright2019}.
	The function $F$ admits the `bounded difference' with respect to each $n$-th argument $\theta_n$.
	To see this, fix all arguments $\theta_1, \dots, \theta_{M+N}$, except the $n$-th argument $\theta_n$, and observe for any $n$ that
	\begin{align}
		\left| F(\theta_1, \dots, \theta_n, \dots, \theta_{M+N}) - F(\theta_1, \dots, \theta_n', \dots, \theta_{M+N}) \right| = \frac{1}{M+N} \left| g(\theta_n) - g(\theta_n') \right| \le \frac{2 b}{M+N}
	\end{align}
	where $\theta_n$ and $\theta_n'$ are arbitrary.
	Thus we can apply Corollary 2.21 of \cite{Wainwright2019} to see that
	\begin{align}
		\mathbb{P}\left( (*_2) \le \delta \right) \ge 1 - 2 \exp\left( - \frac{(M+N) \delta^2}{ 2 b^2 } \right)
	\end{align}
	where $\mathbb{P}$ is the probability taken with respect to the i.i.d.~random variables $\{ \theta_n \}_{n=1}^{M+N}$.
	This concentration inequality is equivalent to the following alternative expression
	\begin{align}
		\mathbb{P}\left( (*_2) \le \frac{\sqrt{2 \log(2 / \epsilon)} b }{\sqrt{M + N}} \right) \ge 1 - \epsilon .
	\end{align}
	where we set $\epsilon := 2 \exp( - (M + N) \delta^2 / (2 b^2) ) > 0$.
	
	\vspace{5pt}
	\noindent
	\textbf{Overall Bound:}
	Combining the concentration inequalities of $(*_1)$ and $(*_2)$, we have
	\begin{align}
		\mathbb{P}\left( d(p, \delta_N) \le B_1 \frac{2 b + \sqrt{2 \log(2 / \epsilon)} b }{\sqrt{M + N}} + B_2 \frac{\sqrt{2 \log(2 / \epsilon)} b }{\sqrt{M + N}} + B_3 \frac{M}{N+M} \right) \ge 1 - \epsilon
	\end{align}
	Setting $C_1 = 2 b B_1$, $C_2 = 2 \sqrt{2} b$, and $C_3 = B_3$ concludes the proof.
\end{proof}

\subsection{Proof of \Cref{col:weak_convergence_as}} \label{apx:proof_weak_convergence_as}

\begin{proof}
	As shown in the proof of \Cref{col:weak_convergence}, the MMD of the kernel $\kappa$ metrizes the weak convergence.
	With no loss of generality, suppose that $\epsilon$ is small enough that $\sqrt{\log(2 / \epsilon)} > 1$.
	Let $C$ be a constant larger than $C_1$, $C_2$, and $C_3 M$.
	Since $N \ge 1$ and $M \ge 1$, we have
	\begin{align}
		d(p, \delta_N) \le \frac{C_1 + C_2 \sqrt{\log(2 / \epsilon)}}{\sqrt{N + M}} + C_3 \frac{M}{N + M} \le \frac{C \sqrt{\log(2 / \epsilon)}}{\sqrt{N + M}} =: \delta
	\end{align}
	Thus we have a simplified concentration inequality $\mathbb{P}( d(p, \delta_N) \le \delta) \ge 1 - \epsilon$, which is equivalent to $\mathbb{P}( d(p, \delta_N) > \delta) \le \epsilon$.
	This concentration inequality can be expressed in terms of $\delta$ as follows:
	\begin{align}
		\mathbb{P}\left( d(p, \delta_N) > \delta \right) & \le 2 \exp\left( - \frac{\delta^2}{2 C^2} ( N + M ) \right) .
	\end{align}
	It then follows from the Borel-Cantelli lemma that $d(p, \delta_N)$ converges to zero almost surely in the limit $N \to 0$, meaning that $\delta_N$ weakly converges to $p$ almost surely.
	The consistency of the Monte Carlo estimator immediately follows from the definition of the weak convergence.
	Here, since the domain $\Theta$ is compact, the continuous integrand $f$ is bounded by extreme value theorem.
\end{proof}

\subsection{Proof of \Cref{prop:gp_ujb_bound}} \label{apx:proof_gp_ujb_bound}

\begin{proof}
	Let $(\star) := \int_\Theta | U_n(\theta) - q(\theta) |^2 d \theta$.
	By the definition of $U_n$ and $q$, the term $(\star)$ is written as
	\begin{align}
		(\star) & = \int_\Theta \Big| \E_{f \sim \mathcal{GP}(m_n, k_n)}\left[ \phi\left( f(\theta) \right) \right] - \phi\left( g(\theta) \right) \Big|^2 d \theta = \int_\Theta \Big| \E_{f \sim \mathcal{GP}(m_n, k_n)}\left[ \phi\left( f(\theta) \right) - \phi\left( g(\theta) \right) \right] \Big|^2 d \theta.
	\end{align}
	It then follows from the Jensen's inequality that
	\begin{align}
		(\star) & \le \int_\Theta \bigg( \underbrace{ \E_{f \sim \mathcal{GP}(m_n, k_n)}\Big[ \left| \phi\left( f(\theta) \right) - \phi\left( g(\theta) \right) \right| \Big] }_{ =: (*) } \bigg)^2 d \theta .
	\end{align}
	We shall upper bound the term $(*)$ at each $\theta$.
	The derivative $\phi'$ of $\phi$ is monotonically non-decreasing because $\phi$ is a convex function.
	By the mean value inequality, it holds for all $a, b \in \R$ that
	\begin{align}
		| \phi(a) - \phi(b) | \le \bigg( \sup_{c \in [a, b]} | \phi'( c ) | \bigg) | b - a | \le \max\left( | \phi'(a) |, | \phi'(b) | \right) | b - a |
	\end{align}
	where the second inequality follows from monotonicity of $\phi'$.
	It thus holds at each $\theta$ that
	\begin{align}
		(*) & \le \E_{f \sim \mathcal{GP}(m_n, k_n)}\Big[ \max\left( | \phi'(f(\theta)) |, | \phi'(g(\theta)) | \right) \left| f(\theta) - g(\theta) \right| \Big] \\
		& \le \E_{f \sim \mathcal{GP}(m_n, k_n)}\Big[ \underbrace{ \max\left( \exp(f(\theta)), \exp(g(\theta)) \right) }_{=: (**)} \left| f(\theta) - g(\theta) \right| \Big] 
	\end{align}
	where we used the assumption $| \phi'(\cdot) | \le \exp(\cdot)$ for the last inequality.
	It then follows from the Cauchy-Schwartz inequality applied for the expectation that
	\begin{align}
		(*) \le \left( \E_{f \sim \mathcal{GP}(m_n, k_n)}\left[ (**)^2 \right] \E_{f \sim \mathcal{GP}(m_n, k_n)}\left[ \left| f(\theta) - g(\theta) \right|^2 \right] \right)^{\frac{1}{2}} . \label{eq:gp_exp_ineq_1}
	\end{align}
	Observe the following trivial inequality
	\begin{align}
		(**)^2 = \max\left( \exp(2 f(\theta)), \exp(2 g(\theta)) \right) \le \exp(2 f(\theta)) + \exp(2 g(\theta)) .
	\end{align}
	Notice that, since $f$ is a sample path of the Gaussian process, the point evaluation $f(\theta)$ at each $\theta$ follows the Gaussian distribution $\mathcal{N}(m_n(\theta), k_n(\theta, \theta))$.
	It thus admits the following analytical expectation $\E_{f \sim \mathcal{GP}(m_n, k_n)}\left[ \exp(2 f(\theta)) \right] = \exp( 2 m_n(\theta) + 2 k_n(\theta, \theta) )$.
	This, in turn, implies that
	\begin{align}
		\E_{f \sim \mathcal{GP}(m_n, k_n)}\left[ (**)^2 \right] & \le \exp( 2 m_n(\theta) + 2 k_n(\theta, \theta) ) + \exp(2 g(\theta)) \\
		& \le \underbrace{ \sup_{\theta \in \Theta} \exp( m_n(\theta) + k_n(\theta, \theta) )^2 + \sup_{\theta \in \Theta} \exp( g(\theta) )^2 }_{ =: C } . \label{eq:gp_exp_ineq_2}
	\end{align}
	Combining the bounds \eqref{eq:gp_exp_ineq_1} and \eqref{eq:gp_exp_ineq_2}, we have the overall bound of the original term $(\star)$:
	\begin{align}
		(\star) \le C \int_\Theta \E_{f \sim \mathcal{GP}(m_n, k_n)}\left[ \left| f(\theta) - g(\theta) \right|^2 \right] d \theta = C \E_{f \sim \mathcal{GP}(m_n, k_n)}\left[ \int_\Theta \left| f(\theta) - g(\theta) \right|^2 d \theta \right] 
	\end{align}
	where the order the integral and expectation can be interchanged by Fubini-Tonelli theorem, whenever the right-hand side is finite.
	Finally, it is straightforward to see from the definition of $C$ that we have $C < 1$ if $m_n(\theta) + k_n(\theta, \theta) < (1/2) \log(1/2)$ and $g(\theta) < (1/2) \log(1/2)$.
	This concludes the proof because $-1/2 < \log(\sqrt{1/2})$ which is immediate to verify.
\end{proof}

\section{Simulation Studies} \label{apx:simulation}

This section presents additional simulation studies on BIS.
\Cref{apx:discrete} discusses the importance of the point-selection procedure of BIS in comparison with the randomized BO.
\Cref{apx:choice_phi} demonstrates the performance of BIS for different choices of the function $\phi$.
\Cref{apx:pool_size} illustrates the performance of BIS for different choices of the size of the proposal sequence $M$.

\subsection{Importance of Discrete Optimization with No Revisit} \label{apx:discrete}

In BO, the point at each iteration is selected by the optimization over an infinite set, such as a compact subset of $\R^d$.
In contrast, BIS select the point by the optimization over a discrete set under the constraint that points selected in the past will never be revisited.
The optimization on the discrete set, together with the constraint, assists in efficient exploration of the domain of interest.
It prevents cases where the point to be selected next will be arbitrarily close to the points already selected in the past.
In the long run, therefore, the points selected in BIS do not concentrate on a particular subset of the domain.
This is vital for the convergence guarantee of the importance sampling scheme.

The randomization approach proposed in \cite{Kim2024} extends BO with the aims of approximating posterior densities.
They proposed performing BO whlie adding an extra random point in the training dataset of GP at each iteration.
This means that at each iteration it obtains two points, where one of them is the point acquired by BO and the other is a random point.
They showed the convergence of the GP posterior mean to a target posterior, when the GP is well specified.
However, we shall demonstrate that importance sampling based on the points selected by such a randomized BO fails to recover the target posterior.

\Cref{fig:fig_s31_01} visualizes 100 points obtained by the randomized BO of \cite{Kim2024}, for each benchmark density in \Cref{sec:benchmark}.
At each iteration of the randomized BO, a point from the scaled Halton sequence is used as the random point to add.
Following the experimental setting in \Cref{sec:benchmark}, the first 10 points were randomly selected from the Halton sequence.
It illustrates that the points selected by BO (triangular) gradually concentrate to the maxima of the posterior density, while the random points from the Halton sequence (square) are dispersed over the domain.
\Cref{fig:fig_s31_02} further confirms that importance sampling based on these samples fails to converge to the target posterior.
For each target density, the approximation error of the resulting weighted samples deteriorated after a certain number of iterations. 

The failure in the convergence roots in that points selected by BO concentrates at the mode of the target posterior after some iterations.
Thus, after some iterations, the majority of the samples used for importance sampling are enforced to have one particular value, which is the mode.
Moreover, those majority points sharing the same value are assigned the highest importance weights, since the mode is the point where the posterior probability is maximal by definition.
Therefore, the influence of those majority points concentrating at the mode will be dominant.
Roughly speaking, this means that, after some iterations, the weighted samples in the randomized BO starts approximating the Dirac distribution at the mode rather than the target density.
On the other hand, BIS, by design, prevents such a concentration of samples.

\begin{figure}[t]
	\subcaptionbox{gaussian}{\includegraphics[width=0.31\textwidth]{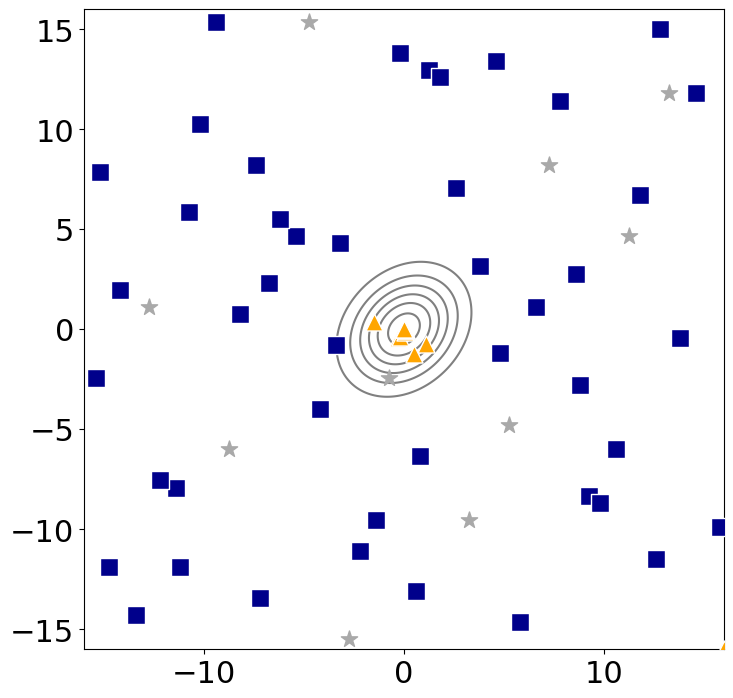}}
	\hfill
	\subcaptionbox{bimodal}{\includegraphics[width=0.3025\textwidth]{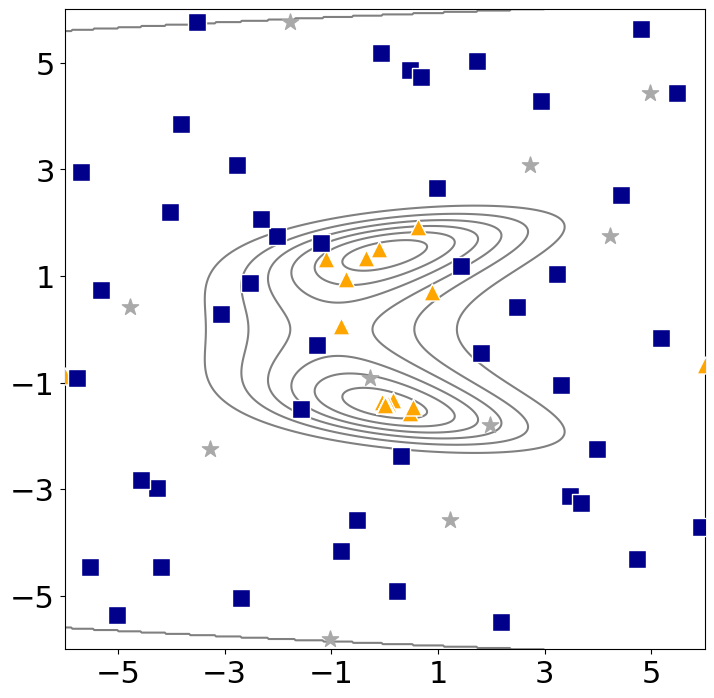}}
	\hfill
	\subcaptionbox{banana}{\includegraphics[width=0.31\textwidth]{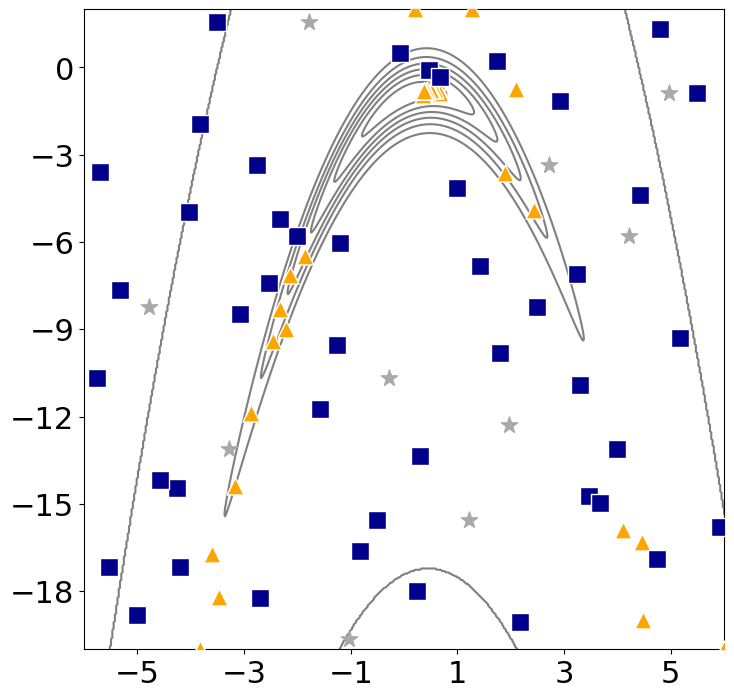}}
	\hfill
	\caption{Visualization of 100 samples obtained by the randomized BO of \cite{Kim2024} for each benchmark density. The initial 10 points are star-shaped. The following points obtained by BO and QMC are illustrated in, respectively, triangle- and square-shaped. The contour values of each density were powered to $1 / 3$ for better visualization of the tail geometry.} \label{fig:fig_s31_01}
\end{figure}

\begin{figure}[t]
	\subcaptionbox{gaussian}{\includegraphics[width=0.32\textwidth]{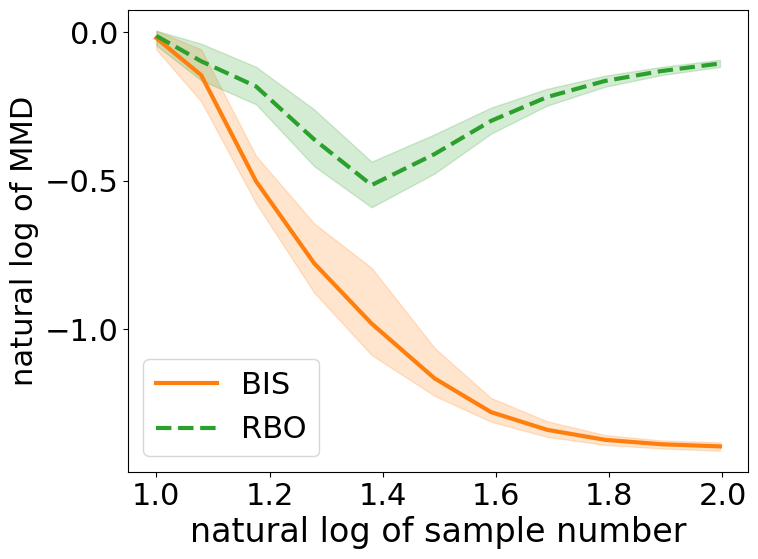}}
	\hfill
	\subcaptionbox{bimodal}{\includegraphics[width=0.32\textwidth]{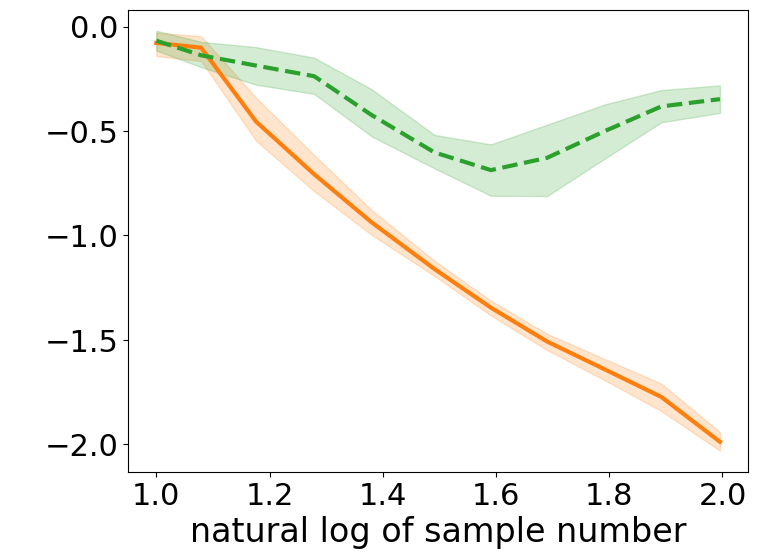}}
	\hfill
	\subcaptionbox{banana}{\includegraphics[width=0.32\textwidth]{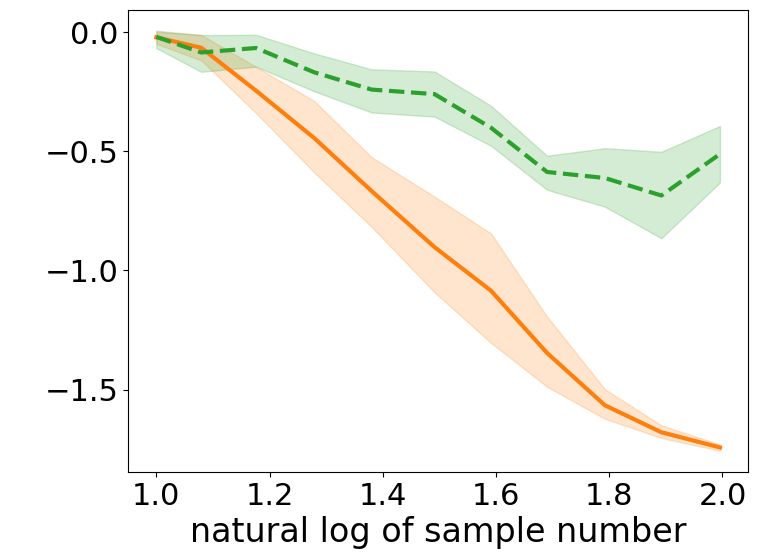}}
	\hfill
	\caption{The approximation error of BIS (solid line) and standard importance sampling based on the points of the randomized BO (dashed line). The experiment was repeated 10 times, where the bold line represents the averaged error and the band represents the 95\% confidence interval.} \label{fig:fig_s31_02}
\end{figure}

\subsection{Different Choices of Function $\phi$ in GP-UJB} \label{apx:choice_phi}

In \Cref{sec:setting}, \Cref{tab:cf} listed examples of the function $\phi$ used in GP-UJB.
This subsection compares the performance of BIS for different choices of the function $\phi$ through a simulation study.
We used the experimental setting in \Cref{sec:benchmark}, while changing the choice of the function $\phi$. 
\Cref{fig:fig_s32_01} visualizes 100 samples of BIS obtained for the banana density in \Cref{sec:benchmark}, under three different choices of the function $\phi$.
It demonstrates that, under the quadratic function $\phi(x) = x^2$ and the relu function $\phi(x) = \max(x, 0)$, a relatively large portion of samples were situated outside the high probability region of the target density.
On the other hand, the exponential function $\phi(x) = \exp(x)$ efficiently situated the majority of samples around the high probability region.

\Cref{fig:fig_s32_02} shows the approximation error of BIS for each benchmark density in \Cref{sec:benchmark}, under three different choices of the function $\phi$.
While all three choices of the function $\phi$ provide a reasonable approximation error, the default choice $\phi(x) = \exp(x)$ achieved a relatively better approximation error overall.
The approximation errors for the Gaussian and bimodal densities were similar, and those for the banana density exhibit a difference, as also exemplified in \Cref{fig:fig_s32_01}.
Given this simulation study, we recommend the exponential function $\phi(x) = \exp(x)$ as the default choice to use in GP-UJB.

\begin{figure}[t]
	\subcaptionbox{$\phi(x) = x^2$}{\includegraphics[width=0.31\textwidth]{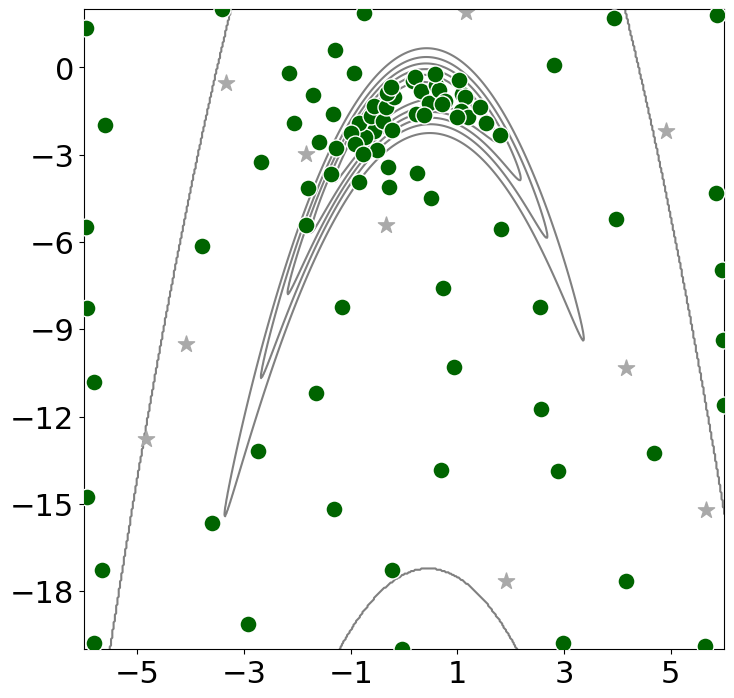}}
	\hfill
	\subcaptionbox{$\phi(x) = \max(x, 0)$}{\includegraphics[width=0.31\textwidth]{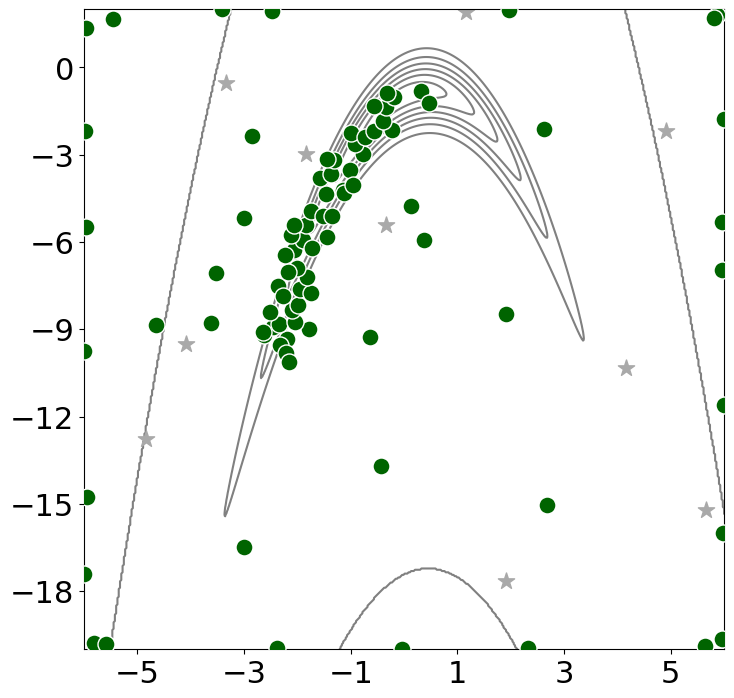}}
	\hfill
	\subcaptionbox{$\phi(x) = \exp(x)$}{\includegraphics[width=0.31\textwidth]{img/fig_41_03_sample_bis}}
	\hfill
	\caption{Visualization of 100 samples of the banana density obtained by BIS for different $\phi$. The contour values of the density were powered to $1 / 3$ for better visualization of the tail geometry.} \label{fig:fig_s32_01}
\end{figure}

\begin{figure}[t]
	\subcaptionbox{gaussian}{\includegraphics[width=0.32\textwidth]{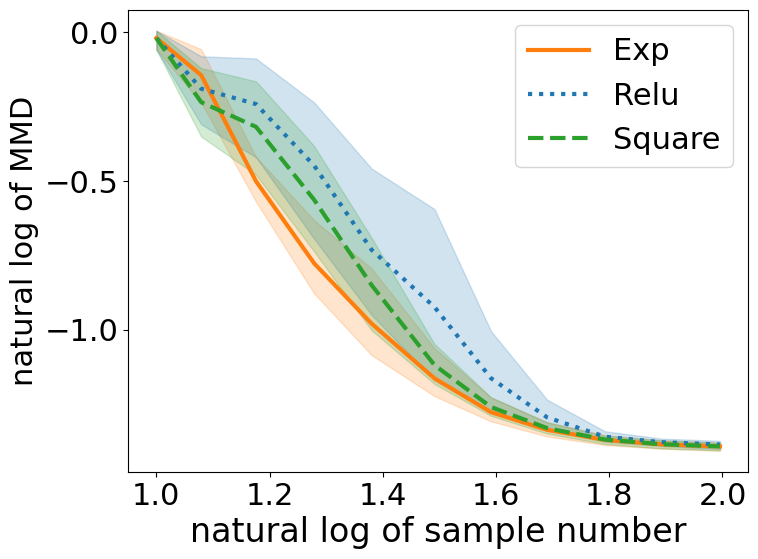}}
	\hfill
	\subcaptionbox{bimodal}{\includegraphics[width=0.32\textwidth]{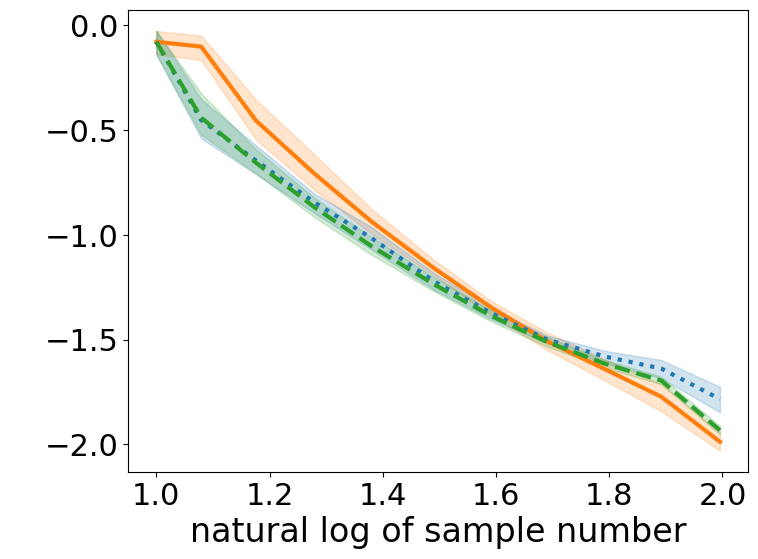}}
	\hfill
	\subcaptionbox{banana}{\includegraphics[width=0.32\textwidth]{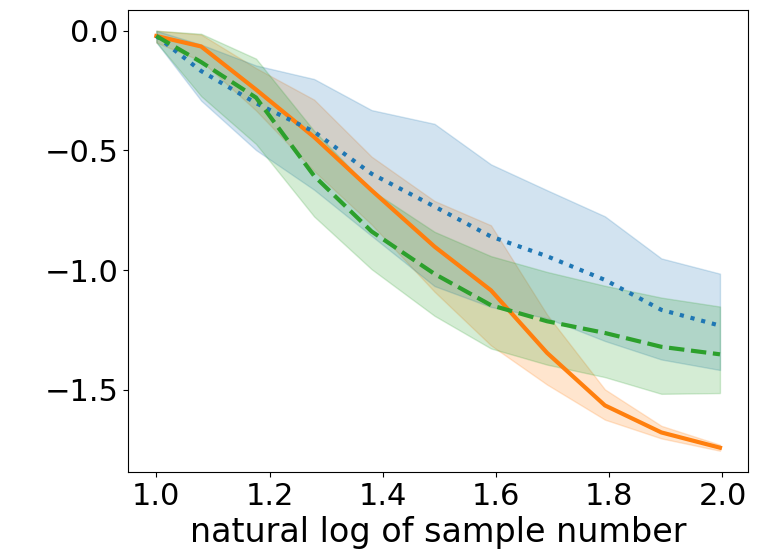}}
	\hfill
	\caption{The approximation error of BIS for each benchmark density when the function $\phi(x)$ in GP-UJB is quadratic $x^2$ (dashed line), relu $\max(x, 0)$ (dotted line), and exponential $\exp(x)$ (solid line). The experiment was repeated 10 times, where the bold line represents the averaged error and the band represents the 95\% confidence interval.} \label{fig:fig_s32_02}
\end{figure}

\subsection{Sensitivity of Proposal-Sequence Size $M$} \label{apx:pool_size}

The size of the proposal sequence $M$ is a hyperparameter of BIS.
The convergence rates in \Cref{sec:theory} imply the trade-off between small and large values of the size $M$. 
If the size $M$ is large, the first term in the convergence rates---regardless of the deterministic or probabilistic one---is suppressed, while the second term amplifies.
Conversely, if the size $M$ is small, the first term is not suppressed, while the second term does not amplify.
In the extreme case where $M$ is equal to one, BIS is identical to standard self-normalized importance sampling by construction.
Choosing a reasonable value of the size $M$ controls the trade-off, leading to an efficient approximation performance of BIS.

This subsection performs a simulation study on sensitivity of BIS to the size of the proposal sequence $M$.
We used the experimental setting in \Cref{sec:benchmark}, while changing the size of the proposal sequence $M$.
\Cref{fig:fig_s33_01} shows the approximation error of BIS for each benchmark density, under different values of the proposal-sequence size $M$.
As anticipated, when the size $M$ was the small value $2$, $8$, and $32$, the approximation performance of BIS got similar to that of standard importance sampling based on the scaled Halton sequence.
In addition, the excessively large value of the proposal-sequence size $M = 32768$ didn't result in the best approximation performance of BIS.
Overall, the size $M = 2048$ and $M = 8196$ led to the fastest decay of the approximation error of BIS.
As the smaller size $M$ reduces the computation cost in BIS, we recommend the size $M = 2048$ for experiments in two-dimensional domains based on this simulation study.

\begin{figure}[t]
	\subcaptionbox{gaussian}{\includegraphics[width=0.32\textwidth]{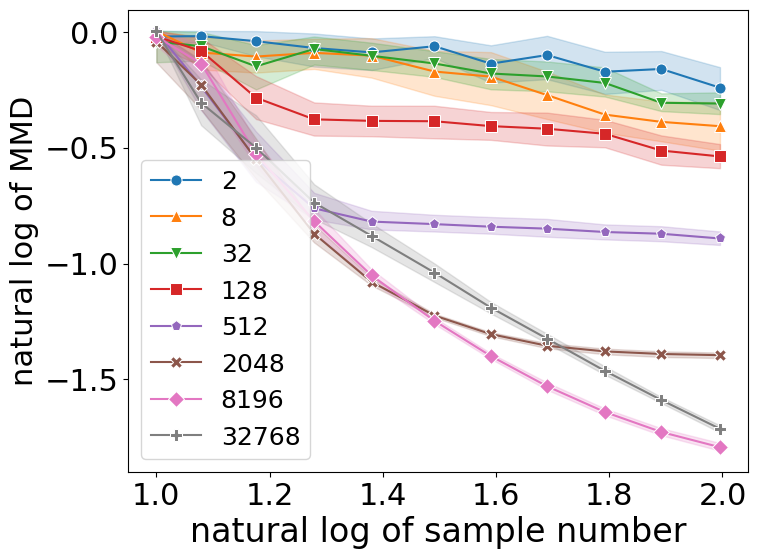}}
	\hfill
	\subcaptionbox{bimodal}{\includegraphics[width=0.32\textwidth]{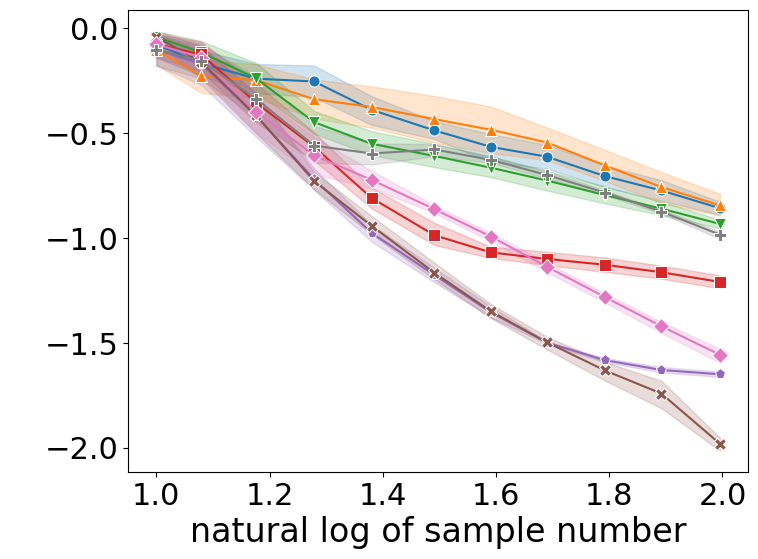}}
	\hfill
	\subcaptionbox{banana}{\includegraphics[width=0.32\textwidth]{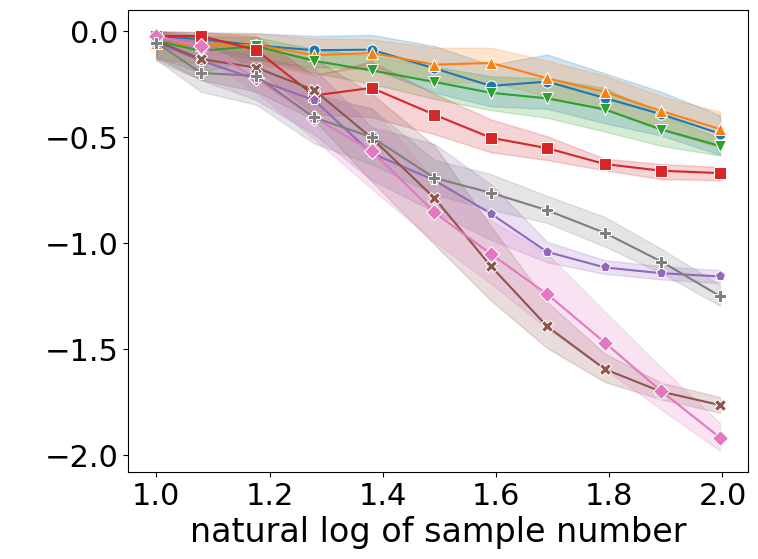}}
	\hfill
	\caption{The approximation error of BIS for each benchmark density when the size of the proposal sequence is $2$, $8$, $32$, $128$, $512$, $2048$, $8196$, and $32768$. The experiment was repeated 10 times, where the bold line represents the averaged error and the band represents the 95\% confidence interval.} \label{fig:fig_s33_01}
\end{figure}

\section{Additional Experimental Details} \label{apx:experiment}

This section provides additional details of the experiments presented in the main text.
\Cref{apx:experiment_1} contains that of \Cref{sec:benchmark} and \Cref{apx:experiment_2} contains that of \Cref{sec:weather}.

\subsection{Additional Details of \Cref{sec:benchmark}} \label{apx:experiment_1}

All the benchmark densities are given in the form proposed in \cite{Jarvenpaa2021} s.t.
\begin{align}
	p(\theta) \propto \exp\bigg( - \frac{1}{2}
	\begin{bmatrix}
		T_1(\theta) \\
		T_2(\theta)
	\end{bmatrix}^\mathrm{T}
	\begin{bmatrix}
		1 & \rho \\
		\rho & 1
	\end{bmatrix}
	\begin{bmatrix}
		T_1(\theta) \\
		T_2(\theta)
	\end{bmatrix}
	\bigg) .
\end{align}
Denote the first and second coordinates of the parameter $\theta$ by, respectively, $\theta_1$ and $\theta_2$.
The first unimodal Gaussian density is defined by $T_1(\theta) = \theta_1$, $T_2(\theta) = \theta_2$, and $\rho = 0.25$.
The second bimodal density is defined by $T_1(\theta) = \theta_1$, $T_2(\theta) = \theta_2^2 - 2$, and $\rho = 0.5$.
The third banana-shaped density is defined by $T_1(\theta) = \theta_1$, $T_2(\theta) = \theta_2 + \theta_1^2 + 1$, and $\rho = 0.9$.
We used the Gaussian kernel for the covariance kernel $k$ of the GP prior in GP-UJB.
The explicit form is given by
\begin{align}
	k(\theta, \theta') = \sigma^2 \exp\left( - \frac{\| \theta - \theta' \|^2}{2 l^2} \right)
\end{align}
where $l$ is the length-scale constant and $\sigma$ is the variance constant.

The approximation quality of the weighted samples of BIS was measured by MMD.
The form \eqref{eq:mmd} of MMD introduced in \Cref{sec:theory} admits a computationally-convenient alternative expression
\begin{align}
	\operatorname{MMD}(p, q)^2 = \E_{\theta, \theta' \sim p}[ \kappa(\theta, \theta') ] - 2 \E_{\theta \sim p, \theta' \sim q}[ \kappa(\theta, \theta') ] + \E_{\theta, \theta' \sim q}[ \kappa(\theta, \theta') ] .
\end{align}
This alternative form of MMD can be efficiently estimated by Monte Carlo integration using samples from the densities $p$ and $q$.
Throughout the experiments, we used MMD associated with the Gaussian kernel $\kappa(\theta, \theta') = \exp( - 0.5 \| \theta - \theta' \|^2 / h )$ with the constant $h = 0.1$.
The expectation in the MMD with respect to each target density was approximated using sufficiently many 100,000 weighted samples drawn from the scaled Halton sequence via standard importance sampling.

The approximation quality of the GP surrogate density was measured by TVD.
The explicit form of TVD between two densities $p$ and $q$ is given as follows:
\begin{align}
	\operatorname{TVD}(p, q) := \frac{1}{2} \int_\Theta \left| p(\theta) - q(\theta) \right| d\theta .
\end{align}
We approximate the integral over $\Theta$ by numerical integration.
In this experiment, the densities $p$ and $q$ to be plugged in TVD are available only up to the normalizing constants.
We also approximate their normalizing constants by numerical integration.
Let $\{ \eta_i \}_{i=1}^{N}$ be the $N$ points from the scaled Halton sequence on the domain $\Theta$.
Denote the volume of the domain $\Theta$ by $V$.
Let $\tilde{p}$ and $\tilde{q}$ be the proportional terms of the densities $p$ and $q$.
Then their normalizing constants, denoted by $Z_p$ and $Z_q$, are approximated as follows: 
\begin{align}
	Z_p \approx \frac{V}{N} \sum_{i=1}^{N} \tilde{p}(\eta_i) =: \hat{Z}_p \quad \text{and} \quad Z_q \approx \frac{V}{N} \sum_{i=1}^{N} \tilde{q}(\eta_i) =: \hat{Z}_q .
\end{align}
Finally, TVD between the densities $p$ and $q$ is numerically computed by 
\begin{align}
	\operatorname{TVD}(p, q) & \approx \frac{V}{2 N} \sum_{i=1}^{N} \left| \frac{ \tilde{p}(\eta_i) }{ \hat{Z}_p } - \frac{ \tilde{q}(\eta_i) }{ \hat{Z}_q } \right| .
\end{align}
We used the $N = 10,000$ points from the scaled Halton sequence throughout.

\Cref{tab:num} summarized the number of samples required for standard importance sampling to surpass the approximation error of 100 weighed samples obtained by BIS.
For standard importance sampling, we used the scaled Halton sequence as samples and assigned self-normalized importance weights to them.
It demonstrates that standard importance sampling required nearly 2,000 samples on average to achieve the approximation error that BIS does with 100 samples.

\begin{table}[h] 
	\caption{The number of samples required for standard importance sampling to surpass the approximation error of 100 BIS weighed samples. The experiment was repeated 10 times, where the averaged number was reported together with the standard deviation in a bracket.} \label{tab:num}
	\centering
	\begin{tabular}{ c c c }
		\hline
		target density & number of samples & approximation error \\
		\hline
		gaussian & 2368 ($\pm$ 232) & 0.040 ($\pm$ 0.002) \\
		bimodal & 1324 ($\pm$ 223) & 0.010 ($\pm$ 0.002) \\
		banana & 2487 ($\pm$ 293) & 0.018 ($\pm$ 0.001)  \\
		\hline
	\end{tabular}
\end{table}

\subsection{Additional Details of \Cref{sec:weather}} \label{apx:experiment_2}

The synthetic likelihood of the modified Lorentz weather model is computed using summary statistics of sample paths simulated from the model.
Denote by $s \in \R^d$ summary statistics of the observed sample path, for some dimension $d$.
Denote by $s(\theta) \in \R^d$ summary statistics of a sample path simulated from the model under a parameter $\theta$, where $s(\theta)$ is a random variable induced by simulation from the model.
Let $m(\theta) \in \R^d$ and $\Sigma(\theta) \in \R^{d \times d}$ be the mean and covariance of $s(\theta)$.
The synthetic likelihood at $\theta$ is then the Gaussian density $N(s \mid m(\theta), \Sigma(\theta))$ evaluated at the summary statistics $s$ of data.
In practice, the mean and covariance are estimated from multiple i.i.d.~realizations of the summary statistics $s(\theta)$.
Thus, the synthetic likelihood is computed using multiple sample paths simulated from the model at each $\theta$.
For the modified Lorentz weather model, \cite{Hakkarainen2012} used the following six summary statistics: (i) mean of $x_k(t)$ over time, (ii) variance of $x_k(t)$ over time, (iii) autocovariance of $x_k(t)$ with time lag one, (iv) covariance of $x_k(t)$ and $x_{k+1}(t)$ over time, (v) cross covariance of $x_k(t)$ and $x_{k-1}(t)$ with time lag one, (vi) cross covariance of $x_k(t)$ and $x_{k+1}(t)$ with time lag one.
The summary statistics computed for each variable $x_k(t)$ is averaged over $k$.
Thus, we have six-dimensional summary statistics for each sample path of the 40 variables.

\end{document}